\documentclass[11pt]{article}
\usepackage[utf8]{inputenc}
\usepackage{amsmath,amsfonts,amsthm,amssymb,url}
\usepackage{xcolor}
\usepackage[margin=1in]{geometry}
\usepackage{thmtools, thm-restate}
\usepackage{algorithm,algpseudocode}
\usepackage[T1]{fontenc}
\usepackage[colorlinks=true, linkcolor=red, urlcolor=blue, citecolor=gray]{hyperref}
\usepackage{mathtools}
\usepackage{graphicx}
\usepackage{caption}
\usepackage{subcaption}
\usepackage{tikz}
\usetikzlibrary{calc,arrows.meta,positioning}
\usetikzlibrary{decorations.pathreplacing}
\usetikzlibrary{arrows}

\hypersetup{linktocpage}
\definecolor{ForestGreen}{RGB}{34,139,34}

\newtheorem{theorem}{Theorem}

\newtheorem{fact}{Fact}
\newtheorem{lemma}{Lemma}
\newtheorem{remark}{Remark}

\newtheorem{corollary}{Corollary}
\newtheorem{definition}{Definition}[section]

\DeclareMathOperator*{\argmax}{arg\,max}

\let\Pr\relax
\DeclareMathOperator*{\Pr}{\mathbb{P}}
\DeclareMathOperator{\tr}{tr}

\DeclareMathOperator{\rank}{rank}

\newcommand{\R}{\mathbb{R}}

\newcommand{\bepsilon}{{\bar \epsilon}}

\newcommand{\poly}{\mathop\mathrm{poly}}

\newcommand{\eqdef}{\mathbin{\stackrel{\rm def}{=}}}
\newcommand{\norm}[1]{\|#1\|}

\newcommand{\bv}[1]{\mathbf{#1}}

\newcommand{\one}{\mathbf{1}} 
\newcommand{\Sbar}{\bar{S}}
\newcommand{\vhat}{\hat{\bv v}}

\newcommand{\Abar}{\bar{\bv A}}
\newcommand{\Ab}{\bv A}
\newcommand{\Sb}{\bv S}

\newcommand{\Zb}{\bv Z}
\newcommand{\Ib}{\bv I}

\newcommand{\Atilde}{\widetilde{\bv A}}

\newcommand{\Ztilde}{\widetilde{\bv Z}}

\newcommand{\sigmatilde}{\widetilde{\sigma}}
\newcommand{\vb}{\bv v}

\newcommand{\zb}{\bv z}
\newcommand{\xb}{\bv x}
\newcommand{\n}{{n \times n}}

\newcommand{\ith}{$i$\textsuperscript{th} }

\newcommand{\specialcell}[2][c]{%
  \begin{tabular}[#1]{@{}c@{}}#2\end{tabular}}

\usepackage{color}

\title{Universal Matrix Sparsifiers and Fast Deterministic Algorithms for  Linear Algebra}
  \author{
   Rajarshi Bhattacharjee\footnote{University of Massachusetts Amherst, \texttt{\{rbhattacharj,cmusco,ray\}@cs.umass.edu}} 
  \and 
   Gregory Dexter\footnote{Purdue University, \texttt{\{gdexter\}@purdue.edu}}
   \and 
   Cameron Musco\footnotemark[1]
   \and 
   Archan Ray\footnotemark[1]
   \and 
     Sushant Sachdeva \footnote{University of Toronto, \texttt{\{sachdeva\}@cs.toronto.edu}}
   \and
   David P Woodruff\footnote{Carnegie Mellon University, \texttt{\{dwoodruf\}@cs.cmu.edu}}
  }

  \usepackage{nth}
  \usepackage{intcalc}

\begin{document}
\sloppy

\date{}

\begin{titlepage}
\maketitle
\thispagestyle{empty}


\begin{abstract}

Let $\mathbf S \in \R^{n \times n}$ be any matrix satisfying $\|\mathbf{1}-\mathbf{S}\|_2 \le \epsilon n$, where $\mathbf{1}$ is the all ones matrix and $\|\cdot\|_2$ is the spectral norm. It is well-known that there exists $\mathbf S$ with just $O(n/\epsilon^2)$ non-zero entries achieving this guarantee: we can let $\mathbf{S}$ be the scaled adjacency matrix of a Ramanujan expander graph. We show that, beyond giving a sparse approximation to the all ones matrix, $\mathbf{S}$ yields a \emph{universal sparsifier} for any positive semidefinite (PSD) matrix. In particular, for any PSD $\mathbf A \in \mathbb{R}^{n \times n}$ which is normalized so that its entries are bounded in magnitude by $1$, we show that $\|\mathbf{A}-\mathbf{A} \circ \mathbf{S}\|_2 \le \epsilon n$, where $\circ$ denotes the entrywise (Hadamard) product. Our techniques also yield universal sparsifiers for non-PSD matrices. In this case, we show that if $\mathbf{S}$ satisfies $\|\mathbf{1}-\mathbf{S}\|_2 \le \frac{\epsilon^2 n}{c \log^2(1/\epsilon)}$ for some sufficiently large constant $c$, then $\|\mathbf A - \mathbf{A} \circ \mathbf{S} \|_2 \le \epsilon \cdot \max(n,\|\mathbf A\|_1)$, where $\|\mathbf A\|_1$ is the nuclear norm. Again letting $\mathbf{S}$ be a scaled Ramanujan graph adjacency matrix, this yields a sparsifier with $\widetilde O(n/\epsilon^4)$ entries. We prove that the above universal sparsification bounds for both PSD and non-PSD matrices are tight up to logarithmic factors. 

Since $\mathbf{A} \circ \mathbf{S}$ can be constructed \emph{deterministically} without reading all of $\mathbf{A}$, our result for PSD matrices derandomizes and improves upon established results for randomized matrix sparsification, which require sampling a random subset of ${O}(\frac{n \log n}{\epsilon^2})$ entries and only give an approximation to any fixed $\mathbf{A}$ with high probability. We further show that any randomized algorithm must read at least $\Omega(n/\epsilon^2)$ entries to spectrally approximate general $\Ab$ to error $\epsilon n$, thus proving that these existing randomized  algorithms are optimal up to logarithmic factors. We leverage our deterministic sparsification results to give the first {deterministic algorithms} for several  problems, including singular value and singular vector approximation and positive semidefiniteness testing, that run in faster than matrix multiplication time. This partially addresses a significant gap between randomized and deterministic algorithms for fast linear algebraic computation.

Finally, if $\mathbf A \in \{-1,0,1\}^{n \times n}$ is PSD, we show that a spectral approximation $\mathbf{\widetilde A}$ with $\|\mathbf A - \mathbf{\widetilde A}\|_2 \le \epsilon n$ can be obtained by deterministically reading $\widetilde O(n/\epsilon)$ entries of $\mathbf A$.  This improves the $1/\epsilon$ dependence on our result for general PSD matrices by a quadratic factor and is information-theoretically optimal up to a logarithmic factor.    

\end{abstract}

\end{titlepage}
\clearpage


\section{Introduction}
A common task in processing large matrices is  \emph{element-wise sparsification}. Given an $n \times n$ matrix $\bv A$,  the goal is to choose a small subset $S$ of coordinates in $[n] \times [n]$, where $[n] = \{1, 2, \ldots, n\}$, such that $\|\bv A - \bv A \circ \bv S\|_2$ is small, where $\circ$ denotes the Hadamard (entrywise) product and $\bv {S}$ is a sampling matrix which equals $\frac{n^2}{|S|}$ on the entries in $S$, but is $0$ otherwise. As in previous work, we consider operator norm error, where for a matrix $\bv B$, $\|\bv B\|_2 \eqdef \sup_{ \bv x \in \R^n} \frac{\|\bv B \bv{x}\|_2}{\| \bv x\|_2}$. Elementwise sparsification has been widely studied \cite{AM07,DZ11,AKL13,Braverman:2021wj} and has been used as a primitive in several applications, including low-rank approximation \cite{AM07,k15}, approximate eigenvector computation \cite{AHK06,AM07}, semi-definite programming \cite{AHK05,d2011subsampling}, and matrix completion \cite{CR12,CT10}. Without loss of generality, one can scale the entries of $\bv A$ so that the maximum entry is bounded by $1$ in absolute value, and we refer to such matrices as having {\it bounded entries}. With this normalization, it will be convenient to consider the task of finding a small subset $S$ with corresponding sampling matrix $\bv S$ such that, for a given error parameter $\epsilon \in (0,1)$,
\begin{eqnarray}\label{eqn:error}
\left\|\bv A - \bv A \circ \bv S\right\|_2 \leq \epsilon \cdot n. 
\end{eqnarray}

One can achieve the error guarantee in (\ref{eqn:error}) for any bounded entry matrix ${\bv A}$ with high probability by uniformly sampling a set of ${O}\left (\frac{n \log n}{\epsilon^2}\right)$ entries of ${\bv A}$. See, e.g., Theorem 1 of \cite{DZ11}.\footnote{To apply  Thm. 1 of \cite{DZ11} we first rescale ${\bv A}$ so that all entries are $\le 1/n$ in magnitude, and note that $\|{\bv A}\|_F^2 \le 1$ in this case. Thus, they sample $s = O \left (\frac{n \log n}{ \epsilon^2}\right)$ entries. Rescaling their error guarantee analogously gives \eqref{eqn:error}.}  
 However, there are no known lower bounds for this problem, even if we consider the harder task of \emph{universal sparsification}, which requires finding a fixed subset $S$ such that \emph{\eqref{eqn:error} holds simultaneously for every bounded entry matrix ${\bv A}$}. The existence of such a fixed subset $S$ corresponds to the existence of a \emph{deterministic sublinear query algorithm} that constructs a spectral approximation to any $\bv A$ by forming $\bv{A} \circ \bv S$ (which requires reading just $|S|$ entries of $\bv A$). As we will see, such algorithms have  applications to fast deterministic algorithms for  linear algebraic computation. 
 We ask:
\begin{center}
{\it What is the size of the smallest set $S$ that achieves \eqref{eqn:error} simultaneously for every bounded entry matrix ${\bv A}$?}
\end{center}
Previously, no bound on the size of $S$ better than the trivial $O(n^2)$  was known.

\subsection{Our Results}
Our first result answers the above question for the class of \emph{symmetric positive semidefinite} (PSD) matrices, i.e., those matrices ${\bf A}$ for which $\bv x^T {\bv A} \bv x \geq 0$ for all vectors $\bv x \in \R^n$, or equivalently, whose eigenvalues are all non-negative. PSD matrices arise e.g., as covariance matrices, kernel similarity matrices, and graph Laplacians, and  significant work has considered efficient algorithms for approximating their properties  \cite{wang2016spsd, clarkson2017low, xia2010robust, chatelin2011spectral, andoni2016sketching, meyer2021hutch++, schafer2021sparse}.
We summarize our results below and give a comparison to previous work in Table \ref{tab:table}. 

We  show that there exists a set $S$ with $|S| = {O}(n/\epsilon^2)$ that achieves \eqref{eqn:error} simultaneously for all bounded entry PSD matrices. This improves the best known randomized bound of ${O}\left (\frac{n \log n}{\epsilon^2}\right )$ for algorithms which only succeed on a \emph{fixed matrix} with high probability.  

\begin{restatable}[Universal Sparsifiers for PSD Matrices]{theorem}{quadraticformnonc}
\label{thm:PSD_Quad}
There exists a subset $S$ of $s = {O}(n/\epsilon^2)$ entries of $[n] \times [n]$ such that, letting $\bv{S} \in \R^{n \times n}$ have $\bv S_{ij} = \frac{n^2}{s}$ for $(i,j) \in S$ and $\bv{S}_{ij} = 0$ otherwise, simultaneously for all PSD matrices $\bv{A} \in \mathbb{R}^{n \times n}$ with bounded entries, 
$\left\|\bv A - \bv A \circ \bv S \right\|_2 \leq \epsilon n.$
\end{restatable}

\begin{table}
\begin{center}
\small
\begin{tabular}{|c|c|c|c|c|c|}
\hline
\specialcell{\textbf{Matrix}\\ \textbf{Type}} & \specialcell{\textbf{Approx.}\\ \textbf{Type}} &  \specialcell{\textbf{Randomized} \\\textbf{Error}}& \specialcell{\textbf{Randomized} \\\textbf{Sample Complexity}} & \specialcell{\textbf{Deterministic} \\\textbf{Error}} &  \specialcell{\textbf{Deterministic} \\\textbf{Sample Complexity}}\\
\hline
\specialcell{$\norm{\bv A}_\infty \le 1$ \\ $\bv A$ is PSD} & sparse & $\epsilon n$ &  $ \Theta(n/\epsilon^2)$ (Thms \ref{thm:PSD_Quad}, \ref{theorem:psd_deterministic_lower}) & $\epsilon n$ & \specialcell{$ \Theta(n/\epsilon^2)$ (Thms \ref{thm:PSD_Quad}, \ref{theorem:psd_deterministic_lower})}\\
\hline
$\norm{\bv A}_\infty \le 1$ & sparse &$\epsilon n$ & \specialcell{$O\left (\frac{n \log n}{\epsilon^2}\right )$ \cite{AM07} \smallskip \\ $\Omega(n/\epsilon^2)$ (Thm \ref{theorem:psd_deterministic_lower})}& $\epsilon \max(n,\norm{\bv A}_1)$ & \specialcell{$O\left (\frac{n \log^4(1/\epsilon)}{\epsilon^4}\right )$ (Thm \ref{thm: general matrix eps4 bound}) \smallskip \\ $\Omega(n/\epsilon^4)$ (Thm \ref{thm: lower bound nonadaptive})}\\
\hline
\hline
\specialcell{$\norm{\bv A}_\infty \le 1$ \\ $\bv A$ is PSD} & any & $\epsilon n$ & \specialcell{$O \left (\frac{n \log(1/\epsilon)}{\epsilon}\right )$ \cite{musco2017recursive}} & $\epsilon n$ & \specialcell{$O\left (\frac{n \log n}{\epsilon}\right )$ (Thm \ref{th:binary query complexity}) \\ for $\bv{A} \in \{-1,0,1\}^{n \times n}$  \smallskip \\ $\Omega(n/\epsilon)$ (Thm \ref{thm: lower bound for PSD matrices})}\\ 
\hline
$\norm{\bv A}_\infty \le 1$ & any & $\epsilon n$ & \specialcell{$\Omega(n/\epsilon^2)$ (Thm \ref{theorem:adaptive_random_lower_bound})} & $\epsilon \max(n,\norm{\bv A}_1)$ & $\Omega(n/\epsilon^2)$ (Thm \ref{thm: lower bound a1}) \\
\hline
\end{tabular}
\end{center}
\vspace{-1em}
\caption{Summary of results. Algorithms in the first two rows output a spectral approximation that is sparse, as in \eqref{eqn:error}. Our deterministic algorithms in this case follow directly from our universal sparsification results. Algorithms in the last two rows can output an approximation of any form. 
}
\vspace{-1em}
\end{table}\label{tab:table}

Theorem \ref{thm:PSD_Quad} can be viewed as a significant strengthening of a classic spectral graph expander guarantee \cite{margulis1973explicit,alon1986eigenvalues}; indeed, letting ${\bv A} = \bv 1$ be the all ones matrix, we have that if ${\bv A} \circ \bv S$ satisfies (\ref{eqn:error}), then it matches the near-optimal spectral expansion of Ramanujan graphs, up to a constant factor.\footnote{If we let $d = s/n = O(1/\epsilon^2)$ be the average number of entries per row of $\bv S$, and let $\bv{G} \in \{0,1\}^{n \times n}$ be the binary adjacency matrix of the graph with edges in $S$ (i.e., $\bv G= \frac{s}{n^2} \cdot \bv{S}$), then by \eqref{eqn:error} applied with $\bv{A} = \bv{1}$, $\lambda_1(\bv G) = \frac{s}{n^2} \lambda_1( \bv S) =\frac{s}{n^2} \lambda_1( \bv 1 \circ \bv S) = \frac{s}{n^2} (\lambda_1(\bv{1}) \pm \epsilon n)  = \Theta(s/n) = \Theta(d)$, while for $i > 1$, $|\lambda_i(\bv G)| \le  \frac{s}{n^2} (|\lambda_i(\bv{1})| + \epsilon n) \le \epsilon d = \Theta(\sqrt{d})$.} In fact, we prove Theorem \ref{thm:PSD_Quad} by proving a more general claim: that any matrix $\bv S$ that sparsifies the all ones matrix also sparsifies every bounded entry PSD matrix. In particular, Theorem \ref{thm:PSD_Quad} follows as a direct corollary of:
\begin{restatable}[Spectral Expanders are Universal Sparsifiers for PSD Matrices]{theorem}{thmReduction}
\label{thm:reduction}
Let $\bv{1} \in \R^{n \times n}$ be the all ones matrix. Let $\bv{S} \in \R^{n \times n}$ be any matrix such that $\norm{\bv 1 - \bv S}_2 \le \epsilon n$. Then for any PSD matrix $\bv{A} \in \R^{n \times n}$ with bounded entries, $\norm{\bv{A}-\bv{A} \circ \bv S}_2 \le \epsilon n$.
\end{restatable}
To prove Theorem \ref{thm:PSD_Quad} given Theorem \ref{thm:reduction}, we let $S$ be the edge set of a Ramanujan spectral expander graph with degree $d = O(1/\epsilon^2)$ and adjacency matrix $\bv G$. We have $s = nd = O(n/\epsilon^2)$ and  $\bv{S} = \frac{n^2}{s} \bv {G} = \frac{n}{d} \bv{G}$. Thus, the top eigenvector of $\bv{S}$ is the all ones vector with eigenvalue $\lambda_1(\bv{S}) = \frac{n}{d} \lambda_1(\bv{G}) = n$. All other eigenvalues are bounded by $|\lambda_i(\bv S)| = \frac{n}{d} |\lambda_i(\bv{G})| = O(\frac{n}{d} \cdot \sqrt{d}) = O(\epsilon n).$ Combined, after adjusting $\epsilon$ by a constant factor, this shows that $\norm{\bv 1 - \bv S}_2 \le \epsilon n$, as required.
\medskip

\noindent\textbf{Sparsity Lower Bound.} We show that Theorem \ref{thm:PSD_Quad} is tight. Even if we only seek to approximate the all ones matrix (rather than all bounded PSD matrices), $\bv{S}$ requires  $\Omega(n/\epsilon^2)$ non-zero entries. 
\begin{restatable}[Sparsity Lower Bound -- PSD Matrices]{theorem}{alloneslower}
\label{theorem:psd_deterministic_lower}
    Let ${\bv 1} \in \mathbb{R}^{n \times n}$ be the all-ones matrix.  Then, for any $\epsilon \in (0,1/2)$ with $\epsilon \ge c/\sqrt{n}$ for large enough constant $c$, any $\bv{S} \in \R^{n \times n}$ with $\norm{\bv 1 - \bv S}_2 \le \epsilon n$ must have $\Omega(n/\epsilon^2)$ nonzero entries.
\end{restatable}

Theorem \ref{theorem:psd_deterministic_lower} resolves an open question of \cite{Braverman:2021wj}, which asked whether a spectral approximation must have $\Omega\left (\frac{\operatorname{ns}(\bv A) \operatorname{sr}(\bv A)}{\epsilon^2}\right )$ non-zeros, where $\operatorname{sr}(\bv A) = \frac{\norm{\bv{A}}_F^2}{\norm{\bv A}_2^2}$ is the stable rank and $\operatorname{ns}(\bv A) = \max_i \frac{\norm{\bv a_i}_1^2}{\norm{\bv a_i}_2^2}$,  where $\bv{a}_i$ is the $i^{th}$ row of $\bv A$, is the  `numerical sparsity'. They give a lower bound when $ \operatorname{sr}(\bv A) = \Theta(n)$ but ask if this can be extended to $ \operatorname{sr}(\bv A) = o(n)$. For the all ones matrix, $\operatorname{sr}(\bv A) = 1$ and $\operatorname{ns}(\bv A) = n$, and so Theorem \ref{theorem:psd_deterministic_lower}  resolves this question. Further, by applying the theorem to  a block diagonal matrix with $r$ disjoint $k \times k$ blocks of all ones,   we resolve the question for integer $\operatorname{sr}(\bv A)$ and $\operatorname{ns}(\bv A)$.

\medskip

\noindent\textbf{Non-PSD Matrices.} 
The techniques used to prove Theorem \ref{thm:PSD_Quad} also give nearly tight universal sparsification bounds for general bounded entry (not necessarily PSD) matrices. In this case, we show that a  subset $S$ of  ${O}\left (\frac{n \log^4(1/\epsilon)}{\epsilon^4}\right )$ entries suffices to achieve spectral norm error depending on the nuclear norm $\|{\bv A}\|_1$, which is the sum of singular values of ${\bv A}$. 
\begin{restatable}[Universal Sparsifiers for Non-PSD Matrices]{theorem}{generalmatrixupperbound}
\label{thm: general matrix eps4 bound}
There exists a subset $S$ of $s = {O}\left (\frac{n\log^4(1/\epsilon)}{\epsilon^4}\right )$ entries of $[n] \times [n]$ such that, letting $\bv{S} \in \R^{n \times n}$ have $\bv S_{ij} = \frac{n^2}{s}$ for $(i,j) \in S$ and $\bv{S}_{ij} = 0$ otherwise, simultaneously for all symmetric matrices $\bv{A} \in \mathbb{R}^{n \times n}$ with bounded entries,
\begin{align}\label{eqn:error2}
\left\|\bv A - \bv A \circ \bv S \right\|_2 \leq \epsilon \cdot \max(n, \|\bv A\|_1).
\end{align}
\end{restatable}
Note that for a bounded entry PSD matrix ${\bv A}$, we have $\|{\bv A}\|_1 = \tr({\bv A}) = O(n)$ and so for PSD matrices, (\ref{eqn:error}) and (\ref{eqn:error2}) are equivalent. \begin{remark}\label{rem:asym}Although Theorem \ref{thm: general matrix eps4 bound} is stated for symmetric ${\bv A}$, one can first symmetrize ${\bv A}$ by  considering ${\bv B} = [{\bv 0}, {\bv A}; {\bv A}^T, {\bv 0}]$. Then for any ${\bv z} = ({\bv x}, {\bv y}) \in \mathbb{R}^{2n}$, we have ${\bv z}^T {\bv B} {\bv z} = 2{\bv x}^T {\bv A} {\bv y}$, and $\|{\bv B}\|_1 = 2 \|{\bv A}\|_1$, so applying the above theorem to ${\bv B}$ gives us a  spectral approximation to ${\bv A}$.
\end{remark}
As with Theorem \ref{thm:PSD_Quad}, Theorem \ref{thm: general matrix eps4 bound} follows from a general claim which shows that sparsifying the all ones matrix to small enough error suffices to sparsify any bounded entry symmetric matrix.
\begin{restatable}[Spectral Expanders are Universal Sparsifiers for Non-PSD Matrices]{theorem}{reductionNonPSD}
\label{thm:reductionNonPSD}
Let $\bv{1} \in \R^{n \times n}$ be the all ones matrix. Let $\bv{S} \in \R^{n \times n}$ be any matrix such that $\norm{\bv 1 - \bv S}_2 \le \frac{\epsilon^2 n}{c \log^2(1/\epsilon)}$ for some  large enough constant $c$. Then for any symmetric $\bv{A} \in \R^{n \times n}$ with bounded entries, $\norm{\bv{A}-\bv{A} \circ \bv S}_2 \le  \epsilon \cdot \max(n, \|\bv A\|_1).$
\end{restatable} 
Theorem \ref{thm: general matrix eps4 bound} follows by applying Theorem \ref{thm:reductionNonPSD} where $\bv{S}$ is taken to be the scaled adjacency matrix of a Ramanujan expander graph with degree $d = O(\log^4(1/\epsilon)/\epsilon^4)$. 

\medskip

\noindent\textbf{Lower Bounds for Non-PSD Matrices.}
We can prove that Theorem \ref{thm: general matrix eps4 bound} is tight up to logarithmic factors. Our lower bound holds even for the easier problem of top singular value (spectral norm) approximation and against a more general class of algorithms, which non-adaptively and deterministically query entries of the input matrix. The idea  is simple: since the entries read by the deterministic algorithm are fixed, we can construct two very different input instances on which the algorithm behaves identically: one which is the all ones matrix, and the other which is one only on the entries read by the algorithm and zero everywhere else.  We show that if the number of entries $s$ that are read is too small, the top singular value of the second instance is significantly smaller than the first (as compared to their Schatten-1 norms), which violates the desired error bound of $\epsilon \cdot \max(n,\norm{\bv A}_1)$. To obtain a tight bound, we apply the above construction on just a subset of the rows of the matrix on which the algorithm does not read too many entries. Here, we critically use non-adaptivity, as this subset of rows can be fixed, independent of the input instance.
\begin{restatable}[Non-Adaptive Query Lower Bound for Deterministic Spectral Approximation of 
Non-PSD Matrices]{theorem}{lowerboundNonAdaptive}
\label{thm: lower bound nonadaptive}
For any $\epsilon\in (1/n^{1/4},1/4)$, any deterministic algorithm that queries entries of  a bounded entry matrix $\bv A$ non-adaptively and outputs $\widetilde \sigma_1$ satisfying $|\sigma_1(\bv A)-\widetilde \sigma_1| \le \epsilon \cdot \max(n, \|\bv A\|_1)$  must read at least $\Omega\left(\frac{n}{\epsilon^4}\right)$ entries. Here $\sigma_1(\bv A) = \norm{\bv A}_2$ is the largest singular value of $\bv A$.
\end{restatable}
Observe that Theorem \ref{thm: general matrix eps4 bound} implies the existence of a non-adaptive deterministic algorithm for the above problem with query complexity $O \left (\frac{n \log^4(1/\epsilon)}{\epsilon^4}\right )$, since  $\bv{A} \circ \bv S$ can be computed with non-adaptive determinstic queries and since $\norm{\bv A - \bv A \circ \bv S}_2 \le \epsilon \cdot \max(n,\norm{\bv A}_1)$ implies via Weyl's inequality that $|\sigma_1(\bv{A})-\sigma_1(\bv A \circ \bv S)| \le \epsilon \cdot \max(n,\norm{\bv A}_1)$. Also recall that while Theorem \ref{thm: general matrix eps4 bound} is stated for bounded entry symmetric matrices, it applies to general (possibly asymmetric) matrices by Remark \ref{rem:asym}.

Note that Theorem \ref{thm: lower bound nonadaptive} establishes a separation between universal sparsifiers and  randomized sparsification since randomly sampling $  O\left (\frac{n\log n }{\epsilon^2}\right )$ entries of any $\bv{A} \in \R^{n \times n}$ achieves error $\epsilon n$  by \cite{DZ11}. In fact, for the problem of just approximating $\sigma_1(\bv{A})$ to error $\pm \epsilon n$, randomized algorithms using just $\poly(\log n,1/\epsilon)$ samples  are known \cite{Bakshi:2020uz,Bhattacharjee:2021wl}. Theorem \ref{thm: lower bound nonadaptive} shows that universal sparsifiers for general bounded entry  matrices (and in fact all non-adaptive deterministic algorithms for spectral approximation)  require a worse $1/\epsilon^4$ dependence and error bound of $\epsilon \cdot \max(n,\norm{\bv A}_1)$. 
We can extend Theorem \ref{thm: lower bound nonadaptive} to apply to general deterministic algorithms that query $\bv{A}$ possibly \emph{adaptively}, however the lower bound weakens to $\Omega(n/\epsilon^2)$. Understanding if this gap between adaptive and non-adaptive query deterministic algorithms is real  is an interesting open question. 

\begin{restatable}[Adaptive Query Lower Bound for Deterministic Spectral Approximation of Non-PSD Matrices]{theorem}{lowerbounda}
\label{thm: lower bound a1}
For any $\epsilon\in (1/\sqrt{n},1/4)$, any deterministic algorithm that queries entries of  a bounded entry matrix $\bv A$ (possibly adaptively) and outputs $\widetilde \sigma_1$ satisfying $|\sigma_1(\bv A)-\widetilde \sigma_1| \le \epsilon \cdot \max(n, \|\bv A\|_1)$  must read at least $\Omega\left(\frac{n}{\epsilon^2}\right)$ entries. Here $\sigma_1(\bv A) = \norm{\bv A}_2$ is the largest singular value of $\bv A$.
\end{restatable}

\subsubsection{Applications to Fast Deterministic Algorithms for Linear Algebra}\label{sec:det_intro} Given sampling matrix $\bv S$ such that $\norm{\bv A - \bv A \circ \bv S}_2 \le \epsilon \cdot \max(n,\norm{\bv A}_1)$, one can use $\bv{A} \circ \bv S$ to approximate various linear algebraic properties of $\bv A$. For example, by Weyl's inequality \cite{weyl1912asymptotic}, the eigenvalues of $\bv{A} \circ \bv S$ approximate those of $\bv{A}$ up to additive error $\pm \epsilon \cdot \max(n,\norm{\bv A}_1)$. Thus, our universal sparsification results (Theorems \ref{thm:PSD_Quad}, \ref{thm: general matrix eps4 bound}) immediately give the first known deterministic algorithms for approximating the eigenspectrum of a bounded entry matrix up to small additive error with \emph{sublinear entrywise query complexity}. Previously, only randomized sublinear query algorithms were known \cite{williams2000using,musco2017recursive,MW17,tropp2017randomized,Bakshi:2020uz,Bhattacharjee:2021wl}.

Further, our results yield the first deterministic algorithms for several  problems that run in $o(n^\omega)$  time, where $\omega \approx 2.373$ is the matrix multiplication exponent  \cite{Alman:2021uk}. 
Consider e.g., approximating  the top singular value $\sigma_1(\bv A)$ (the spectral norm) of $\bv A$. A $(1+\epsilon)$-relative error approximation can be computed in $\widetilde O(n^2/\sqrt{\epsilon})$ time with high probability using $O(\log(n/\epsilon)/\sqrt{\epsilon})$ iterations of the Lanczos  method with a random initialization \cite{Kuczynski:1992va,musco2015randomized}. This can be further accelerated e.g., via randomized entrywise sparsification \cite{DZ11}, allowing an additive $\pm \epsilon n$ approximation to $\sigma_1(\bv A)$ to be computed in $\widetilde O(n/\poly(\epsilon))$ time. However, prior to our work, no fast deterministic algorithms were known, even with  coarse approximation guarantees. The fastest approach was just to perform a full SVD of $\bv A$, requiring $\Theta(n^\omega)$ time. In fact, this gap between randomized and deterministic methods exists for many linear algebraic problems, and resolving it is a central open question.

By combining our universal sparsification results with a derandomized power method that uses a full subspace of vectors for initalization, and iteratively approximates the singular values of $\bv A \circ \bv S$ via `deflation'  \cite{Saad:2011wv}, we give to the best of our knowledge, the first $o(n^\omega)$ time deterministic algorithm for approximating \emph{all singular values} of a bounded entry matrix $\bv{A}$ to small additive error.

\begin{restatable}[Deterministic Singular Value Approximation]{theorem}{deterministicsingval}\label{thm:approxSVD} 
Let $\bv{A} \in \mathbb{R}^{n \times n}$ be a bounded entry symmetric matrix  with singular values $\sigma_1(\bv{A}) \geq \sigma_2(\bv{A}) \geq \ldots \geq \sigma_n(\bv{A})$. Then there exists a deterministic algorithm that, given $\epsilon \in (0,1)$, reads $\widetilde{O}\big(\frac{n}{\epsilon^4} \big)$ entries of $\bv{A}$, runs in $\widetilde{O}\big(\frac{n^2}{\epsilon^8} \big)$ time, and returns singular value approximations $\sigmatilde_1(\bv{A}),\ldots, \sigmatilde_n(\bv A)$ satisfying for all $i$, $$|\sigma_i(\bv{A})-\sigmatilde_i(\bv{A})| \leq \epsilon \cdot \max(n, \|\bv{A} \|_1).$$
Further, for all $i \le 1/\epsilon$, the algorithm returns a unit vector $\bv{z}_i$ such that $|\norm{\bv A \bv z_i}_2 - \sigma_i(\bv A)| \le \epsilon \cdot \max(n,\norm{\bv A}_1)$ and the returned vectors are all mutually orthogonal.
\end{restatable}
The runtime of Theorem \ref{thm:approxSVD} is stated assuming we have already constructed a sampling matrix $\bv S$ that samples $\tilde O(\frac{n}{\epsilon^4})$ entries of $\bv A$ and satisfies Theorem \ref{thm: general matrix eps4 bound}. It is well known that, if we let $\bv{S}$ be the adjacency matrix of a Ramanujan expander graph, it can indeed be constructed deterministically in $\Tilde{O}(n/\poly(\epsilon))$ time~\cite{alon2021explicit}. This is lower order as compared to the runtime of $\widetilde{O}\big(\frac{n^2}{\epsilon^8} \big)$ unless $\epsilon < 1/n^c$ for some large enough constant $c$. Further, for a fixed input size $n$ (and in fact, for a range of input sizes), $\bv{S}$ needs to be constructed only once -- see Section \ref{section:notation} for further discussion.

Recall that if we further assume $\bv A$ to be PSD, then the error in Theorem \ref{thm:approxSVD} is bounded by $\epsilon \cdot \max(n,\norm{\bv A}_1) \le \epsilon n$. The sample complexity and runtime also improve by $\poly(1/\epsilon)$ factors due to the tighter universal sparisfier bound for PSD matrices given in Theorem \ref{thm:PSD_Quad}. Also observe that while Theorem \ref{thm:approxSVD} gives additive error approximations to all of $\bv{A}$'s singular values, these approximations are only meaningful for singular values larger than $\epsilon \cdot \max(n,\norm{\bv A}_1)$, of which there are at most $1/\epsilon$. Similar additive error guarantees have been given using randomized algorithms \cite{Bhattacharjee:2021wl,Woodruff:2023ul}. Related bounds have also been studied in work on randomized methods for spectral density estimation \cite{weisse2006kernel, lin2016approximating, braverman2022sublinear}.

We further leverage Theorem \ref{thm:approxSVD} to give the first $o(n^\omega)$ time deterministic algorithm for testing if a bounded entry  matrix is either PSD or has at least one large negative eigenvalue  $\le -\epsilon \cdot \max( n,\norm{\bv A}_1)$. Recent work has focused on  optimal randomized methods for this problem \cite{Bakshi:2020uz,nsw22}. We also show that, under the assumption that $\sigma_1(\bv A) \ge \alpha \cdot \max(n,\norm{\bv A}_1)$ for some $\alpha \in (0,1)$, one can deterministically compute $\widetilde \sigma_1(\bv{A})$ with $|\sigma_1(\bv{A})-\widetilde \sigma_1(\bv{A})| \leq \epsilon \cdot \max(n, \|\bv{A} \|_1)$ in $\widetilde O \left (\frac{n^2  \log(1/\epsilon)}{\poly(\alpha)} \right )$ time. That is, one can compute a highly accurate approximation to the top singular value in roughly linear time in the input matrix size. Again, this is the first $o(n^\omega)$ time deterministic algorithm for this problem, and matches the runtime of the best known randomized methods for high accuracy top singular value computation, up to a $\poly(\log(n,1/\alpha))$ factor.

\subsubsection{Beyond Sparse Approximations}
It is natural ask if it is possible to achieve better than $O(n/\epsilon^2)$ sample complexity for spectral approximation by using an algorithm that does not output a sparse approximation to $\bv{A}$, but can output more general data structures, allowing it to avoid the sparsity lower bound of Theorem \ref{theorem:psd_deterministic_lower}.  Theorems \ref{thm: lower bound nonadaptive} and \ref{thm: lower bound a1} already rule this out for both non-adaptive and adaptive deterministic algorithms for non-PSD matrices.
However, it is known that this \emph{is possible} with randomized algorithms for  PSD matrices. For example, following \cite{Avron:2017vt}, one can apply Theorem 3 of \cite{musco2017recursive} with error parameter $\lambda = \epsilon n$. Observing that all ridge leverage score sampling probabilities are bounded by $1/\lambda$ (e.g., via Lemma 6 of the same paper and the bounded entry assumption), one can show that a Nystr\"{o}m approximation $\bv{\widetilde A}$ based on $\widetilde O(1/\epsilon)$ uniformly sampled columns satisfies $\norm{\bv A - \bv{\widetilde A}}_2 \le \epsilon n$ with high probability. Further $\bv{\widetilde A}$ can be constructed by reading just $\widetilde O(1/\epsilon)$ columns and thus $\widetilde O(n/\epsilon)$ entries of $\bv A$, giving a linear rather than quadratic dependence on $1/\epsilon$ as compared to Theorem \ref{thm:PSD_Quad}. Unfortunately, derandomizing such a column-sampling-based approach seems difficult -- any deterministic algorithm must read entries in $\Omega(n)$ columns of $\bv{A}$, as otherwise it will fail when $\bv{A}$ is entirely supported on the unread columns.

Nevertheless, in the special case where $\bv{A}$ is PSD and has entries  in $\{-1,0,1\}^{n \times n}$, we show that a spectral approximation can be obtained by deterministically reading just $\widetilde O(n/\epsilon)$ entries of $\bv A$. 

\begin{restatable}[Deterministic Spectral Approximation of Binary Magnitude PSD Matrices]{theorem}{binaryapprox}
\label{th:binary query complexity}
Let $\bv A \in \{-1, 0,1\}^{n \times n}$ be PSD. Then for any $\epsilon \in (0,1)$,  there exists a deterministic algorithm which reads $O\left(\frac{n\log n}{\epsilon}\right)$ entries of $\bv A$ and returns PSD $\bv{\widetilde A} \in \{-1,0,1\}^{n \times n}$ such that $\norm{\bv A - \bv{\widetilde A}}_2 \le \epsilon n$.
\end{restatable}

Using Tur\'{a}n's theorem, we  show that Theorem \ref{th:binary query complexity} is information theoretically optimal up to a $\log n$ factor, even for the potentially much easier problem of eigenvalue approximation:

\begin{restatable}[Deterministic Spectral Approximation of Binary PSD Matrices -- Lower Bound]{theorem}{lowerboundPSD}
\label{thm: lower bound for PSD matrices}
Let $\bv A \in \{0,1\}^{n\times n}$ be PSD. Then for any $\epsilon \in (0,1)$,  any possibly adaptive deterministic algorithm which approximates all eigenvalues of $\bv A$ up to $\epsilon n$ additive error must read $\Omega\left(\frac{n}{\epsilon}\right)$ entries of $\bv A$.
\end{restatable} 
An interesting open question is if $\widetilde O(n/\epsilon)$ sample complexity can be achieved for deterministic spectral approximation of \emph{any bounded entry PSD matrix}, with a matrix $\bv{\widetilde A}$ of any form, matching what is known for randomized algorithms.

Finally, we show that the PSD assumption is critical to achieve $o(n/\epsilon^2)$ query complexity, for both randomized and deterministic algorithms. Without this assumption,  $\Omega(n/\epsilon^2)$ queries are required, even to achieve (\ref{eqn:error}) with constant probability for a single input ${\bv A}$. For bounded entry matrices, the randomized element-wise sparsification algorithms of \cite{AM07, DZ11} read just $\widetilde O(n/\epsilon^2)$ entries of $\bv A$, and so our lower bounds are the first to show near-optimality of the number of entries read of these algorithms, which may be of independent interest. 

\begin{restatable}[Lower Bound for Randomized Spectral Approximation]{theorem}{adaptiverandomlowerbound}
\label{theorem:adaptive_random_lower_bound}
Any randomized algorithm that (possibly adaptively) reads entries of a binary matrix $\bv A \in \{0, 1\}^{n \times n}$ to construct a data structure $f:\mathbb{R}^n \times \mathbb{R}^n \rightarrow \mathbb{R}$ that, for $\epsilon \in (0,1)$ satisfies with probability at least $99/100$, 
\begin{gather*}
    |f(\bv x, \bv y) - \bv x^T \bv A \bv y| \leq \epsilon n,
    \text{ for all unit }\bv x, \bv y \in \mathbb{R}^n,
\end{gather*}
 must read $\Omega\left(\frac{n}{\epsilon^2}\right)$ entries of $\bv A$ in the worst case, provided $\epsilon = \Omega\left(\frac{\log n}{\sqrt{n}}\right)$.
\end{restatable}

\subsubsection{Relation to Spectral Graph Sparsification}
We remark that while our work focuses on general bounded entry symmetric (PSD) matrices, when $\bv{A}$ is a PSD graph Laplacian matrix, it is possible to construct a sparsifier $\bv{\widetilde A}$ with just $\widetilde O(n/\epsilon^2)$ entries, that achieves a \emph{relative error} spectral approximation guarantee that can be much stronger than the additive error guarantee of \eqref{eqn:error} \cite{Batson:2013va}. In particular, one can achieve $(1-\epsilon) \bv{x}^T \bv A \bv{x} \le \bv{x}^T \bv{\widetilde A} \bv{x} \le (1+\epsilon) \bv{x}^T \bv A \bv{x}$ for all $\bv x \in \R^n$. To achieve such a guarantee however, it is not hard to see that the set of entries sampled in $\bv{\widetilde A}$ must depend on $\bv A$, and cannot be universal. Fast randomized algorithms for constructing $\bv{\widetilde A}$ have been studied extensively \cite{ST04,Spielman:2008wm,Batson:2013va}. Recent work has also made progress towards developing fast deterministic algorithms \cite{Chuzhoy:2020wq}.

\subsection{Road Map}

We introduce notation in Section \ref{section:notation}.  In Section \ref{sec:upper}, we prove our main universal sparsification results for PSD and non-PSD matrices (Theorems \ref{thm:PSD_Quad} and \ref{thm: general matrix eps4 bound}). In Section \ref{sec:sparsityLower} we prove  nearly matching lower bounds (Theorems \ref{theorem:psd_deterministic_lower} and \ref{thm: lower bound nonadaptive}). We also prove Theorems \ref{thm: lower bound a1} and \ref{theorem:adaptive_random_lower_bound} which give lower bounds for general deterministic (possibly adaptive) spectral approximation algorithms, and general randomized algorithms, respectively.
 In Section \ref{section:binary_magnitude}, we show how to give tighter deterministic spectral approximation results for PSD matrices with entries  in  $\{-1,0,1\}$. Finally, in Section \ref{sec:applications}, we discuss applications of our results to fast deterministic algorithms for singular value and vector approximation.

\section{Notation and Preliminaries}\label{section:notation}

We start by defining notation  used throughout. For any integer $n$, let $[n]$ denote the set  $\{1,2,\ldots, n\}$. 

\medskip

\noindent\textbf{Matrices and Vectors.} Matrices are represented with bold uppercase literals, e.g., $\bv A$. Vectors are represented with bold lowercase literals, e.g., $\bv x$. $\bv 1$ and $\bv 0$ denote all ones (resp. all zeros) matrices or vectors. The identity matrix is denoted by $\bv I$. The size of these matrices vary based on their applications.
For a vector $\bv{x}$, $\bv{x}(j)$ denotes its $j^{th}$ entry. For a matrix $\bv A$, $\bv A_{ij}$ denotes the entry in the \ith row and $j$\textsuperscript{th} column. For a vector $\bv{x}$ (or matrix $\bv{A}$), $\bv{x}^T$ (resp. $\bv{A}^T)$ denotes its transpose. For two matrices $\bv{A},\bv{B}$ of the same size, $\bv{A} \circ \bv{B}$ denotes the entrywise (Hadamard) product.

\medskip

\noindent\textbf{Matrix Norms and Properties.} For a vector $\bv x$, $\|\bv x\|_2$ denotes its Euclidean norm and $\norm{\bv x}_1 = \sum_{i=1}^n |\bv x(i)|$ denotes its $\ell_1$ norm. We denote the eigenvalues of a symmetric matrix $\bv A$ as $\lambda_1(\bv A) \geq \lambda_2(\bv A) \geq \ldots \geq \lambda_n(\bv A)$ in decreasing order. A symmetric matrix is positive semidefinite (PSD) if $\lambda_i \geq 0$ for all $i \in [n]$. The singular values of a matrix $\bv A$ are denoted as $\sigma_{1}(\bv A) \geq \sigma_2(\bv A) \geq \ldots \geq \sigma_n(\bv A) \ge 0$ in decreasing order. We let $\|\bv A\|_2 = \max_x\frac{\|\bv A\bv x\|_2}{\|\bv x\|_2} = \sigma_1(\bv A)$ denote the spectral norm, $\|\bv A\|_{\infty}$ denote the largest magnitude of an entry, $\|\bv A\|_F = (\sum_{i,j} \bv A_{ij}^2)^{1/2}$ denote the Frobenius norm, and $\|\bv A\|_1 = \sum_{i=1}^n \sigma_{i}(\bv A)$ denote the Schatten-1 norm (also called the trace norm or nuclear norm).

\medskip

\noindent\textbf{Expander Graphs.} Our universal sparsifier constructions are based on Ramanujan expander graphs \cite{lubotzky1988ramanujan,margulis1973explicit,alon1986eigenvalues}, defined below.
\begin{definition}[Ramanujan Expander Graphs \cite{lubotzky1988ramanujan}]
\label{def:ramanujan}
Let $G$ be a connected $d$-regular, unweighted and undirected graph on $n$ vertices. Let $\bv G \in \{0,1\}^{n \times n}$ be the adjacency matrix corresponding to $G$ and $\lambda_i$ be its $i^{th}$ eigenvalue, such that $\lambda_1 \geq \lambda_2 \geq \ldots \geq \lambda_n$. Then $G$ is called a Ramanujan graph if:
\begin{align*}
    |\lambda_i| &\leq 2\sqrt{d-1}, ~ \text{for all} ~ i >1.
\end{align*}
Equivalently, letting $\bv{1}$ be the $n \times n$ all ones matrix, $\norm{\bv{1} - \frac{n}{d} \bv{G}}_2 \le \frac{n}{d} \cdot 2\sqrt{d-1}$.
\end{definition}

\noindent \textbf{Efficient Construction of Ramanujan graphs.} Significant work has studied efficient constructions for Ramanujan Graphs \cite{margulis1973explicit,lubotzky1988ramanujan,Morgenstern:1994vf,Cohen:2016td}, and nearly linear time constructions, called \textit{strongly explicit} constructions~\cite{alon2021explicit}, have been proposed. E.g., by Proposition 1.1 of~\cite{alon2021explicit} we can reconstruct for any $n$ and $d$, a graph on $n(1+o(1))$ vertices with second eigenvalue at most $(2+o(1)) \sqrt{d}$ in $\Tilde{O}(nd + \poly(d))$ time. 
Additionally, in our applications, the expander just needs to be constructed once and can then be used for any input of size $n$. In fact, a single expander can be used for a range of input sizes, by the argument below. 

Though the size of the expander above is $(1+o(1))n$ instead of exactly $n$, and its expansion does not exactly hit the tight Ramanujan bound of $2 \sqrt{d-1}$, we can let our input matrix $\bv{A}$ be the top $n \times n$ principal submatrix of a slightly larger $n(1+o(1)) \times n(1+o(1))$ matrix $\bv{A}'$ that is zero everywhere else. Then, the top $n \times n$ principal submatrix of a spectral approximation to $\bv{A}'$ is a spectral approximation to $\bv{A}$. Thus, obtaining a spectral approximation to $\bv{A}'$ via a Ramanujan sparsifier on $n(1+o(1))$ vertices suffices. Similarly, in our applications, we can adjust $d = 1/\poly(\epsilon)$ by at most a constant factor to account for the $(2+o(1)) \sqrt{d}$ bound on the second eigenvalue, rather than the tight Ramanujan bound of $2\sqrt{d-1}$.


\section{Universal Sparsifier Upper Bounds}\label{sec:upper}

We now prove our main results on universal sparsifiers for PSD matrices (Theorem \ref{thm:PSD_Quad}) and general symmetric matrices (Theorem \ref{thm: general matrix eps4 bound}). Both theorems follow from general reductions which show that any sampling matrix $\bv{S}$ that sparsifies the all ones matrix to sufficient accuracy (i.e., is a sufficiently good spectral expander) yields a universal sparsifier. We prove the reduction for the PSD case in Section \ref{sec:psdNewProof}, and then extend it to the non-PSD case in Section \ref{sec:nonPsdNewProof}.

\subsection{Universal Sparsifiers for PSD Matrices}\label{sec:psdNewProof}

In the PSD case, we prove the following:

\thmReduction*
Theorem \ref{thm:PSD_Quad} follows directly from Theorem \ref{thm:reduction} by letting $\bv{S}$ be the scaled adjacency matrix of a Ramanujan expander graph with degree $d = O(1/\epsilon^2)$ (Definition \ref{def:ramanujan}).

\begin{proof}[Proof of Theorem \ref{thm:reduction}]
To prove the theorem it suffices to show that for any $\bv x \in \R^{n}$ with $\norm{\bv x}_2 = 1$, $|\bv{x}^T \bv{A} \bv{x} - \bv{x}^T (\bv{A} \circ \bv{S}) \bv{x}| \le \epsilon n$.
Let $\bv{v}_1,\ldots,\bv{v}_n$ and $\lambda_1,\ldots,\lambda_n$ be the eigenvectors and eigenvalues of $\bv{A}$ so that $\bv{A} = \sum_{i=1}^n \lambda_i \bv{v}_i \bv{v}_i^T$. Then we can expand out this error as:
\begin{align*}
|\bv{x}^T \bv{A} \bv{x} - \bv{x}^T (\bv{A} \circ \bv{S}) \bv{x}| &= \left | \sum_{i=1}^n \lambda_i \cdot \bv{x}^T \bv{v}_i \bv{v}_i^T \bv{x} - \sum_{i=1}^n \lambda_i \bv{x}^T ( \bv{v}_i \bv{v}_i^T \circ \bv S) \bv{x}) \right |\\
&= \left | \sum_{i=1}^n \lambda_i \bv{x}^T [\bv{v}_i \bv{v}_i^T \circ (\bv{1} - \bv{S})] \bv{x} \right  |,
\end{align*}
where we use that for any matrix $\bv{M}$, $\bv{M} \circ \bv 1 = \bv M$.
Now observe that if we let $\bv{D}_i \in \R^{n \times n}$ be a diagonal matrix with the entries of $\bv{v}_i$ on its diagonal, then we can write $\bv{v}_i \bv{v}_i^T \circ (\bv{1} - \bv{S}) = \bv{D}_i (\bv{1} - \bv{S}) \bv{D}_i$. Plugging this back in we have:
\begin{align}
|\bv{x}^T \bv{A} \bv{x} - \bv{x}^T (\bv{A} \circ \bv{S}) \bv{x}| &= \left |\sum_{i=1}^n \lambda_i \bv{x}^T \bv{D}_i (\bv{1}-\bv{S})\bv{D}_i \bv{x} \right |\nonumber\\
&\le \sum_{i=1}^n \lambda_i \norm{\bv 1 - \bv S}_2 \cdot \bv{x}^T \bv{D}_i^2 \bv{x}\nonumber \\
&\le \epsilon n \cdot \sum_{i=1}^n \lambda_i \bv{x}^T \bv{D}_i^2 \bv{x}.\label{eq:dsquare}
\end{align}
In the second line we use that $\bv{A}$ is PSD and thus $\lambda_i$ is non-negative for all $i$. We also use that $|\bv{x}^T \bv{D}_i (\bv{1}-\bv{S})\bv{D}_i \bv{x}| \le \norm{\bv{1} - \bv S}_2 \cdot \norm{\bv{D}_i \bv x}_2^2 = \norm{\bv{1} - \bv S}_2  \cdot \bv{x}^T \bv D_i^2 \bv{x}$. In the third line we bound $\norm{\bv 1 - \bv S}_2 \le \epsilon n$ using the assumption of the theorem statement.

Finally, writing $\bv{x}^T \bv{D}_i^2 \bv{x} = \sum_{j=1}^n \bv{x}(j)^2 \bv{v}_i(j)^2$ and plugging back into \eqref{eq:dsquare}, we have:
\begin{align*}
|\bv{x}^T \bv{A} \bv{x} - \bv{x}^T (\bv{A} \circ \bv{S}) \bv{x}| &\le \epsilon n \cdot \sum_{i=1}^n \sum_{j=1}^n \lambda_i  \bv{x}(j)^2 \bv{v}_i(j)^2\\
&= \epsilon n \cdot \sum_{j=1}^n \bv{x}(j)^2 \sum_{i=1}^n \lambda_i \bv{v}_i(j)^2\\
&= \epsilon n \cdot \sum_{j=1}^n \bv{x}(j)^2 \bv{A}_{jj}\\ &\le \epsilon n,
\end{align*}
where in the last step we use that $\bv{A}_{jj} \le 1$  by our bounded entry assumption, and that $\bv{x}$ is a unit vector. This completes the theorem.
\end{proof}

\begin{remark}\label{rem:diag}
Note that we can state a potentially stronger bound for Theorem \ref{thm:reduction} by observing that in the second to last step, $|\bv{x}^T \bv{A} \bv{x} - \bv{x}^T (\bv{A} \circ \bv{S}) \bv{x}| \le \epsilon n \cdot \sum_{j=1}^n \bv x(j)^2 \bv{A}_{jj} = \epsilon n \cdot \bv x^T \bv D \bv x$, where $\bv{D}$ is a diagonal matrix containing the diagonal elements of $\bv A$. That is, letting $\bv M \preceq \bv N$ denote that $\bv N - \bv M$ is PSD, we have the following spectral approximation bound:
$-\epsilon n \cdot \bv D \preceq \bv A - (\bv A \circ \bv S) \preceq \epsilon n \cdot \bv D.$
\end{remark}

\subsection{Universal Sparsifiers for Non-PSD Matrices}\label{sec:nonPsdNewProof}

We next extend the above approach to give a similar reduction for universal sparsification of general symmetric matrices.
\reductionNonPSD*
Observe that as compared to the PSD case (Theorem \ref{thm:reduction}), here we require that $\bv{S}$ gives a stronger approximation to the all ones matrix. Theorem \ref{thm: general matrix eps4 bound} follows directly by letting $\bv{S}$ be the scaled adjacency matrix of a Ramanujan expander graph with degree $d = O\left (\frac{\log^4(1/\epsilon)}{\epsilon^4}\right )$ (Definition \ref{def:ramanujan}).

\begin{proof}[Proof of Theorem \ref{thm:reductionNonPSD}]
To prove the theorem, it suffices to show that for any $\bv x \in \R^{n}$ with $\norm{\bv x}_2 = 1$, $|\bv{x}^T \bv{A} \bv{x} - \bv{x}^T (\bv{A} \circ \bv{S}) \bv{x}| \le \epsilon \cdot \max(n, \norm{\bv A}_1)$. We will split the entries of $\bv{x}$ into level sets according to their magnitude. In particular, let $\bepsilon = \frac{\epsilon}{\log_2(1/\epsilon)}$ and let $\ell = \log_2(1/\bepsilon)$. Then we can write $\bv{x} = \sum_{i=0}^{\ell+1}\bv{x}_i$, such that for any $t \in [n]$:
\begin{align}\label{eq:level}
    \bv x_i(t) =
    \begin{cases}
      \bv x(t) & \text{if }i \in \{1,2, \ldots, \ell\}\text{ and } |\bv x(t)|\in\left(\frac{2^{i-1}}{\sqrt{n}}, \frac{2^i}{\sqrt{n}}\right]\\
      \bv x(t) & \text{if }i=0\text{ and }|\bv x(t)|\in\left[0, \frac{1}{\sqrt{n}}\right] \\
        \bv x(t) & \text{if }i=\ell+1\text{ and }|\bv x(t)|\in\left ( \frac{1}{\bepsilon \sqrt{n}}, 1\right] \\
      0 & \text{otherwise}.
    \end{cases}
\end{align}
Via triangle inequality, we have:
\begin{align}\label{eq:firstSplit}
 \left|\bv x^T\bv A\bv x -  \bv x^T(\bv{A} \circ \bv{S})\bv x \right| &\le \sum_{i=0}^{\ell+1}  \left | \bv x_i^T\bv A\bv x -  \bv x_i^T(\bv{A} \circ \bv{S})\bv x \right |.
  \end{align}
  We will bound each term in \eqref{eq:firstSplit} by $O(\bepsilon\cdot \max(n,\norm{\bv A}_1)) = O \left ( \frac{\epsilon \cdot \max(n,\norm{\bv A}_1)}{\log(1/\epsilon)}\right ) = O \left ( \frac{\epsilon \cdot \max(n,\norm{\bv A}_1)}{\ell+2}\right )$. Summing over all $\ell+2$ terms we achieve a final  bound of $O(\epsilon \cdot \max(n,\norm{\bv A}_1))$. The theorem then follows by adjusting $\epsilon$ by a constant factor.
  
 Fixing $i \in \{0,\ldots,\ell\}$, let $\bv{x}_L(t) = \bv{x}(t)$ if $| \bv{x}(t)| \le \frac{1}{2^i \cdot \bepsilon \sqrt{n}}$ and $\bv{x}_L(t) = 0$ otherwise. Let $\bv{x}_H(t)  = \bv{x}(t)$ if $| \bv{x}(t)| > \frac{1}{2^i \cdot \bepsilon \sqrt{n}}$ and $\bv x_H(t) = 0$ otherwise.
 In the edge case, for $i = \ell+1$, let $\bv{x}_H = \bv{x}$ and $\bv{x}_L = \bv{0}$.
 Writing $\bv{x} = \bv{x}_H + \bv{x}_L$ and applying triangle inequality, we can bound:
 \begin{align}\label{eq:highLow}
 \left | \bv x_i^T\bv A\bv x -  \bv x_i^T(\bv{A} \circ \bv{S})\bv x \right | \le  \left | \bv x_i^T\bv A\bv x_L -  \bv x_i^T(\bv{A} \circ \bv{S})\bv x_L \right |+ \left | \bv x_i^T\bv A\bv x_H -  \bv x_i^T(\bv{A} \circ \bv{S})\bv x_H \right |.
 \end{align}
 In the following, we separately bound the two terms in \eqref{eq:highLow}. Roughly, since $\bv{x}_L$ has relatively small entries and thus is relatively well spread, we will be able to show, using a similar approach to Theorem \ref{thm:reduction}, that the first term is small since sparsification with $\bv S$ approximately preserves $\bv{x}_i^T\bv A \bv{x}_L$. On the otherhand, since $\bv{x}_H$ has relatively large entries and thus is relatively sparse, we can show that the second term is small since $\bv{x}_i^T\bv A \bv{x}_H$ is small (and $\bv{x}_i^T(\bv A \circ \bv S) \bv x_H$ cannot be much larger).
 
 \medskip
 
 \noindent \textbf{Term 1: Well-Spread Vectors}. We start by bounding $ \left | \bv x_i^T\bv A\bv x_L-  \bv x_i^T(\bv{A} \circ \bv{S})\bv x_L \right |$. Let $\bv{v}_1,\ldots,\bv{v}_n$ and $\lambda_1,\ldots,\lambda_n$ be the eigenvectors and eigenvalues of $\bv{A}$ so that $\bv{A} = \sum_{k=1}^n \lambda_k \bv{v}_k \bv{v}_k^T$. Then we  write:
\begin{align*}
|\bv{x}_i^T \bv{A} \bv{x}_L - \bv{x}_i^T (\bv{A} \circ \bv{S}) \bv{x}_L| &= \left | \sum_{k=1}^n \lambda_k \cdot \bv{x}_i^T \bv{v}_k\bv{v}_k^T \bv{x}_L - \sum_{k=1}^n \lambda_k \bv{x}_i^T ( \bv{v}_k \bv{v}_k^T \circ \bv S) \bv{x}_L) \right |\\
&\le  \sum_{k=1}^n |\lambda_k| \cdot \left |\bv{x}_i^T [\bv{v}_k \bv{v}_k^T \circ (\bv{1} - \bv{S})] \bv{x}_L \right  |.
\end{align*}
As in the proof of Theorem \ref{thm:reduction}, if we let $\bv{D}_k \in \R^{n \times n}$ be a diagonal matrix with the entries of $\bv{v}_k$ on its diagonal, then we have $\bv{v}_k \bv{v}_k^T \circ (\bv{1} - \bv{S}) = \bv{D}_k (\bv{1} - \bv{S}) \bv{D}_k$. Further,  $\bv{x}_i^T\bv{D}_k (\bv{1} - \bv{S}) \bv{D}_k \bv{x}_L = \bv{v}_k^T \bv{D}_i (\bv{1} - \bv{S}) \bv{D}_L \bv{v}_k$, where $\bv{D}_i,\bv{D}_L$ are diagonal with the entries of $\bv{x}_i$ and $\bv{x}_L$ on their diagonals.  Plugging this back in we have:
\begin{align}
|\bv{x}_i^T \bv{A} \bv{x}_L - \bv{x}_i^T (\bv{A} \circ \bv{S}) \bv{x}_L| &\le \sum_{k=1}^n| \lambda_k| \cdot \left |\bv{v}_k^T \bv{D}_i (\bv{1}-\bv{S})\bv{D}_L \bv{v}_k^T \right |\nonumber\\
&\le \norm{\bv 1 - \bv S}_2 \cdot \sum_{k=1}^n |\lambda_k| \cdot  \| \bv{v}_k^T\|_2 \cdot \| \bv{D}_i \|_2 \cdot \| \bv{D}_L \|_2 \cdot \| \bv{v}_k \|_2 \nonumber\\
&\le \norm{\bv 1 - \bv S}_2 \cdot \sum_{k=1}^n |\lambda_k| \cdot \| \bv{D}_i \|_2 \cdot \| \bv{D}_L \|_2.\label{eq:line3}
\end{align}
Finally, observe that by definition, for any $i \in \{0,\ldots,\ell\}$, $\| \bv{D}_i \|_2 =\max_{t \in [n]}|\bv{x}_i(t)|\leq\frac{2^i}{\sqrt{n}}$ and $\|\bv{D}_L \|_2=\max_{t \in [n]}|\bv{x}_L(t)| \leq \frac{1}{2^i\bepsilon \sqrt{n}}$. Thus, $\| \bv{D}_i \|_2 \cdot \| \bv{D}_L \|_2 \leq \frac{1}{\bepsilon n}$. For $i = \ell+1$, $\bv{x}_L = \bv 0$ by definition and so, the bound $\| \bv{D}_i \|_2 \cdot \| \bv{D}_L \|_2 \leq \frac{1}{\bepsilon n}$ holds vacuously.

Also, by the assumption of the theorem, $\norm{\bv{1}-\bv{S}}_2 \le \frac{\epsilon^2 n}{c \log^2(1/\epsilon)} \le \bepsilon^2 n$. Plugging back into \eqref{eq:line3}, 
\begin{align}
|\bv{x}_i^T \bv{A} \bv{x}_L - \bv{x}_i^T (\bv{A} \circ \bv{S}) \bv{x}_L| \le \bepsilon^2 n \cdot \frac{1}{\bepsilon n} \cdot \norm{\bv{A}}_1 \le \bepsilon \cdot \norm{\bv{A}}_1,\label{eq:spreadBound}
\end{align}
as required.

\medskip

\noindent \textbf{Term 2: Sparse Vectors.}  We next bound  the second term of \eqref{eq:highLow}: $ \left | \bv x_i^T\bv A\bv x_H -  \bv x_i^T(\bv{A} \circ \bv{S})\bv x_H \right |$. We 
 write $\bv{x}_i = \bv{x}_{i,P} + \bv{x}_{i,N}$, where $\bv{x}_{i,P}$ and $\bv{x}_{i,N}$ contain its positive and non-positive entries respectively. Similarly, write $\bv{x}_H = \bv{x}_{H,P} + \bv{x}_{H,N}$. We can then bound via triangle inequality:
\begin{align}
|\bv x_i^T \bv A \bv x_H - \bv x_i^T (\bv A \circ \bv S) \bv x_H| &\le |\bv x_i^T \bv A \bv x_H| + |\bv x_i^T (\bv A \circ \bv S) \bv x_H|\nonumber\\
&\le |\bv x_i^T \bv A \bv x_H| +|\bv x_{i,P}^T( \bv A \circ \bv S) \bv x_{H,P}|  + |\bv x_{i,P}^T ( \bv A \circ \bv S) \bv x_{H,N}| \nonumber\\ &\hspace{2em}+ |\bv x_{i,N}^T ( \bv A \circ \bv S) \bv x_{H,P}|  + |\bv x_{i,N}^T ( \bv A \circ \bv S) \bv x_{H,N}|.\label{eq:pnsplit}
\end{align}
We will bound each term in \eqref{eq:pnsplit} by $O(\bepsilon n)$, giving that overall  $|\bv x_i^T \bv A \bv x_H - \bv x_i^T (\bv A \circ \bv S) \bv x_H| = O(\bepsilon n)$. 
We first observe that $$|\bv{x}_i^T \bv A \bv{x}_H| \le \norm{\bv x_i}_1 \cdot \norm{\bv{x}_H}_1 \cdot \norm{\bv{A}}_\infty.$$ 
By assumption $\norm{\bv A}_\infty \le 1$. Further, for $i \in \{1,\ldots,\ell+1\}$, since $|\bv{x}_i(t)| \ge \frac{2^{i-1}}{\sqrt{n}}$ for all $t$ and since $\norm{\bv{x}_i}_2 \le 1$, $\bv{x}_i$ has at most $\frac{n}{2^{2i-2}}$ non-zero entries. Thus,  $\norm{\bv{x}_i}_1 \le \frac{ \sqrt{n}}{2^{i-1}}$. For $i = 0$, this bound holds trivially since $\norm{\bv{x}_i}_1 \le \sqrt{n} \cdot \norm{\bv x_i}_2 \le \sqrt{n} \le \frac{\sqrt{n}}{2^{i-1}}$.  

Similarly, for $i \in \{0,\ldots \ell\}$, since $\norm{\bv{x}_H}_2 \le 1$ and since $|\bv{x}_H(t)| \ge \frac{1}{2^i \bepsilon \sqrt{n}}$ by definition, $\bv{x}_H$ has at most $2^{2i} \cdot \bepsilon^2 n$ non-zero entries. So $\norm{\bv{x}_H}_1 \le 2^i \bepsilon \cdot \sqrt{n}$. For $i = \ell+1 = \log_2(1/\bepsilon)+1$ this bound holds trivially since $\norm{\bv x_H}_1 \le \sqrt{n}\cdot \norm{\bv x_H}_2 \le 2^i \bar \epsilon \cdot \sqrt{n}$. Putting these all together, we have
\begin{align}\label{eq:xH}|\bv{x}_i^T \bv A \bv{x}_H| \le \norm{\bv x_i}_1 \cdot \norm{\bv{x}_H}_1\cdot  \norm{\bv{A}}_\infty \le  \frac{ \sqrt{n}}{2^{i-1}} \cdot 2^i \bepsilon \cdot \sqrt{n} = 2 \bepsilon n.
\end{align}
We next bound $|\bv x_{i,P}^T( \bv A \circ \bv S) \bv x_{H,P}| $. Since $\bv{x}_{i,P}$ and $\bv{x}_{H,P}$ are both all positive vectors, and since by assumption $\norm{\bv{A}}_\infty \le 1$, 
\begin{align}
|\bv x_{i,P}^T( \bv A \circ \bv S) \bv x_{H,P}| &\le |\bv x_{i,P}^T \bv S \bv x_{H,P}|\nonumber\\
&\le |\bv x_{i,P}^T \bv 1 \bv x_{H,P}| + |\bv x_{i,P}^T (\bv 1 - \bv S) \bv x_{H,P}|\nonumber\\
&\le \norm{\bv x_{i,P}}_1 \cdot \norm{\bv x_{H,P}}_1 \cdot \norm{\bv{1}}_\infty + \norm{\bv x_{i,P}}_2  \cdot \norm{\bv x_{H,P}}_2 \cdot \norm{\bv{1}-\bv{S}}_2.\label{eq:secondtoLast}
\end{align}
Following the same argument used to show \eqref{eq:xH}, $\norm{\bv x_{i,P}}_1 \cdot \norm{\bv x_{H,P}}_1 \cdot \norm{\bv{1}}_\infty \le 2 \bepsilon n$. Further, since by the assumption of the theorem, $\norm{\bv 1 - \bv S}_2 \le \frac{\epsilon^2 n}{c \log^2(1/\epsilon)}\le  \bepsilon^2 n \le \bepsilon n$, and since $\bv{x}_{i,P}$ and $\bv{x}_{H,P}$ both are at most unit norm, $\norm{\bv x_{i,P}}_2 \cdot \norm{\bv x_{H,P}}_2 \cdot \norm{\bv{1}-\bv{S}}_2 \le \bepsilon n$. Thus, plugging back into \eqref{eq:secondtoLast}, we  have $|\bv x_{i,P}^T( \bv A \circ \bv S) \bv x_{H,P}| \le 3 \bepsilon n$. Identical bounds will hold for the remaining three terms of \eqref{eq:pnsplit}, since for each, the two vectors in the quadratic form have entries that either always match or never match on sign. Plugging these bounds  and \eqref{eq:xH} back into \eqref{eq:pnsplit}, we obtain
\begin{align}\label{eq:sparseBound}
|\bv x_i^T \bv A \bv x_H - \bv x_i^T (\bv A \circ \bv S) \bv x_H| \le 2 \bepsilon n + 4 \cdot 3 \bepsilon n \le 14 \bepsilon n,
\end{align}
which completes the required bound in the sparse case.

\medskip

\noindent\textbf{Concluding the Proof.} We finally plug our bounds for the sparse case \eqref{eq:sparseBound} and the well-spread case \eqref{eq:spreadBound} into \eqref{eq:highLow} to obtain:
\begin{align*}
 \left | \bv x_i^T\bv A\bv x -  \bv x_i^T(\bv{A} \circ \bv{S})\bv x \right | \le 14 \bepsilon n + \bepsilon \norm{\bv A}_1 \le 15 \bepsilon \cdot \max(n,\norm{\bv A}_1).
\end{align*}
Plugging this bound into \eqref{eq:firstSplit} gives that $|\bv{x}^T \bv{A} \bv{x} - \bv{x}^T (\bv A \circ \bv S)\bv x| = O(\epsilon \cdot  \max(n,\norm{\bv A}_1))$, which completes the theorem after adjusting $\epsilon$ by a constant factor.
\end{proof}


\section{Spectral Approximation Lower Bounds}\label{sec:sparsityLower}

We now show that our universal sparsifier upper bounds for both PSD matrices (Theorem \ref{thm:PSD_Quad}) and non-PSD matrices (Theorem \ref{thm: general matrix eps4 bound}) are nearly tight. We also give $\Omega(n/\epsilon^2)$ query lower bounds against general deterministic (possibly adaptive) spectral approximation algorithms (Theorem \ref{thm: lower bound a1}) and general randomized spectral approximation algorithms (Theorem \ref{theorem:adaptive_random_lower_bound}).

\subsection{Sparsity Lower Bound for PSD Matrices}

We first prove that every matrix which is an $\epsilon n$ spectral approximation to the all-ones matrix must have $\Omega(\frac{n}{\epsilon^2})$ non-zero entries. This shows that Theorem \ref{thm:PSD_Quad} is optimal up to constant factors, even for algorithms that sparsify just a single bounded entry PSD matrix. The idea of the lower bound is simple: if a matrix $\bv{S}$ spectrally approximates the all-ones matrix, its entries must sum to  $\Omega(n^2)$. Thus, if $\bv S$ has just $s$ non-zero entries, it must have Frobenius norm at least $\Omega(n^2/\sqrt{s})$. Unless $s = \Omega(n/\epsilon^2)$, this Frobenius norm is too large for $\bv S$ to be a $\epsilon n$-spectral approximation of the all ones matrix (which has $\norm{\bv 1}_F = n$.)

\alloneslower*

\begin{proof}
    If we let $\bv{x} \in \R^n$ be the all ones vector, then since $\norm{\bv{x}}_2^2 = n$, $\|\bv 1 - \bv S\|_2 \leq \epsilon n$ implies that $|\bv x ^T \bv S \bv x - \bv{x}^T \one \bv{x}| =  |\bv x^T \bv S \bv x - n^2| \leq \epsilon n^2$. This in turn implies:
    \begin{gather*}
        \sum_{i,j \in [n]} |\bv S_{ij}|
        \geq \sum_{i,j \in [n]} \bv{S}_{ij}
        \geq n^2 - \epsilon n^2 \geq \frac{n^2}{2},
    \end{gather*}
    where the last inequality follows from the assumption that $\epsilon \leq \frac{1}{2}$. If $\bv{S}$ has $s$ non-zero entries, since $\sum_{i,j \in [n]} |\bv S_{ij}|$ is the $\ell_1$-norm of the entries, and $\|\bv S\|_F$ is the $\ell_2$-norm of the entries, we conclude by $\ell_1-\ell_2$ norm equivalence that $\|\bv S\|_F \geq \frac{1}{\sqrt{s}}\sum_{i,j \in [n]} |\bv S_{ij}| \geq \frac{n^2}{2\sqrt{s}}$. Therefore, by the triangle inequality,
    \begin{align*}
    \norm{\bv{1}-\bv{S}}_F \ge \norm{\bv{S}}_F - \norm{\bv{1}}_F = \frac{n^2}{2\sqrt{s}}-n \ge \frac{n^2}{4\sqrt{s}},
    \end{align*}
    as long as $s \le \frac{n^2}{16}$ so that $\frac{n^2}{2\sqrt{s}} \ge 2n$. Now, using that $\norm{\bv{A}}_F^2 = \sum_{i=1}^n \sigma_i^2(\bv A) \le n \cdot \norm{\bv{A}}_2^2$, we have:
    \begin{align*}
      \norm{\bv{1}-\bv{S}}_F \ge \frac{n^2}{4\sqrt{s}} \implies   \norm{\bv{1}-\bv{S}}_F^2 \ge \frac{n^4}{16 s} \implies \norm{\bv{1}-\bv{S}}_2^2 \ge \frac{n^3}{16 s} \implies \norm{\bv{1}-\bv{S}}_2 \ge \frac{n^{3/2}}{4\sqrt{s}}.
    \end{align*}
    By our assumption that $\|\bv 1- \bv S\|_2 \leq \epsilon n$, this means that $s = \Omega\left(\frac{n}{\epsilon^2}\right)$, concluding the theorem.    
\end{proof}

\subsection{Lower Bounds for Deterministic Approximation of Non-PSD Matrices}

We next show that our universal sparsification bound for non-PSD matrices (Theorem \ref{thm: general matrix eps4 bound}) is tight up to a $\log^4(1/\epsilon)$ factor. Our lower bound holds even for the easier problem of top singular value (spectral norm) approximation and against a more general class of algorithms, which non-adaptively and deterministically query entries of the input matrix. We show how to extend the lower bound to possibly adaptive deterministic algorithms  in Theorem \ref{thm: lower bound a1}, but with a $1/\epsilon^2$ factor loss.

\lowerboundNonAdaptive*
\begin{proof}

Assume that we have a deterministic algorithm $\mathcal{A}$ that non-adaptively reads $s$ entries of any bounded  symmetric matrix $\bv{A}$ and outputs $\widetilde \sigma_1$ with $|\sigma_1(\bv A)-\widetilde \sigma_1| \le  \epsilon \max(n,\norm{\bv A}_1)$. 
Assume for the sake of contradiction that $s \le \frac{cn}{\epsilon^4}$ for some sufficiently small constant $c$.

Let $T \subset [n]$ be a set of $n^{3/2} /s^{1/2}$ rows on which $\mathcal{A}$ reads at most $\sqrt{ns}$ entries. Such a subset must exist since the average number of entries read in any set of $n^{3/2}/s^{1/2}$ rows is $\frac{n^{3/2}}{s^{1/2}} \cdot n \cdot \frac{s}{n^2} = \sqrt{ns}$.

Let $\bv 1_T$ be the matrix which is $1$ on all the rows in $T$ and zero everywhere else. Let $\bv{S}_T$ be the matrix which matches $\bv{1}_T$ on all entries, except is $0$ on any entry in the rows of $T$ that is not read by the algorithm $\mathcal{A}$. Observe that $\mathcal{A}$ reads the same entries (all ones) and thus outputs the same approximation $\widetilde \sigma_1$ for $\bv 1_T$ and $\bv S_T$.
We now bound the allowed error of this approximation. $\bv{1}_T$ is rank-1 and so:
$$\norm{\bv{1}_T}_1  = \sigma_1(\bv{1}_T) =  \norm{\bv{1}_T}_F = \frac{n^{5/4}}{s^{1/4}}.$$
We have $\sigma_1(\bv{S}_T) \le \norm{\bv S_T}_F \le n^{1/4} s^{1/4}$ since $\bv{S}_T$ has just $\sqrt{ns}$ entries set to $1$ -- the entries that $\mathcal{A}$ reads in the rows of $T$.  Using that $\bv{S}_T$ is supported only on the $n^{3/2}/s^{1/2}$ rows of $T$, $\bv{S}_T$ has rank at most  $n^{3/2}/s^{1/2}$. Thus, we can bound:
$$\norm{\bv{S}_T}_1 \le \frac{n^{3/4}}{s^{1/4}} \cdot \norm{\bv S_T}_F \le n.$$

Now, since $|\sigma_1(\bv{1}_T) - \sigma_1(\bv{S}_T| \ge \left |\frac{n^{5/4}}{s^{1/4}} - n^{1/4} s^{1/4}\right |$, our algorithm incurs error on one of the two input instances at least 
$$\frac{\left |\frac{n^{5/4}}{s^{1/4}} - n^{1/4} s^{1/4}\right |}{2} \ge \frac{n^{5/4}}{4 s^{1/4}},$$
where the inequality follows since, by assumption, $s \le \frac{cn}{\epsilon^4}$ for some small constant $c$ and $\epsilon \ge \frac{1}{n^{1/4}}$, and thus $n^{1/4} s^{1/4} \le \frac{n^{5/4}}{2s^{1/4}}$. 

The above is a contradiction when $\epsilon < 1/4$ since the above error is at least $1/4 \cdot \norm{\bv 1_T}_1$ and further, for $s \le \frac{cn}{\epsilon^4}$, the error it at least $\frac{\epsilon n}{4 c^{1/4}} > \epsilon n \ge \epsilon \cdot \norm{\bv S_T}_1$ if we set $c$ small enough. Thus, on at least one of the two input instances the error exceeds $\epsilon \cdot \max(n,\norm{\bv A}_1)$, yielding a contradiction.
\end{proof}

We can prove a variant on Theorem \ref{thm: lower bound nonadaptive} when the algorithm is allowed to make adaptive queries. Here, our lower bound reduces to $\Omega(n/\epsilon^2)$, as we are not able to restrict our hard case to a small set of rows of the input matrix. Closing the gap here -- either by giving a stronger lower bound in the adaptive case or giving an adaptive query deterministic algorithm that achieves $\widetilde O(n/\epsilon^2)$ query complexity is an interesting open question.

\lowerbounda*
\begin{proof}

Assume that we have a deterministic algorithm $\mathcal{A}$ that reads at most $s$ entries of any bounded entry matrix $\bv{A}$ and outputs $\widetilde \sigma_1$ with $|\sigma_1(\bv A)-\widetilde \sigma_1| \le  \epsilon \cdot \max(n,\norm{\bv A}_1)$. Assume for the sake of contradiction that $s \le \frac{cn}{\epsilon^2}$ for some sufficiently small constant $c$.

Let $S$ be the set of entries that $\mathcal{A}$ reads when given the all ones matrix $\bv{1}$ as input.
Let $\bv{S}$ be the matrix which is one on every entry in $S$ and zero elsewhere.  Observe that $\mathcal{A}$ reads the same entries and thus outputs the same approximation $\widetilde \sigma_1$ for $\bv 1$ and $\bv S$.
We now bound the allowed error of this approximation. $\bv{1}$ is rank-1 and has $\norm{\bv{1}}_1  = \sigma_1(\bv{1}) =  n$. We have $\sigma_1(\bv{S}) \le \norm{\bv S}_F = \sqrt{s}$ and can bound $\norm{\bv{S}}_1 \le \sqrt{n} \cdot \norm{\bv S}_F \le \sqrt{sn} \le \frac{c^{1/2}n}{\epsilon},$ where the last bound follows from our assumption that $s \le \frac{cn}{\epsilon^2}$. 
Now, since $|\sigma_1(\bv{1}) - \sigma_1(\bv{S}| \ge \left |n - \sqrt{s}\right |$, our algorithm incurs an error on one of the two input instances at least 
$$\frac{\left |n- \sqrt{s}\right |}{2} \ge \frac{n}{4},$$
where the inequality holds since, by assumption, $s \le \frac{cn}{\epsilon^2}$ for some small constant $c$ and $\epsilon \ge \frac{1}{n^{1/2}}$, and thus $\sqrt{s} \le \frac{n}{2}$. 

The above is a contradiction when $\epsilon < 1/4$ since the error mentioned above is at least $1/2 \cdot \norm{\bv 1}_1 = 1/2 \cdot \max(n,\norm{\bv 1}_1)$. Furthermore, since we have bounded $\norm{\bv S}_1 \le \frac{c^{1/2} n}{\epsilon}$, the error is greater than $\epsilon \cdot \max (n,\norm{\bv S}_1)$ when $c < 1$. Thus, on at least one of the two input instances the error exceeds $\epsilon \cdot \max(n,\norm{\bv A}_1)$, yielding a contradiction.
\end{proof}

\subsection{Lower Bound for Randomized Approximation of Non-PSD Matrices}\label{section:randomized_query_lower}

Finally, we prove Theorem \ref{theorem:adaptive_random_lower_bound}, which gives an $\Omega(n/\epsilon^2)$ query  lower bound for spectral approximation of bounded entry matrices that holds even for randomized and adaptive algorithms that approximate $\bv{A}$ in an arbitrary manner (not necessarily via sparsification). In particular, we show the lower bound against algorithms that produce {any} data structure, $f(\cdot, \cdot)$, that satisfies $|f(\bv x, \bv y) - \bv x^T \bv A \bv y| \leq \epsilon n$ for all unit norm $\bv{x,y} \in \mathbb{R}^n$

To prove this lower bound, we let $\bv{A}$ be a random matrix where each row has binary entries that are either i.i.d. fair coin flips, or else are coin flips with bias $+ \epsilon$. Each row is unbiased with probability $1/2$ and biased with probability $1/2$.  We show that an $\epsilon n$-spectral approximation to $\bv{A}$ suffices to identify for at least a $9/10$ fraction of the rows, whether or not they are biased -- roughly since the approximation must preserve the fact that $\bv{x}^T \bv{A} \bv{x}$ is large when $\bv{x}$ is supported just on this biased set.
Consider a communication problem in the blackboard model, in which each of $n^2$ players can access just a single entry of $\bv{A}$. It is known that if the $n$ players corresponding to a single row of the matrix want to identify with good probability whether or not it is biased, they must communicate at least $\Omega(1/\epsilon^2)$ bits \cite{srinivas2022memory}. Further, via a direct sum type argument, we can show that for the $n^2$  players to identify the bias of a $9/10$ fraction of the rows with good probability, $\Omega(n/\epsilon^2)$ bits must be communicated. I.e.,  at least $\Omega(n/\epsilon^2)$  of the players must read their input bits, yielding our query complexity lower bound. 

Note that there may  be other proof approaches here, based on a direct sum of $n$ $2$-player Gap-Hamming instances \cite{BJKS04,CHR12,BGPW13}, but our argument is simple and already gives optimal bounds. 

\medskip

\noindent\textbf{Section Roadmap.} In Section \ref{section:info_complexity_lower}, we formally define the problem of identifying the bias of a large fraction of $\bv{A}$'s rows as the \emph{$(\epsilon, n)$-distributed detection problem} (Definition \ref{def:distdet_n}). We prove a $\Omega(n/\epsilon^2)$ query lower  bound for this problem  in Lemma \ref{lemma:entrywise_lower}..  Then, in Section \ref{section:spectral_info_reduction}, we prove Theorem \ref{theorem:adaptive_random_lower_bound} by showing a reduction from the distributed detection problem to spectral approximation.

\subsubsection{Distributed Detection Lower Bound}\label{section:info_complexity_lower}
Here, we define additional concepts and notation specific to this section. For random variables $X$, $Y$ and $Z$, let $H(X)$ denote the entropy, $H(X|Y)$ denote the conditional entropy, $I(X;Y)$ denote mutual information, and $I(X;Y | Z)$ denote conditional mutual information \cite{cover1999elements}. We also use some ideas from communication complexity. Namely, we work in the \emph{blackboard model}, where $T$ parties communicate by posting messages to a public blackboard with access to public randomness with unlimited rounds of adaptivity. Let $\Pi \in \{0,1\}^*$ denote the transcript of a protocol 
posted to the blackboard, and let $|\Pi|$ denote the length of a transcript.  For a fixed protocol with fixed probability of success, we may assume, without loss of generality, that $\Pi$ has a fixed pre-determined length (see Section 2.2 of \cite{roughgarden2016communication}).

Our query complexity lower bound for spectral approximation will follow from a reduction from the following testing problem.
\begin{definition}
    ($\epsilon$-Distributed detection problem \cite{srinivas2022memory}). For fixed distributions $\mu_0 = \operatorname{Bernoulli}(1/2)$ and $\mu_1 = \operatorname{Bernoulli}(1/2 + \epsilon)$, with $\epsilon \in [0, \frac{1}{2}]$, let $\bv x \in \{0, 1\}^n$ be a random vector such that $\bv x_1, \ldots, \bv x_n$ are sampled i.i.d. from $\mu_V$, for $V \in \{0, 1\}$. The distributed detection problem is the task of determining whether $V = 0$ or $V = 1$, given the values of $\bv x$.
\end{definition}
This decision problem can be naturally interpreted as a communication problem in the blackboard model, where $n$ players each have a single private bit corresponding to whether a unique entry of $\bv x$ is zero or one. Prior work takes this view to lower bound the mutual information between the transcript of a protocol which correctly solves the $\epsilon$-Distributed detection problem with constant advantage and the sampled bits in $\bv x$ in the case $V = 0$.
\begin{theorem}\label{theorem:lower_mutual_distdet}
    (Theorem 6 in \cite{srinivas2022memory}). Let $\Pi$ be the transcript of a protocol that solves $\epsilon$-distributed detection problem with probability $1-p$ for any fixed choice of $p \in [0, 0.5)$. Then    \begin{gather*}
        I(X; \Pi | V = 0) = \Omega(\epsilon^{-2}).
    \end{gather*}
\end{theorem}

Next, we define the $(\epsilon, n)$-Distributed detection problem, which combines $n$ length-$n$ $\epsilon$-Distributed detection problems into a single joint detection problem.
\begin{definition}\label{def:distdet_n}
    ($(\epsilon, n)$-Distributed detection problem) For $\bv v \in \{0, 1\}^n$ distributed uniformly on the Hamming cube, generate the matrix $\bv A \in \{0,1\}^{n \times n}$ such that if $\bv v_i = 0$, then all entries in the $i$-th row of $\bv A$ are i.i.d. samples from $\operatorname{Bernoulli}(1/2)$. Otherwise, let all entries in the $i$-th row be sampled from $\operatorname{Bernoulli}(1/2 + \epsilon)$.  The $(\epsilon, n)$-Distributed detection problem is the task of recovering a vector $\vhat \in \{0,1\}^n$ such that $\|\vhat - \bv v\|_1 \leq \frac{n}{20}$.
\end{definition}
Again, this vector recovery problem has a natural interpretation as a communication problem, where $n^2$ players each hold a single bit of information that corresponds to whether a unique entry of $\bv A$ is zero or one. Throughout this section, we will view Definition \ref{def:distdet_n} as a decision problem or communication problem as needed. Let $\Pi$ be the transcript of a protocol that solves the $(\epsilon, n)$-Distributed detection problem in the communication model introduced at the beginning of this section. Lower bounding the mutual information between the transcript, $\Pi$, and the private information held by the players (the entries of $\bv A$) lower bounds the length of the transcript via the following argument. 
\begin{lemma}\label{lemma:information_to_length}
    If $\Pi$ is the transcript of a protocol that solves the $(\epsilon, n)$-Distributed detection problem and $I(\bv A; \Pi) \geq b$, then $|\Pi| \geq \frac{b}{\log 2}$.
\end{lemma}
\begin{proof}
    For random variables $X,Y$, $I(X;Y) \leq \min \{H(X), H(Y)\}$ \cite{cover1999elements}.
    If a random variable $X$ has finite support, then $H(X) \leq \log |\operatorname{supp}(X)|$. Recall that $|\Pi|$ is the number of bits in the transcript and hence $\log 2^{|\Pi|} = \log 2 \cdot |\Pi| \geq H(\Pi) \geq  I(\bv A; \Pi) \geq b$.  Hence, we conclude the  statement.
\end{proof}

Next, we lower bound the number of entries in the matrix $\bv A$ that must be observed to solve a variant of the $(\epsilon, n)$-Distributed detection problem, where we aim to correctly decide each index of $\bv v$ with constant success probability, rather than constructing $\vhat$ such that $\|\vhat - \bv v\|_1 \leq \frac{n}{20}$. There are three sources of randomness in this problem: 1) the initial randomness in sampling $\bv v$, 2) the random variable $\bv A$ which depends on $\bv v$, and 3) the random transcript $\Pi$ which depends on $\bv A$.  When necessary to avoid confusion, we will be explicit regarding which of these variables are considered fixed and which are considered as random in probabilistic statements.

\begin{lemma}\label{lemma:independent_det_lower}
    Let $\bv A$ and $\bv v$ be distributed as in the $(\epsilon, n)$-Distributed detection problem.  Any protocol with transcript $\Pi$ that constructs a vector $\vhat$ such that $\Pr(\bv v_i = \vhat_i) \geq \frac{9}{10}$, for all $i\in [n]$, (where the randomness is with respect to $\bv A$, $\bv v$, and $\Pi$), must observe $\Omega(\frac{n}{\epsilon^2})$ entries of $\bv A$ in the worst case.
\end{lemma}
\begin{proof}
    Consider $\bv A$ as $n^2$ players each holding a single bit of information corresponding to whether a unique entry of $\bv A$ is zero or one. Here, let $\Pi \in \{0,1\}^*$ be the transcript of a protocol that satisfies $\Pr(\bv v_i = \vhat_i) \geq \frac{9}{10}$ for every $i \in [n]$. Note that any algorithm which solves this problem by querying $k$ entries of $\bv A$ implies a protocol for the communication problem which communicates $k$ bits, since the algorithm could be simulated by (possibly adaptively) posting each queried bit to the blackboard.

    First, we lower bound the un-conditional mutual information between the private information in the $i$-th row of $\bv A$ and the transcript $\Pi$.  Note that $\bv v_i$ is one or zero with equal probability, therefore,
    \begin{gather*}
        I(\bv A_i; \Pi | \bv v_i) = \frac{1}{2}I(\bv A_i; \Pi | \bv v_i = 0) + \frac{1}{2}I(\bv A_i; \Pi | \bv v_i = 1) 
        \geq \frac{1}{2}I(\bv A_i; \Pi | \bv v_i = 0).
    \end{gather*}
    Next, we decompose the mutual information by its definition and the chain rule, then, we use the fact that $0 \leq H(\bv v_i) \leq 1$ and that conditioning can only decrease the entropy of a random variable.
    \begin{gather*}
        I(\bv A_i; \Pi) =  I(\bv A_i; \Pi | \bv v_i)  + H(\bv v_i | \bv A_i, \Pi) - H(\bv v_i | \bv A_i) - H(\bv v_i | \Pi) + H(\bv v_i)\\
        \Rightarrow I(\bv A_i; \Pi) \geq I(\bv A_i; \Pi | \bv v_i) - 2.
    \end{gather*}
    Observe that determining $\vhat_i$ with the guarantee $\Pr(\bv v_i = \vhat_i) \geq \frac{9}{10}$ solves the $\epsilon$-Distributed detection problem with probability at least $\frac{9}{10}$. Therefore, by Theorem \ref{theorem:lower_mutual_distdet} $I(\bv A_i; \Pi | \bv v_i = 0) = \Omega(\epsilon^{-2})$, and so we can use the previous two equations to lower bound for the conditional mutual information:
    \begin{gather}\label{eq:unconditional_mutual_info}
        I(\bv A_i; \Pi) \geq  \frac{1}{2}I(\bv A_i; \Pi | \bv v_i = 0) - 2 = \Omega(\frac{1}{\epsilon^2}).
    \end{gather}

    Next, we lower bound the total mutual information between $\bv A$ and $\Pi$. Mirroring the argument of Lemma 1 in \cite{srinivas2022memory}, we use that, for independent random variables, entropy is additive and conditional entropy is subadditive to lower bound the mutual information between $\bv A$ and $\Pi$:
    \begin{align*}
        I(\{\bv A_i\}_{i\in [n]}; \Pi) &= H(\{\bv A_i\}_{i\in [n]}) - H(\{\bv A_i\}_{i\in [n]} | \Pi) \\
        &\geq \sum_{i=1}^n H(\bv A_i) - H(\bv A_i, \Pi) \\ 
        &= \sum_{i = 1}^n I(\bv A_i; \Pi) = \Omega\left(\frac{n}{\epsilon^2}\right),
    \end{align*}
    where the last step follows from \eqref{eq:unconditional_mutual_info}. By Lemma \ref{lemma:information_to_length}, $|\Pi| = \Omega(I(\{\bv A_i\}_{i\in [n]}; \Pi)) = \Omega(\frac{n}{\epsilon^2})$.  Therefore, every algorithm which samples entries of $\bv A$ to construct a vector $\vhat$ satisfying $\Pr(\bv v_i = \vhat_i) \geq \frac{9}{10}$ must observe at least $\Omega(\frac{n}{\epsilon^2})$ entries of $\bv A$.
\end{proof}

We now have the necessary results to prove a lower bound on the number of entries of $\bv A$ that must be observed to solve the $(\epsilon, n)$-Distributed detection problem, which we do by reducing to the problem in Lemma \ref{lemma:independent_det_lower}.

\begin{lemma}\label{lemma:entrywise_lower}
    Any adaptive randomized algorithm which solves the $(\epsilon, n)$-Distributed detection problem with probability at least $\frac{19}{20}$ (with respect to the randomness in $\bv v$, $\bv A$, and $\Pi$) must observe $\Omega(\frac{n}{\epsilon^2})$ entries of $\bv A$.
\end{lemma}
\begin{proof}
    Any algorithm that can produce a vector $\vhat$ such that $\|\vhat - \bv v \|_1 \leq \frac{n}{20}$ with probability at least $\frac{19}{20}$ could be used to guarantee that $\Pr(\vhat_i = \bv v_i) \geq \frac{9}{10}$ by the following argument.

    For any input matrix $\bv A$, create the matrix $\Abar$ such that $\Abar_{i} = \bv A_{\sigma(i)}$, where $\sigma$ is a permutation sampled uniformly from the symmetric group of size $n$. Let $\bv v_\sigma$ be the vector which satisfies $[\bv v_\sigma]_i = \bv v_{\sigma(i)}$. Run the considered protocol on $\Abar$ to recover $\bv u$ such that $\|\bv u - \bv v_\sigma\|_1 \leq \frac{n}{20}$ with probability $\frac{19}{20}$.  Finally, let $\vhat_i = \bv u_{\sigma^{-1}(i)}$ for all $i \in [n]$.

    Then, $\Pr(\vhat_{i} = \bv v_{i}) = \Pr(\|\bv u - \bv v_\sigma\|_1 \leq \frac{n}{20}) \cdot \Pr(\bv u_{\sigma^{-1}(i)} - \bv v_{\sigma^{-1}(i)} | \|\bv u - \bv v_\sigma\|_1 \leq \frac{n}{20}) = \frac{19}{20} \cdot \frac{19}{20} \geq \frac{9}{10}$. Since $\sigma(i)$ is uniformly distributed over $[n]$, $\Pr(\bv u_{\sigma^{-1}(i)} - \bv v_{\sigma^{-1}(i)} | \|\bv u - \bv v_\sigma\|_1 \leq \frac{n}{20})$ is the number of correctly decided entries divided by $n$.

    By Lemma \ref{lemma:independent_det_lower}, we conclude that solving the $(\epsilon, n)$-Distributed detection with probability at least $\frac{19}{20}$ requires observing $\Omega(\frac{n}{\epsilon^2})$ entries of $\bv A$ in the worst case.
\end{proof}

\subsubsection{Spectral Approximation Query Lower Bound}\label{section:spectral_info_reduction}

Next, we show that the $(\epsilon, n)$-Distributed detection problem can be solved using a spectral approximation of $\bv A$, thereby implying a query complexity lower bound for spectral approximation.
\adaptiverandomlowerbound*

\begin{proof}

    We reduce solving the $(\epsilon, n)$-Distributed detection problem to constructing a spectral approximation satisfying the guarantee in the above theorem statement.  Throughout this proof, let $\epsilon_d$ be the parameter associated with the $(\epsilon, n)$-Distributed detection problem and $\epsilon_s$ be the accuracy of the spectral approximation.

    Let $\bv A$ be the matrix associated with the $(\epsilon_d, n)$-Distributed detection problem. Suppose that by reading $r$ entries of $\bv A$, we can create a data structure $f(\cdot, \cdot)$ satisfying:
    \begin{gather*}
        |f(\bv x, \bv y) - \bv x^T \bv A \bv y| \leq \epsilon_s n \|\bv x\|_2 \|\bv y\|_2.
        \text{ for all input }\bv x, \bv y \in \mathbb{R}^n.
    \end{gather*}

    We show that for $\epsilon_s$ sufficiently small respect to $\epsilon_d$, such a spectral approximation is sufficient to solve the $(\epsilon_d, n)$-Distributed detection problem with no further queries to $\bv A$.

     First, let $B_i$ denote the number of ones in the $i$-th row of $\bv A$. Observe that $B_i \sim \operatorname{Binomial}(n, 1/2)$ if $\bv v_i = 0$, and, $B_i \sim \operatorname{Binomial}(n, 1/2 + \epsilon_d)$ if $\bv v_i = 1$. By Hoeffding's inequality \cite{vershynin2018high},
    \begin{gather*}
        P\left(|B_i - E[B_i]| >  \sqrt{4n \log n}\right)
        \leq 2\exp\left(\frac{-4n \log n}{2n} \right)
        = \frac{2}{n^2}.
    \end{gather*}
    Therefore, by the union bound over all $i \in [n]$,
    \begin{gather*}
        P\left(\max_{i\in S} |B_i - E[B_i] |>  \sqrt{4n \log n}\right)
        \leq  \frac{2}{n}.
    \end{gather*}
    Let $S_1, S_2 \subset [n]$, such that $|S_1|=|S_2|=\frac{n}{2}$.  Then with probability at least $1 - \frac{2}{n}$,
    \begin{align}
        \left|\sum_{i \in S_1} B_i - \sum_{i \in S_2} B_i\right|
        &\leq \left|E\left[\sum_{i \in S_1} B_i - \sum_{i \in S_2} B_i\right]\right| + \frac{2n}{2} \cdot \sqrt{4n \log n} \nonumber\\ 
        &= \left|E\left[\sum_{i \in S_1} B_i - \sum_{i \in S_2} B_i\right]\right| + O(n^{3/2} \log n). \label{eq;row_concentration}
    \end{align}
    For large enough $n$, we have with probability at least $\frac{99}{100}$ that there are at least $\frac{n}{4}$ biased rows, since the number of biased rows is distributed as $\operatorname{Binomial}(n, 1/2)$.  Define the set $S^*$ such that $|S^*| = \frac{n}{2}$ and $S^*$ contains a maximal number of biased rows.  Define the set $\Sbar$ as,
    \begin{gather*}
        \Sbar = \argmax_{|S|=\frac{n}{2}} f(\one, \one_S).
    \end{gather*}
    We will show that $\Sbar$ can contain at most $k \leq \frac{n}{80} + o(n)$ fewer biased rows than $S^*$.  By guarantee of the spectral approximation and \eqref{eq;row_concentration}, (note $\one_S^T \bv A \one_S = \sum_{i,j \in S} \bv A_{ij}),$
    \begin{align*}
        f(\one, \one_{S^*}) - f(\one, \one_{\Sbar})
        &\geq \one^T \bv A \one_{S^*} - \one^T \bv A \one_{\Sbar} - 2\epsilon_s n\|\one\|_2\|\one_{S^*}\|_2 \\
        &\geq E[\one^T \bv A \one_{S^*} - \one^T \bv A \one_{\Sbar}] + O(n^{3/2} \log n) - 2\epsilon_s n^2.
    \end{align*}
    If $\Sbar$ has $k$ fewer biased rows than $S^*$, then, $E[\one^T \bv A \one_{S^*} - \one^T \bv A \one_{\Sbar}] = \epsilon_d nk$.  By the optimality of $\Sbar$, $f(\one, \one_{S^*}) - f(\one, \one_{\Sbar}) \leq 0$. Therefore, if $\epsilon_s = \frac{\epsilon_d}{160}$, then,
    \begin{gather*}
        0 \geq f(\one, \one_{S^*}) - f(\one, \one_{\Sbar})
        \geq \epsilon_d nk - 2\epsilon_s n^2 + O(n^{3/2} \log n)
        = \epsilon_d nk - \frac{\epsilon_d n^2}{80} + O(n^{3/2} \log n).
    \end{gather*}
    Solving for $k$ implies, $k \leq \frac{n}{80} + O(\frac{n^{1/2} \log n}{\epsilon_d})$. By assumption of the theorem statement, $\frac{1}{\epsilon_d} = o(\frac{\sqrt{n}}{\log n})$, therefore, $k \leq \frac{n}{80} + o(n)$.

    Let $b$ be the number of biased rows in $\bv A$.  First, we consider the case where $b = \frac{n}{2}$ exactly. In this case, $S^*$ contains $\frac{n}{2}$ biased rows, and hence, $\Sbar$ contains at least $\frac{n}{2} - \frac{n}{80} - o(n) \geq \frac{n}{2} - \frac{n}{40}$ biased rows for large enough $n$.  Therefore, deciding all rows in $\Sbar$ (i.e., $\vhat_i = 1$ for all $i \in \Sbar)$ are biased will correctly decide $(\frac{n}{2} - \frac{n}{40})/\frac{n}{2} \geq \frac{19}{20}$ of the biased rows. By looking at the complement of $S^*$ and $\Sbar$, we conclude that $\frac{19}{20}$ of the unbiased rows are decided correctly as well. Hence, we can construct $\vhat$ such that $\|\bv v - \vhat\|_1 \leq \frac{n}{20}$ by the assignment $\vhat_i = 1$ for $i \in \Sbar$ and $\vhat_i = 0$ otherwise.

    The number of biased rows is distributed as a binomial random variable, i.e., $b \sim \operatorname{Binomial}(n, 1/2)$.  Therefore, by Hoeffding's inequaliy, $\Pr(|b - \frac{n}{2}| > 10\sqrt{n}) < 0.01$.  For large enough $n$, $10\sqrt{n} \leq 0.01 \cdot n$.  Therefore, with probability at least $\frac{99}{100}$, we can reduce to the case $b = \frac{n}{2}$ by assuming that at most an additional $0.01 \cdot n$ rows are misclassified as biased or unbiased.  Hence, with probability at least $\frac{99}{100}$, a spectral approximation of $\bv A$ with $\epsilon_s = \Theta(\epsilon_d)$ accuracy is sufficient to recover at least $\frac{9}{10}$ of the biased rows with probability at least $\frac{99}{100}$.

    Correctly classifying $\frac{9}{10}$ of the rows as biased or unbiased requires observing $\Omega(\frac{n}{\epsilon_d^2})$ entries of $\bv A$ by Lemma \ref{lemma:entrywise_lower}. Since $\Omega(\frac{n}{\epsilon_d^2}) = \Omega(\frac{n}{\epsilon_s^2})$ entries of $\bv A$ must be observed to construct the spectral approximation used to solve the $(\epsilon, n)$-Distributed detection problem, we conclude the lower bound of the theorem statement.
\end{proof}

Note that the assumption $\frac{1}{\epsilon} = o\left(\frac{\sqrt{n}}{\log n}\right)$ in the previous theorem is mild, since if this assumption does not hold, then we must read nearly all entries of the matrix anyways. We also show that our construction  provides a lower bound even when the input is restricted to be symmetric.

\begin{corollary}
    The lower bound in Theorem \ref{theorem:adaptive_random_lower_bound} applies when the input is restricted to \textbf{symmetric} binary matrices.
\end{corollary}
\begin{proof}
    While construction used in Theorem \ref{theorem:adaptive_random_lower_bound} is not symmetric, we can modify it to give the same lower bound for \textbf{symmetric} input.
    Let $\bv A$ be defined as in the proof of Theorem \ref{theorem:adaptive_random_lower_bound}, and let $\Abar$ be the {Hermitian dilation} of $\bv A$, i.e.,
    \begin{gather*}
        \Abar = \begin{bmatrix}
            \bv 0 & \bv A \\ \bv A^T & \bv 0
        \end{bmatrix}.
    \end{gather*}
    Any query $\bv x^T\bv{ A y}$ can be simulated by the query $\bar {\bv x}^T \Abar \bar{\bv{y}}$, where $\bar{\bv{x}} = [\bv x , \bv 0]^T$ and $\bar{\bv{y}} = [\bv 0, \bv y]^T$. The size of $\bv A$ is $n$ and the size of $\Abar$ is $2n$, hence, if $\bar{f}$ is a spectral approximation of $\Abar$ such that,
    \begin{gather*}
        |\bar{f}(\bar{\bv{x}}, \bar{\bv{y}}) - \bar{\bv{x}^T} \Abar \bar{\bv{y}}| \leq \left(\frac{\epsilon_s}{2}\right) (2n) \|\bar{\bv{x}}\|_2 \|\bar{\bv{y}}\|_2,
        \text{ for all }\bar{\bv{x}}, \bar{\bv{y}} \in \mathbb{R}^n,
    \end{gather*}
    then we can simulate $f(\cdot, \cdot)$ such that,
    \begin{gather*}
        |f(\bv x, \bv y) - \bv x^T \bv A \bv y| \leq \epsilon_s n \|\bv x\|_2 \|\bv y\|_2,
        \text{ for all }\bv x, \bv y \in \mathbb{R}^n.
    \end{gather*}
    Therefore, the proof of Theorem \ref{theorem:adaptive_random_lower_bound} applies to the Hermitian dilation $\Abar$ after adjusting for a constant factor in the accuracy parameter $\epsilon_s$.
\end{proof}


\section{Improved Bounds for Binary Magnitude PSD Matrices}\label{section:binary_magnitude}

We call a PSD matrix $\bv A \in \R^{n\times n}$  a \emph{binary magnitude PSD matrix} if $\forall i,j\in [n], \left|\bv A_{ij}\right| \in \{0,1\}$, i.e., if $\bv A \in \{-1,0,1\}^{n\times n}$. In this section, we give deterministic spectral approximation algorithms for such matrices with improved $\epsilon$ dependencies. Specifically, we show that there exists a deterministic $\epsilon n$ error spectral approximation algorithm reading just $O \left ( \frac{n \log n}{\epsilon}\right)$ entries of the input matrix. This improves on our bound for general bounded entry PSD matrices (Theorem \ref{thm:PSD_Quad}) by a $1/\epsilon$ factor, while losing a $\log n $ factor. Further, we show that it is tight up to a logarithmic factor, via a simple application of Tur\'{a}n's theorem \cite{turaan1941extremal}.

The key idea of our improved algorithm is that for any PSD $\bv{A}$, we can write $\bv{A} = \bv{B}^T \bv{B}$, so that $\bv{A}_{ij} = \langle \bv b_i,\bv b_j \rangle$, where $\bv b_i,\bv b_j$ are the $i^{th}$ and $j^{th}$ columns of $\bv B$ respectfully. Since $\bv A$ has bounded entries, for any $i$, $\bv{A}_{ii} = \norm{\bv b_i}_2^2 \le 1$. Thus, we can apply transitivity arguments: if $|\bv{A}_{ij} | = |\langle \bv b_i,\bv b_j \rangle|$ is large (close to $1$) and $|\bv{A}_{jk} | = |\langle \bv b_j,\bv b_i \rangle|$ is also large, then $|\bv{A}_{ik} | =  |\langle \bv b_i,\bv b_k \rangle|$ must be relatively large as well. Thus, we can hope to infer the contribution of an entry to $\bv{x}^T \bv{A} \bv{x}$ without necessarily reading that entry, allowing for reduced sample complexity.

The logic is particularly clean when $\bv{A}$ is binary. In this case, we can restrict to looking at the principal submatrix corresponding to $i \in [n]$ for which $\bv{A}_{ii} = 1$.  The remainder of the matrix is all zeros. For such $i$, we have $\norm{\bv b_i}_2 = 1$, and thus if $\bv A_{ij} = 1$ and $\bv{A}_{jk} = 1$, we can conclude that $\bv{A}_{ik} = 1$. This implies that, up to a permutation of the row/column indices, $\bv{A}$ is a block diagonal matrix, with $\rank(\bv{A})$ blocks of all ones, corresponding to groups of indices $S_k$ with $\bv{b}_i = \bv{b}_j$ for all $i,j \in S_k$. The eigenvalues of this matrix are exactly the sizes of these blocks, and to recover a spectral approximation, it suffices to recover the set $S_k$ corresponding to any block of size $\ge \epsilon n$. 

To do so, we employ a weak notion of expanders \cite{pippenger1987sorting,wigderson1993expanders}, that requires that any two subsets of vertices both of size $\ge \epsilon n$, are connected by at least one edge. Importantly, such expanders can be constructed with $\widetilde O(n/\epsilon)$ edges -- beating Ramanujan graphs, which would require $\Omega(n/\epsilon^2)$ edges, but satisfy much stronger notions of edge discrepancy. Further, if we consider  the graph whose edges are in the intersection of those in the expander graph and of the non-zero entries in $\bv{A}$, we can show that  any block of ones in $\bv{A}$ of size $\Omega( \epsilon n)$ will correspond to a large $\Omega(\epsilon n)$ sized connected component in this graph. We can form this graph by querying $\bv{A}$ at $\widetilde O(n/\epsilon)$ positions (the edges in the expander) and then computing these large connected components to recover a spectral approximation. Each sparse connected component becomes a dense block of all ones in our approximation, avoiding the sparsity lower bound  of Theorem \ref{theorem:psd_deterministic_lower}. Noting that we can extend this logic to matrices with entries in $\{-1,0,1\}$, yields the main result of this section, Theorem \ref{th:binary query complexity}.

It is not hard to see that $\widetilde O(n/\epsilon)$ is optimal, even for binary PSD matrices (Theorem \ref{thm: lower bound for PSD matrices}). By Tur\'{a}n's theorem, any sample set of size $o(n/\epsilon)$ does not read any entries in some principal submatrix of size $\epsilon n \times \epsilon n$. By  letting $\bv{A}$ be the identity plus a block of all ones on this submatrix, we force our algorithm to incur $\epsilon n$ approximation error, since it cannot find this large block of ones.

\medskip

\noindent\textbf{Section Roadmap.} We formally prove the block diagonal structure of binary magnitude PSD matrices in Section \ref{sec:block}. We leverage this structure to give our improved spectral approximation algorithm for these matrices in Section \ref{sec:block2}. Finally, in Section \ref{sec:block3} we prove a nearly matching lower bound via Tur\'{a}n's theorem.

\subsection{Structure of Binary Magnitude PSD Matrices}\label{sec:block}

As discussed, our improved algorithm hinges on the fact that, up to a permutation of the rows and columns, binary magnitude PSD matrices are {block diagonal}.  Further,  each  block is itself a rank-$1$ matrix, consisting of two on-diagonal blocks of $1$'s and two off-diagonal blocks of $-1$'s. Formally, 
\begin{restatable}[Binary Magnitude PSD Matrices are Block Matrices]{lemma}{blockStructureBinaryPSD}
\label{lem:ternary_partition}
Let $\bv A \in \{-1,0,1\}^{n\times n}$ be PSD and let $\lambda_i$ for $i \in [n]$ be the eigenvalues of $\bv{A}$, with $\lambda_1 \geq \lambda_2 \geq \ldots \geq \lambda_p >0$ and $\lambda_{p+1} =\ldots = \lambda_n=0$. Then, there exists a permutation matrix $\bv{P}$ such that 
\begin{align*}
    \bv P\bv A\bv P^T=
    \begin{bmatrix}
    \bv A^{(1)} & \bv 0 & \ldots & \bv 0 & \bv {0} \\
    \bv 0 & \bv A^{(2)} & \ldots & \bv 0 & \bv{0} \\
    \vdots & \vdots & \ddots & \vdots & \bv 0 \\    
    \bv 0 & \bv 0 & \ldots & \bv A^{(p)} & \bv{0}\\
    \bv 0 & \bv 0 & \ldots & \bv 0 & \bv 0
    \end{bmatrix}.
\end{align*}
Further, 
Each $\bv{A}^{(i)} \in \{-1,1\}^{\lambda_i \times \lambda_i}$ is a $\lambda_i \times \lambda_i$ rank-1 PSD matrix. 
Letting $S_i$ denote the set of indices corresponding to the principal submatrix $\bv{A}^{(i)}$, which we call its \emph{support set}, $S_i$ can be partitioned in to two subsets $S_{i1}$ and $S_{i2}$ such that, if we let $\bv{v}_i(j) = 1$ for $j \in S_{i1}$ and $\bv{v}_i(j) = -1$ for $j \in S_{i2}$, then $\bv{v}_1,\ldots,\bv{v}_p$ are orthogonal and $\bv{A} = \sum_{i=1}^p \bv{v}_i \bv{v}_i^T$.
\end{restatable}
\begin{proof}
Since $\bv A$ is PSD, there exists $\bv B \in \R^{n\times n}$ with rows $\bv B_1,\ldots, \bv B_n \in \R^n$ such that $\bv A = \bv B \bv B^T$. Note that either $\bv{A}_{ii} = 0$ or $\bv{A}_{ii} = 1$ for all $i$ since $\bv{A}_{ii} = \norm{\bv B_i}_2^2 \ge 0$.
Also for any $i,j \in [n]$, if $\bv A_{ij}=1$ or $\bv{A}_{ij}=-1$, then $\bv A_{ii} = \bv A_{jj} = 1$. Otherwise, if $\bv{A}_{ii}$ or $\bv A_{jj}$ were $0$, we would have $\norm{\bv{B}_i}_2 = 0$ or $\norm{\bv B_j}_2 = 0$ and thus $\bv A_{ij} = \bv B_i \bv B_j^T = 0$.

Thus, if $\bv A_{ij}=1$, we have $\bv{A}_{ii} = \|\bv B_i\|_2 = 1$, $\bv{A}_{jj} = \|\bv B_j \|_2=1$, and $\bv B_i^T \bv B_j=1$. I.e., we must have $\bv B_i = \bv B_j$. Similarly, if $\bv{A}_{ij}=-1$, we have $\|\bv{B}_i\|_2=\| \bv{B}_j\|_2=1$ and $\bv B_i^T \bv B_j = -1$ and so $\bv{B}_i=-\bv{B}_j$.  Finally, if $\bv A_{ij} = 0$, then $\bv B_i \perp \bv B_j$. This implies that the indices $[n]$ can be grouped into disjoint subsets $S_1, S_2,\ldots, S_p$ ($S_i \subset [n]$) such that for any $S_i$ and any $j,k \in S_i$, either $\bv{B}_i=\bv{B}_j$ or $\bv{B}_i=-\bv{B}_j$. Also for any $k \in S_i$ and $\ell \in S_j$ with $i \neq j$, we must have $\bv{B}_k \perp \bv{B}_\ell$. Label these subsets such that $|S_1| \geq |S_2| \geq \ldots \geq |S_p|$. Now observe that any subset $S_i$ can be further divided into two disjoint subsets $S_{i1}$ and $S_{i2}$ such that if $j,k \in S_{i1}$ or $j,k \in S_{i2}$, $\bv{B}_j=\bv{B}_k$ but if $j \in S_{i1}$ and $k \in S_{i2}$ (or vice-versa), $\bv{B}_{j}=-\bv{B}_k$. Thus, the principal submatrix corresponding to $S_i$, which we label $\bv{A}^{(i)}$ is a rank-1 block matrix with two blocks of all ones on the principal submatrices corresponding to $S_{i1}$ and $S_{i2}$ and $-1$'s on all remaining entries.

For any $S_i$, let $\bv v_i \in \{0,1\}^n$ be a vector such that $\bv v_i(j)=1$ if $j \in S_{i1}$, $\bv v_i(j)=-1$ if $j \in S_{i2}$ and $\bv v_i(j)=0$ otherwise. Then, observe that $\bv A \bv v_i=\lvert S_i \rvert \bv v_i$. Thus, each $\lvert S_i \rvert$ is an eigenvalue of $\bv A$ and each $\bv{v}_i$ is an eigenvector. Since $S_1,\ldots,S_p$ has disjoint supports, $\bv{v}_1,\ldots,\bv{v}_p$ are orthogonal. Further, $\tr(\bv{A}) = \sum_{i=1}^n \lambda_i = \sum_{i=1}^p |S_i|$ and thus, all other eigenvalues of $\bv{A}$ are $0$. Thus, we can write $\bv{A} = \sum_{i=1}^p \bv{v}_i \bv{v}_i^T$. This concludes the lemma. 
\end{proof}

\subsection{Improved Spectral Approximation of Binary Magnitude PSD matrices}\label{sec:block2}

Given Lemma \ref{lem:ternary_partition}, the key idea to efficiently recovering a spectral approximation to a binary magnitude PSD matrix $\bv{A}$ is that we do not need to read very many entries in each block $\bv{A}^{(i)}$ to identify the block. We just need enough to identify a large fraction of the indices in the support set $S_i$. 

To ensure that our sampling is able to do so, we will sample according to the edges of a certain type of weak expander graph, which ensures that 
any two vertex sets of large enough size are connected. I.e., that we read at least one entry in any large enough off-diagonal block of our matrix. 

\begin{definition}[$\epsilon n$-expander graphs \cite{wigderson1993expanders, pippenger1987sorting}] 
For any $\epsilon \in \R$, an $\epsilon n$-expander graph is any undirected graph on $n$ vertices such that any two disjoint subsets of vertices containing at least $\epsilon n$ vertices each are joined by an edge. 
\end{definition}

Moreover, it is not hard to show that $\epsilon n$-expanders with just $ O(n\log n/\epsilon)$ edges exist.
\begin{fact}\label{fact:eps n expanders}\cite{wigderson1993expanders}
A random $d = O \left (\frac{\log n}{\epsilon} \right )$ regular graph is an $\epsilon n$-expander with high probability.
\end{fact}

We note that while a spectral expander graph, such as a Ramanujan graph, as used to achieve Theorem \ref{thm:PSD_Quad}, will also be an $\epsilon n$-expander by the expander mixing lemma \cite{chung2002sparse}, such graphs require $O(n/\epsilon^2)$ edges, as opposed to the $\widetilde O(n/\epsilon)$ given by Fact \ref{fact:eps n expanders}. Further, while Fact \ref{fact:eps n expanders} is not constructive, in \cite{wigderson1993expanders} an explicit polynomial time construction of $\epsilon n$-expander graphs is shown, with just an $n^{o(1)}$ loss. This result can be plugged directly into our algorithms to give a polynomial time constructible deterministic spectral approximation algorithm with $O(n^{1+o(1)}/\epsilon)$ sample complexity. 
\begin{fact}[Constructive  $\epsilon n$-expanders \cite{wigderson1993expanders}]\label{cor: eps n expanders}
There is a polynomial time algorithm that, given an integer $n$ and $\epsilon \in (0,1)$, constructs an $\epsilon n$-expanding graph on $n$ vertices with maximum degree $\frac{n^{o(1)}}{\epsilon}$. 
\end{fact}

We now prove our main technical result, which shows that if we read the entries of a binary magnitude PSD matrix $\bv A$ according to an $\epsilon n$-expander graph, then we can approximately recover all blocks of this matrix by considering the \emph{connected components} of the expander graph restricted to edges where $\bv A$ is nonzero. 

\begin{theorem}[Approximation via Sampled Connected Components]\label{thm: epsn psd constructive}
Let $\bv A \in \{-1,0,1\}^{n \times n}$ be a PSD matrix with eigenvalues $\lambda_1\geq \lambda_2 \geq \ldots \geq \lambda_p > 0$ and $\lambda_{p+1}=\ldots=\lambda_n=0$. Let $G$ be an $\epsilon/6 \cdot n$-expanding graph on $n$ vertices with adjacency matrix $\bv{B}_G$. Let $\bv{A}_G$ be the binary adjacency matrix with $(\bv A_G)_{ij} = 1$ whenever $(\bv{B}_G)_{ij} = 1$ and $|\bv{A}_{ij}| = 1$.
Let $\bar G$ be the graph whose adjacency matrix is $\bv{A}_G$. Let $S_1,\ldots,S_p$ denote the support sets of $\bv{A}$'s blocks (as defined in Lemma \ref{lem:ternary_partition}) with $|S_i| = \lambda_i$, and let $\bar G_{S_i}$ be the induced subgraph of $\bar G$ restricted to $S_i$. Finally, let $C_i \subset [n]$ be the largest connected component of $\bar G_{S_i}$. We have:
\begin{align*}
    \lambda_i - \epsilon/2 \cdot n < |C_i| \leq \lambda_i.
\end{align*}
\end{theorem}

\begin{proof}
First observe that since $|S_i| = \lambda_i$, we trivially have that $|C_i| \le \lambda_i$. Thus, it remains to show that $|C_i| > \lambda_i - \epsilon/2 \cdot n$. This holds vacuously if $\lambda_i \le \epsilon/2 \cdot n$. Thus, we focus our attention to $\lambda_i \ge  \epsilon/2 \cdot n$.
For contradiction assume $|C_i| \le \lambda_i - \epsilon/2 \cdot n$. 
Let $M_1,\ldots,M_r \subseteq S_i$ denote the connected components of $\bar G_{S_i}$ such that $\left|M_1\right| \geq  \ldots \geq \left|M_r\right|$. So $|C_i| = M_1$. Observe that since all entries in $\bv{A}$ in the principal submatrix indexed by $S_i$ (i.e., in $\bv{A}^{(i)}$ from Lemma \ref{lem:ternary_partition}) have magnitude $1$, $\bar G_{S_i} = G_{S_i}$.
Consider two cases.

\medskip

\noindent \textbf{Case 1: $|C_i| \geq \epsilon/6 \cdot n$.} From our assumption: $\lambda_i - |C_i| \ge \epsilon/2 \cdot  n$ which implies that $\sum_{j=2}^r |M_j| \ge \epsilon/2 \cdot n$. 
Let $C' = \bigcup_{j=2}^r M_j$ so  $|C'| \ge \epsilon/2 \cdot n$. Since $M_1,\ldots,M_r$ are disconnected, there are no edges from $C'$ to $C_i$ in $\bar G$, and thus no edges in $G$, since $\bar G_{S_i} = G_{S_i}$. However, this contradicts the fact that $G$ is an $\epsilon/6 \cdot n$ expander and thus must have at least one edge between any two set of vertices of size at least $\epsilon/6 \cdot n$ (Definition \ref{fact:eps n expanders}).
Thus we have a contradiction and our assumption $\lambda_i - |C_i| \ge  \epsilon n$ is incorrect. So we must have $|C_i| > \lambda_i - \epsilon n$ as needed.

\medskip

\noindent \textbf{Case 2: $|C_i| < \epsilon/6 \cdot n$.} Since $|C_i| < \epsilon/6 \cdot n$, we have $\max_{i\in [r]} |M_j| < \epsilon/6 \cdot n$. Consider  partitioning these connected components into two sets, $T_1$ and $T_2$. Let $V_1,V_2$ be the vertex sets corresponding to these two partitions, i.e., $V_1 = \bigcup_{M_i \in T_1} M_i$ and $V_2 = \bigcup_{M_i \in T_2} M_i$. Pick $T_1$ and $T_2$ to minimize $||V_1|-|V_2||$. Since $\max_{i\in [r]} |M_j| < \epsilon/6 \cdot n$, for this optimal $T_1,T_2$, we must have $||V_1|-|V_2|| < \epsilon/6 \cdot n$.

 Since $V_1$ and $V_2$ are disconnected in $\bar G_{S_i}$, they must be disconnected in $G_{S_i}$. Thus, we must have either $|V_1| < \epsilon/6 \cdot  n$ or $|V_2| < \epsilon/6 \cdot n$ since any two subsets of vertices with size $\ge \epsilon/6 \cdot n$ must have an edge between them in $G$ due to it being an $\epsilon/6 \cdot n$-expander.
 Assume w.l.o.g. that $|V_2| < \epsilon/6 \cdot n$. Then since $|V_1| + |V_2| = |S_i| = \lambda_i \ge \epsilon/2 \cdot n$, we must have $|V_1| > \epsilon/3 \cdot n$. However, this means that $||V_1| -|V_2|| > \epsilon/6 \cdot n$, which contradicts the fact that we picked $T_1,T_2$ such that $||V_1| -|V_2||  < \epsilon/6 \cdot n$.
 
\end{proof}

We now present our deterministic algorithm for spectral approximation of binary magnitude PSD matrices, based on Theorem \ref{thm: epsn psd constructive}. The key idea is to compute all connected components of $\bar G$ -- the graph whose edges lie in the intersection of the edges of $G$ and the non-zero entries of $\bv A$. Any large enough block in $\bv{A}$ (see Lemma \ref{lem:ternary_partition}) will correspond to a large connected component in this graph and can thus be identified. 

\begin{algorithm}[H] 
\caption{Deterministic spectral approximator using expanders}
\label{alg: expander DSA}
\begin{algorithmic}[1]
\Require{PSD matrix $\bv A \in \{-1, 0,1\}^{n\times n}$ and $\epsilon \in (0,1)$.}
\State Construct an $\epsilon/6 \cdot  n$-expanding graph $G$ on $n$ nodes with adjacency matrix $\bv{B}_G$.
\State Let $\bv A_G$ be an $n \times n$ matrix such that $( \bv A_G)_{ij} =1$ if $(\bv{B}_G)_{ij} = 1$ and $|\bv A_{ij}| = 1$. Let $\bar G$ be the graph corresponding to $\bv A_G$. 
 \State Compute all connected components of $\bar G$,  $C_1,C_2,\ldots, C_r \subseteq [n]$.
\State Initialize $\widetilde{\bv A} = \bv 0^{n \times n}$. 
\For{$i=1,2,\ldots,r$}
    \If{$|C_i| >  \frac{\epsilon n}{2}$}
        \State Pick any $j \in C_i$ and let $\Atilde = \Atilde+\bv{A}_{j} \bv{A}_{j}^T$, where $\bv{A}_j$ is the $j^{th}$ column of $\bv{A}$.
    \EndIf
\EndFor
\State \Return{$\widetilde{\bv A}$}. 
\end{algorithmic}
\end{algorithm}
\noindent\textbf{Query complexity:} Observe that Algorithm \ref{alg: expander DSA} requires querying $\bv{A}$ at the locations where $\bv{B}_G$ is non-zero (to perform step 2). By Fact \ref{fact:eps n expanders}, there are $\epsilon/6 \cdot n$-expander graphs with just $O(\frac{n \log n}{\epsilon})$ edges, so this requires $O(\frac{n \log n}{\epsilon})$ queries.
Further, in a graph with $n$ nodes,there can be at most $2/\epsilon$  connected components of size $\ge \frac{\epsilon n}{2}$. Thus, line (7) requires total query complexity $\le \frac{2n}{\epsilon}$ to read one column of $\bv{A}$ corresponding to each large component identified. 

\binaryapprox*
\begin{proof}
We will show that Algorithm \ref{alg: expander DSA} satisfies the requirements of the theorem. We have already argued above that its query complexity is $O(\frac{n \log n}{\epsilon})$. 

By Lemma \ref{lem:ternary_partition}, we know that, up to a permutation, $\bv A$ is a block matrix with rank-$1$ blocks $\bv A^{(1)}, \bv A^{(2)}, \ldots, \bv A^{(k)}$ with support sets $S_1,\ldots,S_k$. Since $\bar G$ has edges only where $\bv{A}$ is non-zero, each connected component $C_i$ in $\bar G$ is entirely contained within one of these support sets. Further, if $|S_i| = \lambda_i \ge \epsilon n$, then by Theorem \ref{thm: epsn psd constructive}, there is exactly one connected component, call it $D_i$ (the largest connected component of $\bar G_{S_i}$) with vertices in $S_i$ and $|D_i| \ge \epsilon/2$. If $|S_i| = \lambda_i < \epsilon n$, then there it at most one such connected component by Theorem \ref{thm: epsn psd constructive} (there may be none).

So, Step 7 is triggered for exactly one connected component $D_i \subseteq S_i$ for each $S_i$ with $|S_i| = \lambda_i \ge \epsilon n$ and at most one connected component for all other $S_i$. If Line 7 is triggered, then we add $\bv{A}_j \bv{A}_j^T$ to $\bv {\widetilde A}$. Observe that $\bv{A}_j$ is a column for $j \in S_i$, and thus by Lemma  \ref{lem:ternary_partition}, $\bv{A}_j$ is supported only on $\bv{S}_i$. We have that either $\bv{A}_j(t) = 1$ on $S_{i1}$ and $\bv{A}_j(t) = -1$ on $S_{i2}$, or $\bv{A}_j(t) = -1$ on $S_{i1}$ and $\bv{A}_j(t) = 1$. That is, $\bv{A}_j$ is equal to either $\bv{v}_i$ or $-\bv{v}_i$ as defined in Lemma  \ref{lem:ternary_partition}.

So overall, letting $Z \subseteq \{S_1,\ldots,S_k\}$ be the set of blocks for which one connected component is recovered in Line 7, $\bv{\widetilde A} = \sum_{i \in Z} \bv{v}_i \bv{v}_i^T$. From Lemma \ref{lem:ternary_partition}, $\bv{A} = \sum_{i=1}^k \bv{v}_i \bv{v}_i^T$ and so $\bv{A} -  \bv{\widetilde A} = \sum_{i \notin Z} \bv{v}_i \bv{v}_i^T$. Since the $\bv{v}_i$ are orthogonal and since for $i \notin Z$ we must have $\lambda_i = |S_i| =  \norm{\bv{v}_i \bv{v}_i^T}_2 \le \epsilon n$, this gives that $\norm{\bv A - \bv{\widetilde A}}_2 \le \epsilon n$, completing the theorem.

\end{proof} 

\subsection{Lower Bounds for Binary PSD Matrix Approximation}\label{sec:block3}

We can prove that our $ O\left (\frac{n \log n}{\epsilon}\right )$ query complexity for binary magnitude matrices is optimal  up to a $\log n$ factor via Tur\'{a}n's Theorem, stated below. Our lower bound holds for the easier problem of eigenvalue approximation for PSD $\bv A \in \{0,1\}^{n \times n}$.
\begin{fact}[Turan's theorem \cite{turaan1941extremal, aigner1995turan}]\label{fact:turan's theorem}
Let $G$ be a graph on $n$ vertices that does not include a $(k+1)$-clique as a subgraph. Let $t(n,k)$ be the number of edges in $G$. Then
\begin{align*}
    t(n,k) \leq \frac{(k-1)n^2}{2k}.
\end{align*}
\end{fact}
Using Turan's theorem we can prove:
\lowerboundPSD*

\begin{proof}
Let $\mathcal{A}$ be a deterministic algorithm that approximates the eigenvalues of a binary PSD input matrix $\bv A$ to $\epsilon n$ additive error. Let $\bv A_0 = \bv I_n$, i.e., the identity matrix. Let $S \subset [n] \times [n]$ be the first $\frac{n}{\epsilon}$ off-diagonal entries read by $\mathcal{A}$ on input $\bv A$.  Since the choices of the algorithm are deterministic, these same entries will be the first $\frac{n}{\epsilon}$ entries queried by $\mathcal{A}$ from any input matrix which has ones on the diagonal and zeroes in all entries of $S$, even if the algorithm is adaptive.

Without loss of generality, we can assume that the entries in $S$ are above the diagonal, since reading any entry below the diagonal of an input matrix would be redundant due to symmetry.  Consider a graph G with vertex set $V = [n]$ and undirected edge set $E = S^c$, that is, the edge set contains all undirected edges $(i,j)$ such that $(i, j),(j, i) \not \in S$ for $i,j \in [n]$. Hence, for all $\epsilon \in (0, 1)$,
\begin{gather*}
    |E| = {n \choose 2} - \frac{n}{\epsilon} = \frac{n^2}{2} - \left(\frac{n}{2} + \frac{n}{\epsilon}\right) > \frac{n^2}{2} - \frac{2n}{\epsilon} = \frac{(\frac{\epsilon n}{4} - 1)n^2}{2 \cdot \frac{\epsilon n}{4}},
\end{gather*}
where the first inequality is true since $\frac{n}{2} < \frac{n}{\epsilon}$. Therefore, if we set $k = \frac{\epsilon n}{4}$, we conclude by the contrapositive of Fact \ref{fact:turan's theorem} that there exists a clique of size $\frac{\epsilon n}{4}$ in $G$.  Let $T \subset E$ denote the set of edges which forms this clique in $G$.

Construct $\bv A_1$ such that $[\bv A_1]_{ij} = [\bv A_0]_{ij}$ for all $i,j$ such that $(i, j) \not \in T$ and $[\bv A_1]_{ij} = [\bv A_1]_{ji} = 1$ for all $i,j$ such that $(i,j) \in T$.  Then, there is a principal submatrix containing all ones in $\bv A_1$ where the off-diagonal entries correspond to edges in $T$ along with the diagonal entries which are all ones by construction.  Hence, $\bv A_1$ contains a non-zero principal submatrix of size $\frac{\epsilon n}{4}$, and so, it has an eigenvalue that is at least $\frac{\epsilon n}{4}$.  Both $\bv A_0$ and $\bv A_1$ have identical values on the diagonal and on the first $\frac{n}{\epsilon}$ off-diagonal entries read by $\mathcal{A}$ (i.e., the entries in $S$). Therefore, the $\frac{n}{\epsilon}$ entries of $\bv A_0$ and $\bv A_1$ queried by $\mathcal{A}$ are identical, and hence $\mathcal{A}$ cannot distinguish these two matrices, despite the fact their maximum eigenvalue differs by $\frac{\epsilon n}{4}$.

After adjusting for constant factors, we conclude that any deterministic (possibly adaptive) algorithm that estimates the eigenvalues of a $\{0,1\}$-matrix to $\epsilon n$ additive error must observe $\Omega(\frac{n}{\epsilon})$ entries of the input matrix in the worst case.

\end{proof}

\section{Applications to Fast Deterministic Algorithms for Linear Algebra}\label{sec:applications}

We finally show how to leverage our universal sparsifiers to design the first $o(n^{\omega})$ time deterministic algorithms for several problems related to singular value and vector approximation. As discussed in Section \ref{sec:det_intro}, throughout this section, runtimes are stated assuming that we have already constructed a deterministic sampling matrix $\bv{S}$ satisfying the universal sparsification guarantees of Theorem \ref{thm:PSD_Quad} or Theorem \ref{thm: general matrix eps4 bound}. When $\bv S$ is the adjacency matrix of a Ramanujan graph, this can be done in $\Tilde{O}(n/\poly(\epsilon))$ time -- see Section \ref{sec:det_intro} and Section \ref{section:notation} for further discussion. 

In Section \ref{sec:singvalapprox}, we consider approximating the singular values of $\bv{A}$ to additive error $\pm \epsilon  \max(n,\norm{\bv A}_1)$. Observe that if we deterministically compute $\bv{\widetilde A} = \bv A \circ \bv S$ with $\norm{\bv A - \bv{\widetilde A}}_2 \le \epsilon  \max(n,\norm{\bv A}_1)$ via Theorem \ref{thm: general matrix eps4 bound}, by Weyl's inequality, all singular values of $\bv{\widetilde A}$ approximate those of $\bv{A}$ up to additive error $\pm \epsilon  \max(n,\norm{\bv A}_1)$. Thus we can focus on approximating $\bv{\widetilde A}$'s singular values.

 To do so efficiently, we will use the fact that $\bv{\widetilde A}$ is sparse and so admits very fast $\widetilde O\left (\frac{n}{\epsilon^4} \right )$ time matrix-vector products. Focusing on the top singular value, it is well known that the power method, initialized with a vector that has non-negligible (i.e.,  $\Omega(1/\poly(n))$) inner product with the top singular vector $\bv{v}_1$ of $\bv{\widetilde A}$, converges to a $(1 \pm \epsilon)$ relative error approximation in  $O(\log n/\epsilon)$ iterations. I.e., the method outputs $\bv{\widetilde v}_1$ with $(1-\epsilon) \sigma_1(\bv{\widetilde A}) \le \norm{\bv{\widetilde A} \bv{\widetilde v}_1}_2 \le \sigma_1(\bv{\widetilde  A})$. Each iteration requires a single matrix-vector product with $\bv{\widetilde A}$, yielding total runtime $\widetilde O\left( \frac{n}{\epsilon^5}\right )$. Since $\bv{v}_1$ is unknown, one typically  initializes the power method with a random vector, which has non-negligible inner product with $\bv{v}_1$ with high probability. To derandomize this approach, we simply try $n$ orthogonal starting vectors, e.g., the standard basis vectors, of which at least one must have large inner product with $\bv v_1$. This simple brute force approach yields total runtime $\widetilde O\left( \frac{n^2}{\epsilon^5}\right )$, which, for large enough $\epsilon$, is  $o(n^\omega)$.

To approximate more than just the top singular value, one would typically use a block power or Krylov method. However, derandomization is difficult here -- unlike in the single vector case, it is not clear how to pick a small set of deterministic starting blocks for which at least one gives a good approximation. Thus, we instead use  \emph{deflation} \cite{Saad:2011wv,Allen-Zhu:2016vf}. We first use our deterministic power method to compute $\bv{\widetilde v}_1$ with $\norm{\bv{\widetilde A} \bv{\widetilde v}_1}_2 \approx \sigma_1(\bv{\widetilde A})$. We then run the method again on the deflated matrix $\bv{\widetilde A}(\bv I - \bv{\widetilde v}_1 \bv{\widetilde v}_1^T)$, which we can still multiply by in $O\left (\frac{n}{\epsilon^4} \right )$ time. This second run outputs $\bv{\widetilde v}_2$ which is orthogonal to $\bv{\widetilde v}_1$ and which has $ \norm{\bv{\widetilde A}\bv{\widetilde v}_2}_2 \approx \sigma_1(\bv{\widetilde A}(\bv I - \bv{\widetilde v}_1 \bv{\widetilde v}_1^T))$. Since $\bv{\widetilde v}_1$ gives a good approximation to $\sigma_1(\bv{\widetilde A})$, we  then argue that $\sigma_1(\bv{\widetilde A}(\bv I - \bv{\widetilde v}_1 \bv{\widetilde v}_1^T)) \approx \sigma_2(\bv{\widetilde A})$ and so $\norm{\bv{\widetilde A} \bv{\widetilde v}_2}_2 \approx \sigma_2(\bv{\widetilde A})$. We repeat this process to approximate all large singular values of $\bv{\widetilde A}$. The key challenge is bounding the accumulation of error that occurs at each step.

In Section \ref{sec:psdtesting} we apply the above derandomized power method approach to the problem of testing whether $\bv{A}$ is PSD or has at least on large negative eigenvalue $< -\epsilon \max(n,\norm{\bv A}_1)$. We reduce this problem to approximating the top singular value of a shifted version of $\bv{A}$.

Finally, in Section \ref{sec:highaccuracy}, we consider approximating $\sigma_1(\bv{A})$ to high accuracy -- e.g., to relative error $(1 \pm \epsilon)$ with $\epsilon = 1/\poly(n)$, under the assumption that $\bv{\sigma}_1(\bv{A}) \ge \alpha \max(n,\norm{\bv A}_1)$ for some $\alpha$. Here, one cannot simply approximate $\sigma_1(\bv{\widetilde A})$ in place of $\sigma_1(\bv A)$, since the error due to sparsification is too large. Instead, we use our deflation approach to compute a `coarse' approximate set of top singular vectors for $\bv{\widetilde A}$. We then use these singular vectors to initialize a block Krylov method run on $\bv{A}$ itself, which we argue converges in $O\left(\frac{\log(n/\epsilon)}{\poly(\alpha)}\right)$ iterations, giving total run time $\widetilde O\left(\frac{n^2 \log(1/\epsilon)}{\poly(\alpha)}\right)$. The key idea in showing convergence is to argue that since $\bv{\sigma}_1(\bv{A}) \ge \alpha \max(n,\norm{\bv A}_1)$ by assumption, we have $\bv{\sigma}_{2/\alpha}(\bv{A}) \le \alpha/2  \norm{\bv A}_1 \le \sigma_1(\bv{A})/2$. So, there must be at least some singular value $\sigma_p(\bv{A})$ with $p \le 2/\alpha$ that is larger than $\sigma_{p+1}(\bv A)$ by $\alpha^2/4 \cdot \max(n,\norm{\bv A}_1)$. The presence of this spectral gap allows us to (1) argue that the first $p$ approximate singular vectors output by our deflation approach applied to $\bv{\widetilde A}$ have non-neglible inner product with the true top $k$ singular vectors of $\bv{A}$ and (2) that the block Krylov method initialized with these vectors converges rapidly to a high accuracy approximation to $\sigma_1(\bv{A})$, via known gap-dependent convergence bounds \cite{musco2015randomized}.

\subsection{Deterministic Singular Value Approximation}\label{sec:singvalapprox}

We start by giving  a deterministic algorithm to approximate all  singular values of a symmetric bounded entry matrix $\Ab$ up to error $\pm \epsilon \max(n,\|\Ab \|_1)$ in $\widetilde O(n^2/\poly(\epsilon))$ time. Our main result appears as Theorem \ref{thm:approxSVD}.

As discussed,  we first deterministically compute a  sparsified matrix $\Atilde$ from $\Ab$ with $\norm{\bv A - \Atilde}_2 \le \epsilon \max(n,\norm{\bv A}_1)$. We then use a simple brute-force derandomization of the classical power method to approximate to the top singular value of $\Atilde$ (Lemma \ref{lem:top_eig}).  To approximate the rest of $\bv{\widetilde A}$'s singular values, we  repeatedly deflate the matrix and compute the top singular value of the deflated matrix (Lemma \ref{lem:all_sing_val}). By Weyl's inequality, the singular values of $\Atilde$ approximate those of $\bv{A}$ up to $\pm \epsilon \max(n,\| \Ab\|_1)$ error. Thus, we can argue that our approximations will approximate to the singular values of $\Ab$ itself, yielding Theorem \ref{thm:approxSVD}.
We first present the derandomized power method that will be applied to $\bv{\widetilde A}$  in Algorithm \ref{alg:pow_it}.

\begin{algorithm}[H] 
\caption{Deterministic Power Method for Largest Singular Value Estimation}
\label{alg:pow_it}
\begin{algorithmic}[1]
\Require{$\bv A \in \mathbb{R}^{n\times n}$, error parameter $\epsilon \in (0,1)$}
\State Initialize $\sigmatilde=0$, $t=\frac{c \log (n/\epsilon)}{\epsilon}$ where $c$ is a sufficiently large constant.
\For{$k=1\to n$}
    \State $\bv{y}_k=\bv{e}_{k}$ where $\bv{e}_k$ is the k\textsuperscript{th} standard basis vector.
    \For{$j=1 \to t$}
        \State $\bv{y}_k \gets (\bv{A}\bv{A}^T)\bv{y}_k$.
        \State $\bv{y}_k \gets \frac{\bv{y}_k}{\|\bv{y}_k\|_2}$.
    \EndFor
    \If{$\|\Ab^T\bv{y}_k\|_2 \geq \sigmatilde$}
        \State $\bv{z} \gets \bv{y}_k$.
        \State $\sigmatilde \gets \|\Ab^T\bv{z}\|_2$.
    \EndIf
\EndFor
\State \Return {\ $\bv{z}$}
\end{algorithmic}
\end{algorithm}

Let $\bv{A} \in \mathbb{R}^{n \times n}$ be a matrix with its SVD given by $\Ab=\bv{U}\bv{\Sigma} \bv{V}^T$ where $\bv{U}, \bv{V} \in \R^{n \times n}$ are matrices with orthonormal columns containing the left and right singular vectors respectively and $\bv{\Sigma}$ is a diagonal matrix containing the  singular values $\sigma_1(\bv{A}) \geq \sigma_2(\bv{A}) \geq \ldots \geq \sigma_n(\bv{A}) \geq 0$. Observe that any vector $\bv{y} \in \mathbb{R}^{n}$ can be written as $\bv{y}=\sum_{i=1}^n c_i \bv{u}_i$ where $c_i$ are scalar coefficients and $\bv{u}_i$ are the columns of $\bv{U}$. Then, we have the following well-known result which shows that  when the starting unit vector $\bv{y}$ has a large enough inner product with $\bv{u}_1$ i.e., $|c_1| \geq \frac{1}{\sqrt{n}}$, power iterations on $\bv{A}\bv{A}^T$ converge to a $1-\epsilon$ approximation to $\sigma_1(\bv{A})$ within $\widetilde O(1/\epsilon)$ iterations. We will  leverage this lemma to prove the correctness and runtime of our deterministic power method, Algorithm~\ref{alg:pow_it}.
 
\begin{lemma}[Power iterations -- gap independent bound]\label{lem:power_it}
 Let $\bv{A} \in \mathbb{R}^{n \times n}$ be a matrix with largest singular value $\sigma_1(\bv{A})$ and left singular vectors $\bv{u}_1,\ldots,\bv{u}_n$. Let $\bv{y}=\sum_{i=1}^n c_i \bv{u}_i$ be a unit vector where $|c_1| \geq \frac{1}{\sqrt{n}}$. Then, for  $t = O\big( \frac{\log (n/\epsilon)}{\epsilon}\big)$, if we set $\bv{z}=\frac{(\bv{A}\bv{A}^T)^t\bv{y}}{\|(\bv{A}\bv{A}^T)^t\bv{y} \|_2}$, 
 \begin{align*}
     \|\bv{A}^T\bv{z} \|_2 \geq (1-\epsilon)\sigma_1(\bv{A}).
 \end{align*}
\end{lemma}

\begin{proof} We prove Lemma \ref{lem:power_it} for completeness.
Let $\bv{r}=(\bv{A}\bv{A}^T)^t\bv{y}$. Then, $\bv{z}=\frac{\bv{r}}{\|\bv{r} \|_2}$. First observe that $\bv{A}\bv{A}^T=\sum_{i=1}^n \sigma^2_i(\bv{A})\bv{u}_i\bv{u}_i^T$. Thus, $\bv{r}=\sum_{i=1}^n c_i (\sigma_i(\bv{A}))^{2t} \bv{u}_i$. Let $k \in [n]$ be the smallest $k$ such that $\sigma^2_k(\bv{A}) \leq (1-\epsilon)\sigma^2_1(\bv{A})$. Observe that if no such $k$ exists, then for any unit vector $\bv z$, $\norm{\bv A^T \bv z}_2 \ge \sigma_n(\bv A) > (1-\epsilon) \sigma_1(\bv A)$, and thus the lemma holds trivially.

Let $\bv{r}=\bv{r}_1+\bv{r}_2$ where $\bv{r}_1=\sum_{i=1}^{k-1}c_i\sigma^{2t}_i(\bv{A})\bv{u}_i$ and $\bv{r}_2=\sum_{i=k}^{n}c_i\sigma^{2t}_i(\bv{A})\bv{u}_i$. Then:
\begin{align*}
    \frac{\|\bv{r}_2\|_2^2}{\|\bv{r}_1\|_2^2} =\frac{\sum_{i=k}^{n}c^2_i\sigma^{4t}_i(\bv{A})}{\sum_{i=1}^{k-1}c^2_i\sigma^{4t}_i(\bv{A})} \leq \frac{\sum_{i=k}^{n}c^2_i\sigma^{4t}_i(\bv{A})}{c^2_1\sigma^{4t}_1(\bv{A})} \leq \sum_{i=k}^{n} \bigg(\frac{c_i}{c_1}\bigg)^2 \bigg(\frac{\sigma_i(\bv{A})}{\sigma_1(\bv{A})} \bigg)^{4t}.
\end{align*}
Since $c_i \leq 1$ for all $i \in [n]$, and since by assumption $|c_1| \ge \frac{1}{\sqrt{n}}$, $\left |\frac{c_i}{c_1} \right | \leq \sqrt{n}$ for all $i \in [n]$. Also, for any $i \geq k$, $\big(\frac{\sigma_i(\bv{A})}{\sigma_1(\bv{A})} \big)^2 \leq (1-\epsilon)$ and $(1-\epsilon)^{2t} \leq \delta$ for $t \geq O\big(\frac{\log (1/\delta)}{\epsilon} \big)$. Thus, we have $\frac{\|\bv{r}_2\|_2^2}{\|\bv{r}_1\|_2^2} \leq n^2 \delta$. Setting $\delta=\frac{\epsilon^2}{n^2}$, we get $\frac{\|\bv{r}_2\|_2^2}{\|\bv{r}_1\|_2^2} \leq \epsilon^2$.

Next, observe that $\|\bv{r} \|^2_2=\|\bv{r}_1 \|^2_2+\|\bv{r}_2 \|^2_2$ which implies that $\| \bv{r}\|_2^2 \leq (1+\epsilon^2) \| \bv{r}_1\|_2^2$ or $\| \bv{r}_1\|_2^2 \geq (1-\epsilon)\| \bv{r}\|_2^2$. Thus, using the fact that $\bv{r}_1^T\bv{A}\bv{A}^T\bv{r}_2=0$ and $\bv{r}_2^T\bv{A}\bv{A}^T\bv{r}_2 \geq 0$, and the fact that $\sigma_{i}^2(\bv A) \ge (1-\epsilon) \sigma_1^2(\bv A)$ for $i < k$, we get:

\begin{align*}   \bv{r}^T\bv{A}\bv{A}^T\bv{r}=\bv{r}_1^T\bv{A}\bv{A}^T\bv{r}_1+\bv{r}_2^T\bv{A}\bv{A}^T\bv{r}_2 
&\geq  \bv{r}_1^T\bv{A}\bv{A}^T\bv{r}_1 \\
&\ge \sigma_{k-1}^2(\bv A) \cdot \norm{\bv{r}_1}_2^2\\
&\ge (1-\epsilon)^2 \cdot \sigma_1^2  (\bv A) \cdot \norm{\bv{r}}_2^2.
\end{align*}
Finally dividing both sides of the above equation by $\|\bv{r} \|^2_2$, we get $\bv{z}^T\bv{A}\bv{A}^T\bv{z} \geq (1-\epsilon)^2\sigma^2_1(\bv{A})$. Taking square root on both sides gives us, $\|\bv{A}^T\bv{z} \|_2 \geq (1-\epsilon)\sigma_1(\bv{A})$.
\end{proof}
Next we prove the correctness of Algorithm~\ref{alg:pow_it} using Lemma~\ref{lem:power_it}. The idea is to run power iterations using all $n$ basis vectors as starting vectors. At least one of these vectors will have high inner product with $\bv{u}_1$, and so so this starting vector will give us a good approximation to $\sigma_1(\bv{A})$ by Lemma~\ref{lem:power_it}. We also prove that the squared singular values of the deflated matrix $\bv{A}-\bv{z}\bv{z}^T\bv{A}$ are close to the squared singular values of the original matrix up to an additive error of $\epsilon \cdot \sigma^2_1(\bv{A})$. This will be used to as part of the  correctness proof for our main algorithm for singular value approximation.
\begin{lemma}\label{lem:top_eig}
Let $\bv{z}$ be the output of Algorithm~\ref{alg:pow_it} for some matrix $\bv A \in \mathbb{R}^{n\times n}$ with singular values $\sigma_1(\bv{
A}) \geq \sigma_2(\bv{
A}) \geq \ldots \geq \sigma_n(\bv{
A})$ and error parameter $\epsilon \in (0,1)$ as input. Then:
\begin{equation*}
    (1-\epsilon)\sigma_1(\bv{A}) \leq \|\bv{A}^T\bv{z} \|_2 \leq \sigma_1(\bv{A}).
\end{equation*}
Further, for any $i \in [n]$:
\begin{equation*}
    \sigma_i^2 (\Ab-\zb\zb^T \Ab) \leq \sigma_{i+1}^2(\Ab) + \epsilon \sigma^2_1(\Ab).
\end{equation*}
\end{lemma}
\begin{proof}
    Let $\bv u_1$ be the left singular vector corresponding to the largest singular value of $\bv{A}$. Observe that since $\norm{\bv u_1}_2 = 1$, there exists at least one standard basis vector $\bv{e}_k$ such that $|\bv{e}_k^T\cdot \bv{u}_1| \geq \frac{1}{\sqrt{n}}$. 
     After $t$ iterations of the inner loop of Algorithm~\ref{alg:pow_it}, we have $\bv{y}_k=\frac{(\bv{A}\bv{A}^T)^t\bv{e}_k}{\|(\bv{A}\bv{A}^T)^t\bv{e}_k \|_2}$. Using Lemma~\ref{lem:power_it}, we thus have $\|\bv{A}^T\bv{y}_k\|_2^2 \geq (1-\epsilon)\sigma_1(\bv{A})$. Also since $\bv{y}_k$ is a unit vector, we trivially have $\|\Ab^T\bv{y}_k\|_2 \leq \sigma_1(\bv{A})$ for all $k$. Since for $\bv z$ output by Algorithm \ref{alg:pow_it}, $\|\Ab^T\bv{z}\|_2=\max_{k} \|\Ab^T\bv{y}_k\|_2$, we thus have:
    \begin{equation}\label{Eq:sing_val}
        (1-\epsilon)\sigma_1(\bv{A}) \leq \| \bv{A}^T\bv{z}\|_2 \leq \sigma_1(\bv{A}),
    \end{equation}
    giving our first error bound.

We now prove the second bound on $\sigma_i^2(\bv{A} - \bv{zz}^T \bv A)$. Using the Pythagorean theorem, $\|\bv{A}-\bv{z}\bv{z}^T \Ab\|_F^2 = \|\Ab\|_F^2 - \|\zb\zb^T \Ab\|_F^2$. We have $\|\zb\zb^T \Ab\|_F^2 = \tr(\bv{zz}^T \bv{AA}^T \bv{zz}^T) = \bv{z}^T \bv{AA}^T \bv{z} = \norm{\bv A^T \bv z}_2^2$. We thus have $\|\Ab-\zb\zb^T \Ab\|_F^2 = \norm{\bv A}_F^2 - \norm{\bv A^T \bv z}_2^2$.
Combining this with~\eqref{Eq:sing_val}:
    \begin{align}
        \|\Ab-\zb\zb^T \Ab\|_F^2 
        &\leq \|\Ab\|_F^2 - (1-\epsilon)^2 \sigma_1^2(\Ab) \notag\\
        &\le \|\Ab\|_F^2- \sigma_1^2(\Ab) +2\epsilon \sigma_1^2(\Ab) \notag \\
        \label{eq:A-zzA_F}
        &= \|\Ab-\Ab_1\|_F^2 + 2\epsilon\sigma_1^2(\Ab),
    \end{align}
    where $\bv{A}_1 = \bv{u}_1 \bv{u}_1^T \bv{A}$ is the best rank-$1$ approximation to $\bv{A}$ in the Frobenius norm, with  $\|\Ab-\Ab_1\|_F^2 = \sum_{i=2}^n \sigma_i^2(\Ab)$. So, we have:
    \begin{align*}
        \sum_{j=1}^n \sigma_j^2 (\Ab-\zb\zb^T \Ab)=\|\Ab-\zb\zb^T \Ab\|_F^2 &\leq \|\Ab-\Ab_1\|_F^2 + 2\epsilon \sigma^2_1(\Ab) = \sum_{j=2}^n \sigma_j^2(\Ab) + 2\epsilon \sigma^2_1(\Ab).
    \end{align*}
Simply subtracting $\sigma_n^2(\bv A - \bv{zz}^T\bv A)$ from the lefthand side and pulling $\sigma_i^2(\bv{A}-\bv{zz}^T)$ out of the summation for some $i \in [n]$ gives that
\begin{align}\label{eq:sum_of_deflated_singular_vals}
        \sigma_i^2 (\Ab-\zb\zb^T \Ab) + \sum_{j\neq i, j\in[n-1]} \sigma_j^2 (\Ab-\zb\zb^T \Ab) &\leq \sum_{j=2}^n \sigma_j^2(\Ab) + 2\epsilon \sigma^2_1(\Ab).
    \end{align}
    Using the eigenvalue min-max theorem (Fact \ref{fact:minimax}), we have for all $j \in [n-1]$, $\sigma_j(\Ab-\zb\zb^T \Ab) \geq \sigma_{j+1}(\Ab)$. Using this fact in \eqref{eq:sum_of_deflated_singular_vals},
    \begin{align*}
        \sigma_i^2 (\Ab-\zb\zb^T \Ab) + \sum_{j\neq i, j\in[n-1]} \sigma_{j+1}^2 (\Ab) &\leq \sum_{j=2}^n \sigma_j^2(\Ab) + 2\epsilon \sigma^2_1(\Ab),
    \end{align*}
    which implies that $\sigma_i^2 (\Ab-\zb\zb^T \Ab) \leq \sigma_{i+1}^2(\Ab) + 2\epsilon \sigma^2_1(\Ab)$. This completes the proof after adjusting $\epsilon$ by a factor of $2$.
\end{proof}

We now state as Algorithm \ref{alg:pow_it2} our main deterministic algorithm for estimating $k$ singular values of a matrix by leveraging Algorithm~\ref{alg:pow_it} as a subroutine. We will then show that by applying Algorithm \ref{alg:pow_it2} to a sparse spectral approximation $\bv{\widetilde A}$ of $\bv{A}$, we can estimate the singular values of $\bv{A}$ up to error $\pm \epsilon \max(n,\norm{\bv A}_1)$ in $\widetilde O(n^2/\poly(\epsilon))$ time.

\begin{algorithm}[H] 
\caption{Deterministic Singular Value Estimation via Deflation}
\label{alg:pow_it2}
\begin{algorithmic}[1]
\Require{$\bv A \in \mathbb{R}^{n\times n}$, error parameter $\epsilon \in (0,1)$, number of singular values to be estimated $k \le n$}
\State $\bv{A}^{(1)} = \Ab$.
\For{$i=1\to k $}
    \State $\bv{z}_i \gets$  output of Algorithm~\ref{alg:pow_it} with input $\bv{A}^{(i)}$, error parameter $\epsilon$.
    \State $\Ab^{(i+1)} \gets \bv{A}^{(i)} - \bv{z}_i\bv{z}_i^T\bv{A}$.
\EndFor
\State \Return {$\zb_1, \zb_2, \ldots, \zb_{k}$}
\end{algorithmic}
\end{algorithm}

\begin{lemma}\label{lem:all_sing_val}
Let $\{\zb_1, \zb_2, \ldots \zb_{k} \}$ be the output of Algorithm~\ref{alg:pow_it2} for some matrix $\bv A \in \mathbb{R}^{n\times n}$ with singular values $\sigma_1(\bv{
A}) \geq \sigma_2(\bv{
A}) \geq \ldots \geq \sigma_n(\bv{
A})$ and error parameter $\epsilon \in (0,1)$ as input. Then, $\{\zb_1, \zb_2, \ldots \zb_{k} \}$ are orthogonal unit vectors and for any $i \in [k]$:
\begin{equation*}
    (1-\epsilon)\sigma_i(\bv{
A}) \leq \|\bv{A}^T\bv{z}_i \|_2 \leq \sigma_i(\bv{A}) +\sigma_1(\bv{A})\sqrt{i \cdot \epsilon}
\end{equation*}
\end{lemma}
\begin{proof}
Let $\sigmatilde_i=\|(\bv{A}^{(i)})^T \bv{z}_i \|_2$ for all $i$. Since each $\bv{z}_i$ is the output of Algorithm~\ref{alg:pow_it} with $\bv{A}^{(i)}$ as the input, using the bound of Lemma~\ref{lem:top_eig}, for any $i \in [k]$ we get that
\begin{equation}\label{eq:sigma_bound}
    (1-\epsilon)\sigma_1(\bv{A}^{(i)}) \leq \sigmatilde_i  \leq \sigma_1(\bv{A}^{(i)}).
\end{equation}
Note that $\bv{A}^{(i)}=\bv{A}^{(i-1)}-\bv{z}_{i-1}\bv{z}_{i-1}^T\bv{A}=\bv{A}-\sum_{j=1}^{i-1}\bv{z}_{j}\bv{z}_{j}^T\bv{A}=\bv{A}-\bv{Z}_{i-1}\bv{Z}_{i-1}^T\bv{A}$ where $\bv{Z}_{i-1}$ is a matrix with $i-1$ columns where the $j$\textsuperscript{th} column is equal to $\bv{z}_j$. We first prove that $\bv{Z}_{i-1}$ is a matrix with orthonormal columns. First note that each $\bv{z}_i$ is a unit vector since Algorithm~\ref{alg:pow_it} always outputs a unit vector. We can prove that all $\bv{z}_i$ are orthogonal to each other by induction. Suppose that all $\bv{z}_{k}$ with $k \in [j]$ are orthogonal to each other for some $j<i-1$. Now, $\bv{A}^{(j+1)}=\bv{A}-\bv{Z}_{j}\bv{Z}_{j}^T\bv{A}=(\bv{I}-\bv{Z}_{j}\bv{Z}_{j}^T)\bv{A}$ where $\bv{I}$ is the identity matrix. Observe that $\bv{z}_{j+1}$ lies in the column span of $\bv{A}^{(j+1)}$. Thus, to prove that $\bv{z}_{j+1}$ is orthogonal to all $\bv{z}_k$ with $k \in [j]$, it is enough to show that each such $\bv{z}_{k}$ is orthogonal to the column space of $\bv{A}^{(j+1)}$. This follows since $\bv{z}_k = \bv{Z}_j\bv{Z}_j^T\bv{z}_k$ since $\bv{Z}_j \bv{Z}_j^T$ is a projection onto the span of $\bv{z}_1,\ldots,\bv{z}_j$. Thus, $(\bv{A}^{(j+1)})^T\bv{z}_k= \bv{A}^T(\bv{I}-\bv{Z}_j\bv{Z}_j^T)\bv{z}_k= \bv{A}^T(\bv{z}_k-\bv{z}_k)=0$.
Thus, $\bv{z}_{j+1}$ is orthogonal to all $\bv{z}_k$ with $k \in [j]$. Hence, $\bv{Z}_{i-1}$ is a matrix with orthonormal columns. This implies that $\bv{Z}_{i-1}\bv{Z}_{i-1}^T$ is a rank $(i-1)$ projection matrix. 

By the above orthogonality claim, $\bv{Z}_{i-1}^T \bv{z}_i = \bv{0}$ and  we have $\sigmatilde_i=\|(\bv{A}^{(i)})^T\bv{z}_i \|_2=\|(\Ab^T-\Ab^T\Zb_{i-1}\Zb_{i-1}^T)\zb_{i} \|_2=\| \Ab^T\zb_i\|_2$. So, to prove the lemma, it is enough to bound $\sigmatilde_i$. We first prove the lower bound on $\widetilde \sigma_i$ using the min-max principle of eigenvalues, which we state below.
\begin{fact}[Eigenvalue Minimax Principle -- see e.g., \cite{bhatia2013matrix}]\label{fact:minimax}
    Let $\Ab\in\R^\n$ be a symmetric matrix with eigenvalues $\lambda_1(\Ab) \geq \lambda_2(\Ab) \geq \ldots \geq \lambda_n(\Ab)$, then for all $i\in[n]$, 
    \begin{align*}
        \lambda_i(\Ab) &= \max_{S\colon dim(S)=i} \min_{\xb\in S, \|\xb\|_2 = 1} \xb^T \Ab\xb
      = \min_{S\colon dim(S)=n-i+1} \max_{\xb\in S, \|\xb\|_2 = 1} \xb^T \Ab\xb.
    \end{align*}
\end{fact}
By Fact \ref{fact:minimax}, $\sigma_1(\bv{A}^{(i)})^2=\sigma_1(\bv{A}-\bv{Z}_{i-1}\bv{Z}_{i-1}^T\bv{A})^2 = \lambda_1((\bv I-\bv{Z}_{i-1}\bv{Z}_{i-1}^T)\bv{A} \bv{A}^T (\bv I-\bv{Z}_{i-1}\bv{Z}_{i-1}^T)) \geq \lambda_i(\bv{AA}^T) = \sigma_i(\bv{A})^2$ since $\bv{Z}_{i-1}\bv{Z}_{i-1}^T$ is a rank $i-1$ projection matrix. Thus, from~\eqref{eq:sigma_bound}, we get $\sigmatilde_i \geq (1-\epsilon)\sigma_1(\bv{A}^{(i)}) \geq (1-\epsilon)\sigma_i(\bv{A})$.

We now prove the upper bound on $\widetilde \sigma_i$. For any $i \in [k]$ and any $j \in [n]$ from Lemma~\ref{lem:top_eig} we have:
\begin{equation*}
   \sigma_j^2 (\Ab^{(i)}-\zb_i\zb_i^T \Ab^{(i)}) \leq \sigma_{j+1}^2(\Ab^{(i)}) + \epsilon \sigma^2_1(\Ab^{(i)}).
\end{equation*}
Now using the fact that $\Ab^{(i)}-\zb_i\zb_i^T \Ab^{(i)}=\bv{A}-\bv{Z}_{i-1}\bv{Z}_{i-1}^T\bv{A}-\zb_i\zb_i^T \Ab=\bv{A}-\bv{Z}_{i}\bv{Z}_{i}^T\bv{A}$ and $\sigma_1(\Ab^{(i)}) \leq \sigma_1(\Ab)$, we get:
\begin{equation}\label{eq:low_rank_bound}
    \sigma_j^2 (\bv{A}-\bv{Z}_{i}\bv{Z}_{i}^T\bv{A}) \leq \sigma_{j+1}^2(\bv{A}-\bv{Z}_{i-1}\bv{Z}_{i-1}^T\bv{A}) + \epsilon \sigma^2_1(\Ab).
\end{equation}
From~\eqref{eq:sigma_bound} we have
$\sigmatilde_i \leq \sigma_1(\bv{A}-\bv{Z}_{i-1}\bv{Z}_{i-1}^T\bv{A})$. Squaring both sides and applying the bound from~\eqref{eq:low_rank_bound} $i-1$ times repeatedly on the RHS, we get 
\begin{align*}
    \sigmatilde^2_i &\leq \sigma^2_1(\bv{A}-\bv{Z}_{i-1}\bv{Z}_{i-1}^T\bv{A}) \leq \sigma^2_2(\bv{A}-\bv{Z}_{i-2}\bv{Z}_{i-2}^T\bv{A})+\epsilon \sigma^2_1(\bv{A})\\
    &\leq \sigma^2_3(\bv{A}-\bv{Z}_{i-3}\bv{Z}_{i-3}^T\bv{A})+2\epsilon \sigma^2_1(\bv{A}) \leq \ldots \leq \sigma^2_i(\bv{A})+i \cdot \epsilon \cdot\sigma^2_1(\bv{A}).
\end{align*} 
Finally, taking the square root on both sides and using that for any non-negative $a,b$ $\sqrt{a}+ \sqrt{b} \ge \sqrt{a+b}$ gives us the upper bound.
\end{proof}
Lemma \ref{lem:all_sing_val} in place, we are now ready to state and prove our main result on deterministic singular value estimation for any bounded entry matrix.
\deterministicsingval*
\begin{proof}
 Note that $\bv{A}$ can have at most $\frac{1}{\epsilon}$ singular values greater than $\epsilon \|\bv{A} \|_1$. Thus, it is enough to approximate only the top $\frac{1}{\epsilon}$ singular values of $\bv{A}$ up to error $\epsilon \cdot \max(n,\|\bv{A} \|_1)$. The rest of the singular values can be approximated by $0$ to still have the required approximation error of $\epsilon \cdot \max(n,\|\bv{A} \|_1)$. Let $\widetilde{\bv A} = \bv A \circ \bv S$ be the sparsification of $\bv{A}$ as described in Theorem~\ref{thm: general matrix eps4 bound}, which satisfies
$\|\bv{A}-\widetilde{\bv{A}} \|_2 \leq \epsilon \cdot \max(n, \| \bv{A}\|_1).$ Using Weyl's inequality we have for all $i \in [n]$: $$|\sigma_i(\bv{A})- \sigma_i(\widetilde{\bv{A}})|\leq \epsilon \cdot \max(n, \|\bv{A} \|_1).$$ Thus it is enough to approximate the top $\frac{1}{\epsilon}$ singular values of $\widetilde{\bv{A}}$ up to $\epsilon \cdot \max(n, \|\bv{A} \|_1)$ error and then use triangle inequality to get the final approximation to the top $ \frac{1}{\epsilon} $ singular values of $\bv{A}$. To do so, we run Algorithm~\ref{alg:pow_it2} with $\widetilde{\bv A}$ as the input matrix, $\epsilon^3$ as the error parameter, and $k = \lceil 1/\epsilon \rceil$ as  the number of singular values to be estimated. Using Lemma~\ref{lem:all_sing_val} we get orthonormal vectors 
$\zb_i$ for $i \in \lceil \frac{1}{\epsilon} \rceil$ such that:
\begin{equation}\label{Eq:approx_bound}
    (1-\epsilon^3)\sigma_i(\widetilde{\bv A})\leq \| \Atilde\zb_i\|_2 \leq \sigma_i(\widetilde{\bv A})+\sigma_1(\widetilde{\bv A})\sqrt{i \cdot \epsilon^3}.
\end{equation}
Let $\sigmatilde_i(\bv{A})= \| \Atilde\zb_i\|_2$ for $i \in \lceil \frac{1}{\epsilon} \rceil$. Now since $\sigma_1(\bv A) \le n$, using Weyl's inequality, we have $\sigma_1(\bv{\widetilde A}) \le \sigma_1(\bv{A}) + \norm{\bv{A}-\bv{\widetilde A}}_2 \le \sigma_1(\bv{A})+ \epsilon \max(n,\norm{\bv A}_1) \le 2 \max(n,\norm{\bv A}_1)$. Thus, for any $i \in [\lceil 1/\epsilon \rceil]$, $\sigma_1(\widetilde{\bv A})\sqrt{i \cdot \epsilon^3} \leq 2\epsilon \max(n,\norm{\bv A}_1)$. So, from Equation~\eqref{Eq:approx_bound}, we get:
\begin{equation*}
    (1-\epsilon)\sigma_i(\widetilde{\bv A})\leq \widetilde{\sigma}_i(\Ab) \leq \sigma_i(\widetilde{\bv A})+2\epsilon \cdot \max(n,\norm{\bv A}_1).
\end{equation*}
This gives us $|\widetilde{\sigma}_i(\Ab)-\sigma_i(\widetilde{\bv A})| \leq 2\epsilon \cdot \max(n,\norm{\bv A}_1)$. Finally, we get: \begin{align*}
    |\widetilde{\sigma}_i(\Ab)-\sigma_i(\Ab)| &\leq |\widetilde{\sigma}_i(\Ab)-\sigma_i(\widetilde{\bv A})| +|\sigma_i(\widetilde{\bv A})-\sigma_i(\Ab)| \leq 3\epsilon \max(n,\norm{\bv A}_1),
\end{align*}
where the second inequality follows from the bound on $|\widetilde{\sigma}_i(\Ab)-\sigma_i(\widetilde{\bv A})|$ and Weyl's inequality which gives us $|\sigma_i(\widetilde{\bv A})-\sigma_i(\Ab)| \leq \epsilon \max(n,\norm{\bv A}_1)$. This gives us the required additive approximation error after adjusting $\epsilon$ by constants.
Next, observe that for $i \in \lceil \frac{1}{\epsilon} \rceil$, using triangle inequality:
\begin{align*}
    |\|\Ab \zb_i \|_2 -\sigma_i(\bv{A})| &\leq |\|\Ab \zb_i \|_2-\|\Atilde \zb_i \|_2| +|\|\Atilde \zb_i \|_2 -\sigma_i(\Ab)|
    \leq 4\epsilon \max(n,\norm{\bv A}_1), 
\end{align*}
where in the second step, the first term is bounded as $|\|\Ab \zb_i \|_2-\|\Atilde \zb_i \|_2| \leq \|(\Ab -\Atilde)\zb_i \|_2 \leq \| \Ab -\Atilde\|_2 \leq \epsilon \max(n,\norm{\bv A}_1)$. This gives us the second bound (after adjusting $\epsilon$ by constants).

\medskip 

\noindent \textbf{Sample Complexity and Runtime Analysis.} 
By Theorem \ref{thm: general matrix eps4 bound}, the number of queries to $\bv{A}$'s entries need to construct $\widetilde{\bv A} = \bv A \circ \bv S$ is $\widetilde{O}\big(\frac{n}{\epsilon^4} \big)$. The loop estimating the singular values in Algorithm~\ref{alg:pow_it2} runs $\frac{1}{\epsilon}$ times and Algorithm~\ref{alg:pow_it} is called in Line 3 inside the loop each time with error parameter $\epsilon^3$. At the i\textsuperscript{th} iteration of the loop in Algorithm~\ref{alg:pow_it2}, in Line 4 we have $\Atilde^{(i+1)}=\bv{\widetilde A}^{(i)} - \bv{z}_{i}^T \bv{z}_{i}^T \bv{\widetilde A} =\bv{\widetilde A} - \bv{Z}_{i}^T \bv{Z}_{i}^T \bv{\widetilde A}$. Note that since the matrix $\Atilde^{(i+1)}$ could be dense we don't explicitly compute $\Atilde^{(i+1)}$ to do power iterations in Algorithm~\ref{alg:pow_it} with this matrix. Instead, in each step of power iteration in Line 5 of Algorithm~\ref{alg:pow_it}, we can first calculate $\bv{\widetilde A}\bv{y}_j$ in time $\widetilde{O}\big(\frac{n}{\epsilon^4} \big)$ and then multiply this vector first with $\bv{Z}_{i}^T$ and then with $\bv{Z}_{i}$ in  $O\big(\frac{n}{\epsilon} \big)$ time to get $\bv{Z}_{i} \bv{Z}_{i}^T \bv{\widetilde A}$. Thus, each power iteration (Lines 5 and 6 of Algorithm~\ref{alg:pow_it}) takes $\widetilde{O}\big(\frac{n}{\epsilon^4} \big)$ time. Since we call Algorithm~\ref{alg:pow_it} with error parameter $\epsilon^3$, the number of power iterations with each starting vector is $O\big(\frac{\log(n/\epsilon)}{\epsilon^3} \big)$. There are also $n$ different starting vectors. Thus, each call to Algorithm~\ref{alg:pow_it} in Line 3 of Algorithm~\ref{alg:pow_it2} takes  $\widetilde{O}\big(\frac{n^2}{\epsilon^7} \big)$ time. Thus, the total running time  of the algorithm is $\widetilde{O}\big(\frac{n^2}{\epsilon^8} \big)$ (as the loop in Algorithm~\ref{alg:pow_it2} runs $\frac{1}{\epsilon}$ times).  
\end{proof}


\subsection{Deterministic PSD Testing}\label{sec:psdtesting}

We next leverage our deterministic power method approach (Algorithm \ref{alg:pow_it}) to give an $o(n^\omega)$ time deterministic algorithm for testing if a bounded entry  matrix is either PSD or has at least one large negative eigenvalue  $\leq -\epsilon \max(n,\norm{\bv A}_1)$. An optimal randomized algorithm for this problem with detection threshold $-\epsilon n$ was presented
 in~\cite{Bakshi:2020uz}. The idea of our approach is to approximate the maximum singular value of $\Ib-\frac{\Atilde}{\|\Atilde\|_2}$ which will be at smaller than 1 if $\bv{A}$ is PSD and greater than 1 if $\bv{A}$ is $\epsilon n$ far from being PSD.

\begin{theorem}[Deterministic PSD testing with $\ell_\infty$-gap]\label{cor:testing_psd}
    Given $\epsilon \in (0,1)$, let $\Ab\in\R^\n$ be a symmetric bounded entry matrix with eigenvalues $\lambda_1(\bv{A}) \geq \lambda_2(\bv{A}) \geq \ldots \geq \lambda_n(\bv{A})$ such that either $\bv{A}$ is PSD i.e. $\lambda_n(\Ab) \geq 0$ or $\Ab$ is at least $\epsilon \cdot \max(n, \|\bv{A} \|_1)$-far from being PSD, i.e., $ \lambda_n(\Ab) \leq -\epsilon \cdot\max(n, \|\bv{A} \|_1)$. There exists a deterministic algorithm that distinguishes between these two cases by reading $\widetilde{O}\left (\frac{n}{\epsilon^4}\right )$ entries of $\Ab$ and which runs in time $\widetilde{O}\left(\frac{n^2}{\epsilon^5}\right)$.
\end{theorem}

\begin{proof}
    Let $\Atilde = \bv A \circ \bv S$ be a sparsification of $\Ab$ as described in Theorem \ref{thm: general matrix eps4 bound} with approximation parameter $\frac{\epsilon}{10}$ that satisfies $\|\Ab - \Atilde\|_2 \leq \frac{\epsilon}{10}\max(n, \|\Ab\|_1)$. Using Weyl's inequality we have for all $i\in [n]$, $|\lambda_i(\Ab) - \lambda_i(\Atilde)| \leq \frac{\epsilon}{10}\max(n, \|\Ab\|_1)$. 

    Let $\widetilde{\sigma}_{1}(\Atilde)$  be the estimate of $\| \Atilde \|_2=\sigma_{1}(\Atilde)$ output by  Algorithm~\ref{alg:pow_it} with input $\Atilde$ and error parameter $\frac{\epsilon}{100}$. Then by Lemma~\ref{lem:top_eig}, we have:
    \begin{equation*}
       \left(1-\frac{\epsilon}{100}\right)\sigma_{1}(\Atilde) \leq \widetilde{\sigma}_{1}(\Atilde) \leq \sigma_{1}(\Atilde).
    \end{equation*}
    Let $\bv B = \Ib - \frac{1}{\widetilde{\sigma}_{1}(\Atilde)}\Atilde$. Let $\widetilde{\sigma}_{1}(\bv B)$  be the estimate of $\sigma_1(\bv{B})$ output by Algorithm~\ref{alg:pow_it} with input $\bv B$ and error parameter $\frac{\epsilon}{1000}$. Then by Lemma~\ref{lem:top_eig}, we again have:
    \begin{equation*}
       \left(1-\frac{\epsilon}{1000}\right)\sigma_1(\bv{B}) \leq \widetilde{\sigma}_{1}(\bv B) \leq \sigma_1(\bv{B}).
    \end{equation*}
    We now distinguish between the two cases as follows:
    
    \medskip
    
  \noindent  \textbf{Case 1:} When $\bv{A}$ is PSD, we have $\lambda_{n}(\widetilde{\bv{A}}) \geq \lambda_{n}(\bv{A}) -\frac{\epsilon n}{10} \geq -\frac{\epsilon n}{10}$. Now, $\sigmatilde_1(\bv{B}) \leq \sigma_1(\bv{B})=\max(|\lambda_1(\bv{B})|, |\lambda_n(\bv{B})|)=\max(\big|1-\frac{\lambda_n(\Atilde)}{\sigmatilde_1(\Atilde)} \big|,\big|1-\frac{\lambda_1(\Atilde)}{\sigmatilde_1(\Atilde)} \big| )$.
    
    First assume that $\big|1-\frac{\lambda_n(\Atilde)}{\sigmatilde_1(\Atilde)} \big|  \geq \big|1-\frac{\lambda_1(\Atilde)}{\sigmatilde_1(\Atilde)} \big| $. Now, $1-\frac{\lambda_n(\Atilde)}{\sigmatilde_1(\Atilde)} \leq 1+\frac{\epsilon n}{10 \sigmatilde_1(\Atilde)}$. Thus, $\sigmatilde_1(\bv{B}) \leq 1+\frac{\epsilon n}{10 \sigmatilde_1(\Atilde)}$. Now, assume $\big|1-\frac{\lambda_n(\Atilde)}{\sigmatilde_1(\Atilde)} \big| < \big|1-\frac{\lambda_1(\Atilde)}{\sigmatilde_1(\Atilde)} \big| $. Then, we must have $\lambda_1(\Atilde)  \geq 0$ as otherwise, we will have $1-\frac{\lambda_n(\Atilde)}{\sigmatilde_1(\Atilde)} \geq 1-\frac{\lambda_1(\Atilde)}{\sigmatilde_1(\Atilde)} >0$. Thus, we get that $\sigmatilde_1(\bv{B})\leq \big|1-\frac{\lambda_1(\Atilde)}{\sigmatilde_1(\Atilde)} \big|=\big|1-\frac{\sigma_1(\Atilde)}{\sigmatilde_1(\Atilde)} \big|=\big| 1-\frac{1}{1-\frac{\epsilon}{100}}\big|=\frac{\epsilon/100}{1-\epsilon/100}< \frac{\epsilon}{50} < 1 + \frac{\epsilon n}{10 \widetilde \sigma_1(\bv{\widetilde A})}$.

\medskip

    \noindent \textbf{Case 2:} When $\bv{A}$ is $\epsilon \cdot \max(n, \|\bv{A} \|_1)$ far from being PSD, we have $\lambda_{n}(\widetilde{\bv{A}}) \leq  \lambda_{n}(\bv{A}) +\frac{\epsilon \cdot \max(n, \|\bv{A} \|_1)}{10} \leq -\frac{9}{10}\epsilon \cdot \max(n, \|\bv{A} \|_1)$. Observe that $\sigma_1(\bv{B}) \geq 1-\frac{1}{\sigmatilde_1(\Atilde)}\lambda_n(\Atilde) \geq 1+\frac{9}{10\sigmatilde_1(\Atilde)}\epsilon \cdot \max(n, \|\bv{A} \|_1)$. Then, we have $\sigmatilde_1(\bv{B}) \geq \left(1-\frac{\epsilon}{1000}\right)\sigma_1(\bv{B}) \geq \left(1-\frac{\epsilon}{1000}\right)\left(1+\frac{9}{10\sigmatilde_1(\Atilde)}\epsilon \cdot \max(n, \|\bv{A} \|_1)\right) > 1+\frac{2}{10\sigmatilde_1(\Atilde)}\epsilon \cdot \max(n, \|\bv{A} \|_1)) >1+\frac{2\epsilon n}{10\sigmatilde_1(\Atilde)}$ .

    Thus, in \textbf{Case 1}, we have $\sigmatilde_1(\bv{B}) \leq 1+\frac{\epsilon n}{10\sigmatilde_1(\Atilde)}$ while for \textbf{Case 2}, we have $\sigmatilde_1(\bv{B})>1+\frac{2\epsilon n}{10\sigmatilde_1(\Atilde)}$. So we can distinguish between the above two cases. 

    \medskip 
    
    \noindent \textbf{Sample Complexity and Runtime Analysis.} The total runtime is the time taken to estimate $\sigma_1(\Atilde)$ and $\sigma_1(\bv{B})$ using Algorithm~\ref{alg:pow_it} with error parameter $\Theta(\epsilon)$. This  is $\widetilde{O}\left(\frac{n^2}{\epsilon^5}\right)$ since we run $O\left(\frac{\log(n/\epsilon)}{\epsilon}\right)$ iterations of the power method with $n$ different starting vectors, where each iteration takes time $\widetilde{O}\left(\frac{n}{\epsilon^4}\right)$ since $\bv{\tilde A}$ and $\bv B$ both have $\tilde O(n/\epsilon^4)$ non-zero entries. The sample complexity required for form $\bv {\tilde A}$ and $\bv B$ is $\tilde O(n/\epsilon^4)$.
 
\end{proof}

\subsection{High Accuracy Spectral Norm Approximation}\label{sec:highaccuracy}

Finally, we show that, under the assumption that $\sigma_1(\bv{A}) \geq \alpha \max(n,\|\Ab \|_1)$ for some $\alpha \in (0,1)$, we can deterministically compute $\sigmatilde_1(\bv{A})$ such that $\sigmatilde_1(\Ab) \geq (1-\epsilon)\sigma_1(\Ab)$ in $\widetilde{O}\big(\frac{n^2 \log(1/\epsilon)}{\poly(\alpha)}\big)$ time. This allows us to set $\epsilon$ to $\frac{1}{n^c}$ to get a high accuracy approximation to $\sigma_1(\Ab)$ in roughly linear time in the input matrix size. This is the first $o(n^\omega)$ time deterministic algorithm for this problem, and to the best of our knowledge matches the best known randomized methods for high accuracy top singular value computation, up to a $\poly\log(n,1/\alpha)$ factor. Our approach is described in Algorithm \ref{alg:high_acc_spec_norm} below.

\begin{algorithm}[H] 
\caption{High Accuracy Spectral Norm Approximator}
\label{alg:high_acc_spec_norm}
\begin{algorithmic}[1]
\Require{Symmetric $\bv A \in \mathbb{R}^{n\times n}$, parameter $\alpha$, and error parameter $\epsilon \in (0,1)$}
\State Let $\Atilde=\bv{A} \circ \bv{S}$ where $\bv{S}$ is as specified in Theorem~\ref{thm: general matrix eps4 bound} with error parameter $\epsilon=c\alpha^4$ for a sufficiently small constant $c$.
\State $\Zb \gets$ output of Algorithm \ref{alg:pow_it2} with inputs $\Atilde$, error parameter $\widetilde\epsilon=c\alpha^4$, and $k=2/\alpha$.
\State $\Ztilde \gets$ Output of Block Krylov method (Algorithm 2 of \cite{musco2015randomized}) with input matrix $\Ab$, starting block $\Zb$, and number of iterations $q=\frac{c'\log(n/\epsilon)}{\alpha^{3/2}}$ for some enough large constant $c'$.
\State \Return {$\sigmatilde_1(\bv{A})=\sqrt{\widetilde{\bv {z}}_1^T \bv{A}\bv{A}^T\widetilde{\bv {z}}_1}$, where $\widetilde{\bv {z}}_1$ is the first column of $\Ztilde$}. 
\end{algorithmic}
\end{algorithm}

Algorithm \ref{alg:high_acc_spec_norm} first computes (Step 2) a coarse approximate subspace of top singular vectors for $\bv{A}$ using our deflation based Algorithm \ref{alg:pow_it2} applied to the sparsified matrix $\bv{\widetilde A}$. This subspace is then used in Step 3 to initialize a Block Krylov method applied to $\bv{A}$ to compute a much more accurate approximation to $\sigma_1(\bv{A})$. Our proof will leverage the analysis of Block Krylov methods given in~\cite{musco2015randomized}. Specifically, we will apply Theorem 13 of~\cite{musco2015randomized} which bounds the number of iterations required to find a $(1 \pm \epsilon)$ accurate approximation  to the top $k$ singular values in terms of the singular value gap $\frac{\sigma_k(\bv{A})}{\sigma_{p+1}(\bv{A})}$ for some block size $p \ge k$. 

The proof of Theorem 13 in \cite{musco2015randomized} uses a degree $q$ amplifying polynomial $f(\bv{A})$ (for large enough $q$) in the Krylov subspace to significantly separate the top $k$ singular values from the rest. It is shown that if the starting block $\bv{Z}$ is such that $f(\bv{A})$ projected onto the column span of $f(\bv{A})\bv{Z}$ is a good rank $p$ approximation to $f(\bv{A})$, then the approximate singular values and vectors output by the Block Krylov method will be accurate. To prove this condition for our starting block, we will apply the following well-known result, which bounds the error of any low rank approximation to a matrix in terms of the error of the best possible low rank approximation.

\begin{lemma}[Lemma 48 of \cite{woodruff2014sketching}, proved in \cite{boutsidis2014near}]\label{lem:fro_broken_bound}
    Let $\Ab = \Ab  \Sb\Sb^T + \bv E$ be a low-rank matrix factorization of $\Ab$, with $\Sb\in\R^{n\times k}$, and $\Sb^T\Sb = \Ib_k$. Let $\bv Z \in \R^{n\times c}$ with $c\geq k$ be any matrix such that $\rank{(\Sb^T\bv Z)} = \rank{(\Sb)} = k$. Let $\bv{Q}$ be an orthonormal basis for the column span of $\bv{AZ}$. Then, $$\|\Ab - \bv Q\bv Q^T\Ab\|_F^2 \leq \|\bv E\|_F^2 + \|\bv E\bv Z(\Sb^T\bv Z)^{+}\|_F^2.$$
\end{lemma}

The above result lets us bound the difference between the Frobenius norm error of projecting $f(\bv{A})$ onto the column span of $f(\bv{A})\bv{Z}$, and the error of the best rank $p$ approximation of $f(\bv{A})$. The proof of Theorem 13 of~\cite{musco2015randomized} uses Lemma 4 of~\cite{musco2015randomized}, (which in turn is proven using Lemma \ref{lem:fro_broken_bound} above) to show that $\|f(\bv{A}) -\bv{Q}\bv{Q}^Tf(\bv{A})\|^2_F \leq n^c \cdot \|f(\bv{A})-(f(\bv{A}))_p \|^2_F$ where $\bv{Q}$ is an orthonormal basis for $\bv{AZ}$, where $\bv{Z}$ is a random starting block, $c$ is some constant, and $(f(\bv{A}))_p$ is the best rank $p$ approximation to $f(\bv{A})$. In order to apply this result in our setting, we need to show that the same bound applies for our starting block $\bv{Z}$, which is computed deterministically by Algorithm \ref{alg:pow_it2}. 

In particular, by Lemma \ref{lem:fro_broken_bound}, it suffices to bound $\|(\bv{V}_p^T\bv{Z})^{+} \|_2$, where $\bv{V}_p$ contains the top $p$ singular vectors of $\bv{A}$ (which are identical to those of $f(\bv A)$). That is, we need to ensure that the columns of $\bv{Z}$ and the top $p$ singular vectors of $\bv{A}$ have large inner product with each other. We next prove that this is the case when there is a large gap between the top $p$ singular values and the rest. We will later show that under the assumption that $\sigma_1(\bv A) \ge \alpha \cdot \max(n,\norm{\bv A}_1)$, there is always some $p$ that satisfies this gap assumption.

\begin{lemma}[Starting Block Condition]\label{lem:spectral bound zu}
    Let $\Ab \in \R^\n$ be a bounded entry symmetric matrix with singular values $\sigma_1(\Ab) \geq \sigma_2(\Ab) \geq \ldots \geq \sigma_n(\Ab)$. Let $\Atilde = \bv{A} \circ \bv{S}$ where the sampling matrix $\bv S$ is as specified in Theorem~\ref{thm: general matrix eps4 bound} with error parameter $\epsilon=c\alpha^4$ for a sufficiently small constant $c$  and some $\alpha \in (0,1)$. For some $k\in[n]$, let $\Zb \in \R^{n \times k}$ be the output of Algorithm \ref{alg:pow_it2} with input $\Atilde$ and parameters $k=\frac{2}{\alpha}$ and $\epsilon=c\alpha^4$. For some $p\leq \frac{2}{\alpha}$, assume that for some constant $c_1$, $$
        \sigma_p^2(\Ab) - \sigma_{p+1}^2(\Ab) \geq {c}_1\alpha^3\cdot(\max(n,\|\Ab\|_1))^2.
    $$
    Also let $\bv V_p \in \R^{n \times p}$ be the matrix with columns equal to the $p$ singular vectors of $\bv{A}$ corresponding to its largest $p$ singular values. Let $\Zb_p$ contain the first $p$ columns of $\Zb$. Then for some constant $c_2$, $$\norm{(\bv{V}_p^T \bv{Z}_p)^+}_2 \le c_2n.$$
\end{lemma}
\begin{proof}
    Before proving the main result, we will bound the error between $\sigma^2_j(\Ab)$ and $\|\Ab\zb_j \|^2_2$ where $\bv{z}_j$ is a column of $\bv{Z}$ for any $j \in \lceil \frac{2}{\alpha} \rceil$. From Lemma~\ref{lem:all_sing_val}, we have for any $j \in \lceil \frac{2}{\alpha} \rceil$:
    \begin{equation*}
        \|\Atilde\bv{z}_j \|_2 \geq (1-c \alpha^4)\sigma_j(\Atilde),
    \end{equation*}
    since we run Algorithm~\ref{alg:pow_it2} with error parameter $\epsilon=c\alpha^4$. Also since $\|\Atilde -\Ab \|_2 \leq c \alpha^4 \max(n,\|\bv{A} \|_1)$, using Weyl's inequality and the fact that $\sigma_j(\bv{A}) \leq \max(n,\|\bv{A} \|_1)$ we have (for some constant $b$) $$ \sigma_j(\bv{A})-c\alpha^4 \max(n,\|\bv{A} \|_1) \leq \sigma_j(\Atilde) \leq \sigma_j(\bv{A})+ c\alpha^4 \max(n,\|\bv{A} \|_1) \leq b \max(n,\|\bv{A} \|_1).$$
    
    Now, using triangle inequality, $$\|\bv{A} \bv{z}_j \|_2 \geq \|\Atilde\bv{z}_j \|_2-\|(\Atilde -\Ab)\bv{z}_j \|_2 \geq (1-c \alpha^4)\sigma_j(\Atilde) -c\alpha^4 \max(n,\|\bv{A} \|_1),$$ where the second inequality uses that $\|(\Atilde -\Ab)\bv{z}_j \|_2 \leq \|\Atilde -\Ab \|_2\|\bv{z}_j \|_2\leq c \alpha^2 \max(n,\|\bv{A} \|_1)$. Thus, using the bounds on $\sigma_j(\Atilde)$, we get: 
    \begin{align*}
        \|\bv{A}\bv{z}_j \|_2 &\geq \sigma_j(\bv{A})-c\alpha^4 \max(n,\|\bv{A}\|_1)-c\alpha^4 \sigma_j(\Atilde)-c\alpha^4\max(n,\|\bv{A}\|_1) \\
        &\geq \sigma_j(\bv{A}) -c''\alpha^4\max(n,\|\bv{A}\|_1),
    \end{align*}
    for some constant $c''$. Finally, squaring both sides, we get for any $j \in \lceil \frac{2}{\alpha} \rceil$:
    \begin{align}
        \|\bv{A}\bv{z}_j \|_2^2 &\geq \sigma^2_j(\Ab) -2c''\alpha^4 \sigma_j(\Ab)\max(n,\|\Ab \|_1) +(c'')^2\alpha^8 (\max(n,\|\Ab \|_1))^2 \notag\\
        \label{Eq:abc}
        &\geq \sigma^2_j(\Ab)-2c''\alpha^4(\max(n,\|\Ab \|_1))^2.
    \end{align}
    
    We are now ready to prove the main result of the lemma. Note that showing $\norm{(\bv{V}_p^T \bv{Z}_p)^+}_2 \le c_2n$ is equivalent to showing that $\sigma_{\min}(\bv{V}_p^T \bv{Z}_p) \ge \frac{1}{c_2n}$. Assume for contradiction that $\sigma_{\min}(\bv{V}_p^T \bv{Z}_p) \le \frac{1}{c_2n}$. Then there is some vector $\bv{y}$ with $\norm{\bv{y}}_2 = 1$ such that $\norm{\bv{Z}_p^T \bv{V}_p \bv{y}}_2 \le \frac{1}{c_2n}$. Let $\bv{w} = \bv{V}_p \bv{y}$. Observe that $\| \bv w\|_2=1$ since $\bv{V}_p$ has orthonormal columns. Then, for any column $\bv{z}_i$ of $\bv{Z}_p$, 
    \begin{align}\label{Eq:projn}
        \norm{\bv{w} \bv{w}^T \bv{z}_i}_2 = |\bv{w}^T \bv{z}_i| \le \norm{\bv w^T \bv{Z}_p}_2 \le \frac{1}{c_2n},
    \end{align}
    where the last step follows by our assumption. Then, using triangle inequality and the above bound, we have $\norm{\bv{A}(\bv I - \bv{ww}^T) \bv{z}_i}_2 \ge \norm{\bv{A} \bv z_i}_2 - \norm{\bv A \bv{ww}^T\bv{z}_i}_2 \ge \norm{\bv{A} \bv{z}_i}_2 - \frac{\norm{\bv{A}}_2}{c_2n} \geq  \norm{\bv{A} \bv{z}_i}_2 - \frac{1}{c_2}$, where the last equality uses the fact that $\|\Ab\|_2 \leq  n$ since it has bounded entries. Thus, for any $i \in [p]$, squaring both sides of this bound, we get:    \begin{align}
    \norm{\bv{A}(\bv I - \bv{ww}^T) \bv{z}_i}_2^2 &\ge \norm{\bv A \bv{z}_i}_2^2 - \frac{2\|\Ab\zb_i\|_2}{c_2 } \notag\\ 
    &\geq \norm{\bv A \bv{z}_i}_2^2 - \frac{2}{c_2}\max(n,\|\bv{A} \|_1)\notag\\
    \label{eq:AIwwz}
    &\ge \sigma^2_i(\bv A) - c_3\alpha^4 \max(n,\norm{\Ab}_1)^2,
    \end{align}
    for some constant $c_3$, where the second inequality uses $\|\Ab\zb_i\|_2 \leq n$, and the last inequality uses~\eqref{Eq:abc}. Moreover since $\bv{w}$ is spanned by $\bv{V}_p$, we have:
    \begin{align}
        \label{eq:Aw22}
        \norm{\bv{A} \bv w}_2^2 \ge 
       \min_{\bv{x}:\|\bv{x}\|_2=1} \|\bv{A}\bv{V}_p \bv{x}\|_2 \ge \sigma^2_p(\bv A) > \sigma_{p+1}^2(\bv A) + c_1\alpha^3\cdot(\max(n,\|\Ab\|_1))^2,
    \end{align} 
    where the final inequality is by the gap assumption in the lemma statement. 
    
    Consider the matrix $\bv{W} = [\bv w, (\bv I-\bv{ww}^T) \bv Z_p]$. Next, we show that under our assumption that $\sigma_{\min}(\bv{V}_p^T \bv{Z}_p) \le \frac{1}{c_2 n}$, $\bv{W}$ will be very close to orthonormal, and further that by \eqref{eq:AIwwz} and \eqref{eq:Aw22}, $\norm{\bv{AW}}_F^2 > \sum_{i=1}^{p+1} \sigma_i^2(\bv{A})$. Since $\bv{W}$ has $p+1$ columns, this will lead to a contradiction since we must have $\norm{\bv{AW}}_F^2 \le \norm{\bv{AV}_{p+1}}_F^2 =  \sum_{i=1}^{p+1} \sigma_i^2(\bv{A})$.
   In particular, using \eqref{eq:AIwwz} and \eqref{eq:Aw22},
    \begin{align}
    \norm{\bv A \bv W}_F^2 &= \norm{\bv A \bv w}_2^2 + \sum_{i=1}^p \norm{\bv{A} (\bv I - \bv{ww}^T)\bv{z}_i}_2^2\notag\\
    &\ge \sigma_{p+1}^2(\bv A) + c_1\alpha^3\cdot(\max(n,\|\Ab\|_1))^2 + \sum_{i=1}^p \sigma_i(\bv A)^2 - c_3p\alpha^4 \max(n,\norm{A}_1)^2\notag\\
    \label{eq:AWF2}
    & \ge \sum_{i=1}^{p+1} \sigma_i^2(\bv A) + c_4 \alpha^3 \max(n,\norm{\bv A}_1)^2,
    \end{align}
    for some constant $c_4$ which will be positive as long as $c_1$ is large enough compared to $c_3$ (which can be made arbitrary small by our setting of $c_2,c$ in the lemma statement).
    
    We will now show that all eigenvalues of $\bv{W}^T \bv{W}$ are close to $1$, which implies that $\bv{W}$ is approximately orthogonal. This allows us to translate the bound on $\norm{\bv{AW}}_F^2$ given in \eqref{eq:AWF2} to an orthonormal basis $\bv{Q}$ for the column span of $\bv W$. 
    
    Observe that $(\bv{W}^T \bv{W})_{i,i} =\|\bv{w}\|_2^2= 1$ for $i = 1$ and using~\eqref{Eq:projn}, $|(\bv{W}^T \bv{W})_{i,i}  - \bv{z}_{i-1}^T \bv{z}_{i-1}| =|\bv{z}_{i-1} \bv{ww}^T \bv{z}_{i-1}| \leq \frac{1}{c_2^2 n^2}$ otherwise. Thus, $|(\bv{W}^T \bv{W})_{i,i}-1| \leq \frac{1}{c_2^2 n^2}$ for all $i$. Further, $(\bv W^T \bv W)_{i,j} = \bv{w}^T(\Ib-\bv{w}\bv{w}^T)\bv{z}_{j-1}=0$ when $i \neq j$ and $i=1$ as $\|\bv{w} \|_2=1$ (and similarly when $i \neq j$ and  $j = 1$). Also, when $i \neq j$ and $i,j>1$, $|(\bv W^T \bv W)_{i,j}| = |\bv{z}_i^T \bv{z}_j - \bv{z}_i^T \bv{ww}^T \bv{z}_j|=|\bv{z}_i^T \bv{ww}^T \bv{z}_j| \leq  \frac{1}{c_2^2 n^2}$ otherwise. We can then apply Gershgorin's circle theorem to bound the eigenvalues of $\bv{W}^T \bv{W}$.
    \begin{fact}[Gershgorin's circle theorem \cite{gershgorin1931uber}]
\label{fact:gershgorin}
Let $\bv A \in \R^{n\times n}$ with entries $\bv A_{ij}$. For $i \in [n]$, let $\bv R_i$ be the sum of absolute values of non-diagonal entries in the $i$\textsuperscript{th} row. Let $D(\bv A_{ii}, \bv R_i)$ be the closed disc centered at $\bv A_{ii}$ with radius $\bv R_i$. Then every eigenvalue of $\bv A$ lies within one of the discs $D(\bv A_{ii}, \bv R_i)$. 
\end{fact}
From Fact \ref{fact:gershgorin},  all eigenvalues of $\bv{W}^T \bv {W}$ lie in the range $\left[1 - \frac{1}{ c_2^2n} ,  1+\frac{1}{c_2^2 n}\right]$. Thus, $\bv{W}^T\bv{W}$ is a full rank matrix and so, $\bv{W}$ is a rank $p+1$ matrix. Let $\bv{Q} \bv{S} \bv{R}^T$ be the SVD of $\bv{W}$, where $\bv Q \in \R^{n \times p+1}$ and $\bv{R} \in \R^{p +1\times p+1}$ are orthonormal matrices and $\bv{S} \in \R^{p+1 \times p1}$ is a diagonal matrix of singular values of $\bv{W}$. Then, all eigenvalues of $(\bv{W}^T \bv {W})^{-1/2}$ or all diagonal entries of $\bv{S}^{-1}$ lie in $\left[1 - \frac{c_5}{n} ,  1+\frac{c_5}{n}\right]$ for some constant $c_5$.
    
    Now $\| \bv{A}\bv{Q}\|_F^2=\| \bv{A}\bv{W} \bv{R}\bv{S}^{-1}\|_F^2 \geq \left(1 - \frac{c_6}{ n}\right)\|\bv{A}\bv{W} \bv{R} \|_F^2 \geq \left(1 - \frac{c_6}{ n}\right)\|\bv{A}\bv{W} \|_F^2$ (for some constant $c_6$) where the last step follows from the fact that $\|\bv{A}\bv{W} \bv{R} \|_F^2=\|\bv{A}\bv{W} \|_F^2$ since $\bv{R}$ is an orthonormal square matrix. Then we have
    \begin{align}
    \label{eq:AQF2}
        \norm{\bv A \bv{Q}}_F^2 \ge (1 - \frac{c_6}{n})\|\bv{A}\bv{W} \|_F^2 \ge \sum_{i=1}^{p+1} \sigma_i^2(\bv A) + c_7 \alpha^3 \max(n,\norm{\bv A}_1)^2, 
    \end{align}
    for some constant $c_7$, where the last inequality uses~\eqref{eq:AWF2} and the fact that $(1-\frac{c_6}{n})c_4\alpha^3\max(n,\norm{\bv A}_1)^2 -\frac{c_5\sum_{i=1}^{p+1}\sigma^2_i{\bv{A}}}{n}=\Omega(\alpha^3\max(n,\norm{\bv A}_1)^2)$. However, by the optimality of the SVD for Frobenius norm low-rank approximation (based on the eigenvvalue min-max theorem of Fact \ref{fact:minimax}), for any matrix $\bv X \in \R^{n \times (p+1)}$ with orthonormal columns, $ \|\Ab\bv X\|_F^2 \leq \|\Ab\bv V_{p+1}\|_F^2 = \sum_{i=1}^{p+1} \sigma_i^2(\bv A)$. This contradicts \eqref{eq:AQF2} which was derived using the assumption $\sigma_{\min}(\bv V_p^T\Zb_p) \leq \frac{1}{c_2n}$. 
    \end{proof}

We now prove the main result of this section. We will first show that there will be a large (in terms of $O((\alpha\cdot \max(n,\|\bv{A} \|_1))^c)$) singular value gap for some $p \leq O(\frac{1}{\alpha})$ and then use Lemmas~\ref{lem:fro_broken_bound} and~\ref{lem:spectral bound zu} to prove a gap dependent bound similar to Theorem 13 of~\cite{musco2015randomized}.

\begin{theorem}[High accuracy spectral norm approximation]\label{thm:alpha-promise}
    Let $\Ab\in\R^\n$ be a bounded entry symmetric matrix with largest singular value $\sigma_1(\Ab)$, such that $\sigma_1(\bv A) \geq \alpha \cdot \max(n,\|\Ab\|_1)$ for some $\alpha \in (0,1)$. Then there exists an $\widetilde O\left(\frac{n^2\log (n/\epsilon)}{\alpha^{29}}\right)$ time deterministic algorithm that computes $\widetilde{\zb} \in \R^n$, such that $\|\Ab \widetilde \zb\|_2 \geq (1-\epsilon)\sigma_1(\Ab)$.
\end{theorem}

\begin{proof}
    We will first show that there exists some $p \in \frac{2}{\alpha}$ such that there is a large enough gap between the squared singular values $\sigma^2_p(\Ab)$ and $\sigma^2_{p+1}(\Ab)$. For $k=\frac{2}{\alpha}$ we have $\sigma_k(\Ab) \leq \frac{\alpha}{2}\cdot(\max(n,\|\Ab\|_1))$ and squaring both sides we have $\sigma_k^2(\Ab) \leq \frac{\alpha^2}{4}\cdot(\max(n,\|\Ab\|_1))^2$. Also since $\sigma_1(\Ab) \geq \alpha^2 \cdot \max(n, \|\Ab\|_1)^2$, we have $\sigma_1^2(\Ab) - \sigma_k^2(\Ab) \geq \frac{3}{4}\alpha^2\cdot(\max(n,\|\Ab\|_1))^2$.
    Thus there exists $p\in\left[\frac{2}{\alpha}\right]$ such that
    \begin{align}
    \label{eq:consecutive-sq-sing-val-bound}
        \sigma_p^2(\Ab) - \sigma_{p+1}^2(\Ab) \geq \frac{\frac{3}{4}\alpha^2\cdot(\max(n,\|\Ab\|_1))^2}{\frac{2}{\alpha}} = \frac{3}{8}\alpha^3\cdot(\max(n,\|\Ab\|_1))^2.
    \end{align}

    Let $\Zb$ be the output of step 2 of Algorithm \ref{alg:high_acc_spec_norm}. Let $\Zb_p \in \R^{n\times p}$ be the matrix with the first $p$ columns of $\Zb$. Then performing Block Krylov iterations (i.e. step 3 of Algorithm~\ref{alg:high_acc_spec_norm}) with $\Zb$ as a starting block can only yield a larger approximation for $\sigma_1(\bv{A})$ as compared to doing Block Krylov iterations with $\Zb_p$ as a starting block. Thus it suffices to show that Block Krylov iterations (step 3 of Algorithm~\ref{alg:high_acc_spec_norm}) with starting block $\Zb_p$ produces some matrix $\widetilde{\bv{Z}}_p \in \R^{n \times p}$ with orthonormal columns such that $\sqrt{\widetilde{\bv{z}}_1^T\bv{A}\bv{A}^T\widetilde{\bv{z}}_1}\geq (1-\epsilon)\sigma_1(\Ab)$ where $\widetilde{\bv{z}}_1$ is the first columns of $\widetilde{\bv{Z}}_p$.

    To show this, we will be using Theorem 13 of~\cite{musco2015randomized} which bounds the number of iterations of randomized Block Krylov iterations in terms of the singular value gap. To apply this in our setting, we need to show that $\bv{Z}_p$ is a good enough starting block. Specifically, let $f(\bv{A})$ be the amplifying polynomial used in the proof of Theorem 13 of~\cite{musco2015randomized}. Let $\bv{Y}$ be an orthonormal basis for the span of $f(\bv{A})\bv{Z}_p$. Then, following the proof in~\cite{musco2015randomized}, to prove the gap-dependent convergence bound, we just need to show that $$\|f(\bv{A})-\bv{Y}\bv{Y}^Tf(\bv{A}) \|_F^2 \leq cn^2\| f(\bv{A})-(f(\bv{A}))_p\|_F^2 $$ where $(f(\bv{A}))_p$ is the best rank $p$ approximation to $f(\bv{A})$ and $c$ is a constant. To prove this, we will be proceeding in a similar manner as to the proof of Lemma 4 of~\cite{musco2015randomized}. Using Lemma~\ref{lem:fro_broken_bound}, we have:
    $$\|f(\bv{A})-\bv{Y}\bv{Y}^Tf(\bv{A}) \|_F^2 \leq \| f(\bv{A})-(f(\bv{A}))_p\|_F^2+ \|(f(\bv{A})-(f(\bv{A}))_p)\bv{Z}_p(\bv{U}_p^T\bv{Z}_p)^{+} \|_F^2$$ where $\bv{U}_p \in \R^{n \times p}$ contains the top $p$ singular vectors of $\bv{A}$ as its columns (note that $\bv{A}$ is symmetric so it has same left and right singular vectors).   
Thus, it suffices to show that $$\|(f(\bv{A})-(f(\bv{A}))_p)\bv{Z}_p(\bv{U}_p^T\bv{Z}_p)^{+} \|_F^2 \leq cn^2\| f(\bv{A})-(f(\bv{A}))_p\|_F^2$$ for some constant $c$. Using the spectral submultiplicativity property, we have 
\begin{align*}
    \|(f(\bv{A})-(f(\bv{A}))_p)\bv{Z}_p(\bv{U}_p^T\bv{Z}_p)^{+} \|_F^2 &\leq \|(f(\bv{A})-(f(\bv{A}))_p) \|_F^2 \| \bv{Z}_p(\bv{U}_p^T\bv{Z}_p)^{+}\|_2^2 \\
    &\leq (f(\bv{A})-(f(\bv{A}))_p) \|_F^2\| (\bv{U}_p^T\bv{Z}_p)^{+}\|_2^2,
\end{align*}
where the second step uses the fact that $\|\bv{Z}_p \|_2^2=1$. From Lemma~\ref{lem:spectral bound zu}, we have $\| (\bv{U}_p^T\bv{Z}_p)^{+}\|_2 \leq c_1 n$ for some constant $c_1$. Thus, we have shown that doing Block Krylov iterations with starting block $\bv{Z}_p$ gives us the same guarantees as Theorem 13 of~\cite{musco2015randomized} for gap-dependent convergence of randomized Block Krylov. 
    
    Specifically, we can apply the per-vector guarantee (equation 3 of~\cite{musco2015randomized}) i.e. if $\widetilde{\bv{z}}_1$ is the first column of the matrix $\widetilde{\bv{Z}}_p$ after doing Block Krylov iterations, and $\bv{v}_1$ is the singular vector of $\bv{A}$ corresponding to its largest singular value, we have $|\bv{v}_1^T\bv{A}\bv{A}^T\bv{v}_1-\widetilde{\bv{z}}_1\bv{A}\bv{A}^T\widetilde{\bv{z}}_1| \leq \epsilon \sigma_{p+1}^2(\bv{A}) \leq \epsilon \sigma_{1}^2(\bv{A})$. This implies that $\sqrt{\widetilde{\bv{z}}_1\bv{A}\bv{A}^T\widetilde{\bv{z}}_1} \geq (1-\epsilon)\sigma_1(\bv{A})$ (as $\bv{v}_1^T\bv{A}\bv{A}^T\bv{v}_1=\sigma^2_1(\bv{A})$).

    \medskip
    
    \noindent \textbf{Runtime Analysis.}
    The time to compute $\Zb$ from $\Atilde$ in step 2 of Algorithm~\ref{alg:high_acc_spec_norm} with $k=\frac{2}{\alpha}$ and error parameter $ \tilde \epsilon = c\alpha^4$ (for some small $c<1$) is $\widetilde O\left(\frac{n^2}{\alpha^{29}}\right)$, since we only need to estimate the top $\frac{2}{\alpha}$ singular values, and estimating each singular value using power iterations in Algorithm~\ref{alg:pow_it2} takes time $\widetilde O\left(\frac{kn^2\log (n/ \tilde \epsilon)}{\tilde \epsilon^7}\right)$ as described in the analysis of Theorem~\ref{thm:approxSVD}.
    
    The number of iterations of the Block Krylov with starting block $\bv{Z}_p$ is given by $O\left(\log(n/\epsilon)/ \sqrt{\min(1,\frac{\sigma_1(\bv{A})}{\sigma_{p+1}(\bv{A})}-1)}\right)$ as specified in Theorem 13 of~\cite{musco2015randomized}. Thus, we first need to bound $\frac{\sigma_1(\bv{A})}{\sigma_{p+1}(\bv{A})}$. We have:
    \begin{align*}
        \frac{\sigma_1^2(\Ab)}{\sigma^2_{p+1}(\Ab)} \geq \frac{\sigma^2_p(\Ab)}{\sigma^2_{p+1}(\Ab)} \geq \frac{\sigma^2_p(\Ab)}{\sigma^2_p(\Ab)-\frac{3}{8}\alpha^3 \cdot (\max(n,\|\Ab\|_1))^2} \geq \frac{1}{1-\frac{3}{8}\alpha^{3}} \geq 1+\frac{3}{8}\alpha^{3},
    \end{align*}
    where the second inequality uses \eqref{eq:consecutive-sq-sing-val-bound} and for the third inequality we use the fact that $\sigma_p(\bv{A}) \leq \max(n,\|\bv{A} \|_1)$. Taking square root we have $\frac{\sigma_1(\Ab)}{\sigma_{p+1}(\Ab)} \geq 1+c\alpha^3$, where $c$ is a large enough constant. Thus, the number of iterations of Block Krylov in step 3 of Algorithm~\ref{alg:high_acc_spec_norm} is $O\left(\frac{\log (n/\epsilon)}{\alpha^{3/2}}\right)$. From Theorem 7 of \cite{musco2015randomized}, for $q$ iterations, Block Krylov takes time $O(n^2kq + nk^2q^2 + k^3q^3)$  where $k=\frac{2}{\alpha}$. Thus, the total time taken by step 3 of Algorithm~\ref{alg:high_acc_spec_norm} is $ \widetilde O\left(\frac{n^2\log (n/\epsilon)}{\alpha^{5/2}} + \frac{n\log^2(n/\epsilon)}{\alpha^{5}} + \frac{\log^3(n/\epsilon)}{\alpha^{15/2}}\right)$. Since $\widetilde{O}\left(\frac{n^2\log (n/\epsilon)}{\alpha^{29}}\right)$ asymptotically dominates each term of $ \widetilde O\left(\frac{n^2\log (n/\epsilon)}{\alpha^{5/2}} + \frac{n\log^2(n/\epsilon)}{\alpha^{5}} + \frac{\log^3(n/\epsilon)}{\alpha^{15/2}}\right)$, the total time taken by the Algorithm~\ref{alg:high_acc_spec_norm} is $\widetilde O\left(\frac{n^2\log (n/\epsilon)}{\alpha^{29}}\right)$.
  \end{proof}

\noindent\textbf{Remark}. Since we can apply the per-vector guarantee (Equation 3 of \cite{musco2015randomized}) for $\widetilde\Zb_p$, i.e., the output of Block Krylov in Algorithm \ref{alg:high_acc_spec_norm}, then for all $p\in\lceil\frac{2}{\alpha}\rceil$ such that the condition in \eqref{eq:consecutive-sq-sing-val-bound} holds and any $i\leq p$, $\lvert\vb_i^T\Ab\Ab^T\vb_i - \widetilde{\zb}_i^T\Ab\Ab^T\widetilde{\zb}_i\rvert \leq \epsilon\sigma_{p+1}^2(\Ab) \leq \epsilon\sigma_{i}^2(\Ab)$, where $\widetilde{\zb}_i$ is the $i$\textsuperscript{th} column of $\widetilde\Zb_p$. This implies for all $i\leq p$ we have $\sqrt{\widetilde{\bv{z}}_i\bv{A}\bv{A}^T\widetilde{\bv{z}}_i} \geq (1-\epsilon)\sigma_i(\bv{A})$ (as $\bv{v}_i^T\bv{A}\bv{A}^T\bv{v}_i=\sigma^2_i({\bv{A}})$). Thus we are able to approximate the top $p$ singular values and corresponding singular vectors of $\Ab$ with high accuracy using Algorithm \ref{alg:high_acc_spec_norm}, where the $p$\textsuperscript{th} singular value of $\Ab$ satisfies the condition of \eqref{eq:consecutive-sq-sing-val-bound}.


\section{Conclusion}

Our work has shown that it is possible to deterministically construct an entrywise sampling matrix $\bv S$ (a \emph{universal sparsifier}) with just $\widetilde O(n/\epsilon^c)$ non-zero entries such that for any bounded entry $\bv A$ with $\norm{\bv A}_\infty \le 1$, $\|\bv A - \bv A \circ \bv S\|_2 \leq \epsilon \cdot \max ( n, \|\bv A\|_1)$. We show how to achieve sparsity $O(n/\epsilon^2)$ when $\bv{A}$ is restricted to be PSD (Theorem \ref{thm:PSD_Quad}) and $\widetilde O(n/\epsilon^4)$ when $\bv{A}$ is general (Theorem \ref{thm: general matrix eps4 bound}), and prove that both these bounds are tight up to logarithmic factors (Theorems \ref{theorem:psd_deterministic_lower} and \ref{thm: lower bound nonadaptive}). Further, our proofs are based on simple reductions, which show that any  $\bv{S}$ that spectrally approximates the all ones matrix to sufficient accuracy (i.e., is a sufficiently good spectral expander) yields a universal sparsifier. 

We apply our universal sparsification bounds to give the first $o(n^\omega)$ time deterministic algorithms for several core linear algebraic problems, including singular value/vector approximation and positive semidefiniteness testing. When $\bv{A}$ is restricted to be PSD and have entries in $\{-1,0,1\}$, we show how to give achieve improved deterministic query complexity of $\widetilde O(n/\epsilon)$ to construct a general spectral approximation $\bv{\widetilde A}$, which may not be sparse (Theorem \ref{th:binary query complexity}). We again show that this bound is tight up a to a logarithmic factor (Theorem \ref{thm: lower bound for PSD matrices})

Our work leaves several open questions: 
\begin{enumerate}
    \item An interesting  question is if $\widetilde O(n/\epsilon)$ sample complexity can be achieved for deterministic spectral approximation of any bounded entry PSD matrix if the sparsity of the output is not restricted, thereby generalizing the upper bound for PSD $\{-1,0,1\}$-matrices proven in Theorem \ref{th:binary query complexity}. This query complexity is known for randomized algorithms based on column sampling \cite{musco2017recursive}, however it is not currently known how to derandomize such results. 
    \item It would also be interesting to close the $\widetilde{O}(1/\epsilon^{2})$ factor gap between our universal sparsification upper bound of $\widetilde O(n/\epsilon^4)$ queries for achieving $\epsilon \cdot \max(n, \|\bv A\|_1)$ spectral approximation error for non-PSD matrices (Theorem \ref{thm: general matrix eps4 bound}) and our $\Omega(n/\epsilon^2)$ query lower bound for general deterministic algorithms that make possibly adaptive queries to $\bv A$ (Theorem \ref{thm: lower bound a1}). By Theorem \ref{thm: lower bound nonadaptive}, our universal sparsification bound is tight up to log factors for algorithms that make \emph{non-adaptive} deterministic queries to $\bv A$. It is unknown if adaptive queries can be used to give improved algorithms. 
    \item Finally, it would be interesting to understand if our deterministic algorithms for spectrum approximation can be improved. For example, can one compute an $\epsilon n$ additive error or a $(1+\epsilon)$ relative error approximation to the top singular value $\norm{\bv A}_2$ for bounded entry $\bv{A}$ and constant $\epsilon$ in $o(n^\omega)$ time deterministically? Are there fundamental barriers that make doing so difficult?
\end{enumerate}

\subsection*{Acknowledgements}

We thank Christopher Musco for helpful conversations about this work. RB, CM and AR was partially supported by an Adobe Research grant, along with NSF Grants 2046235, 1763618, and 1934846. AR was also partially supported by the Manning College of Information and Computer Sciences Dissertation Writing Fellowship. GD was partially supported by NSF AF 1814041, NSF FRG 1760353, and DOE-SC0022085. SS was supported by an NSERC Discovery Grant RGPIN-2018-06398 and a Sloan Research Fellowship. DW was supported by a Simons Investigator Award.

\bibliography{refs}

\newcommand{\etalchar}[1]{$^{#1}$}
\begin{thebibliography}{MMMW21}

\bibitem[ACK{\etalchar{+}}16]{andoni2016sketching}
Alexandr Andoni, Jiecao Chen, Robert Krauthgamer, Bo~Qin, David~P Woodruff, and
  Qin Zhang.
\newblock On sketching quadratic forms.
\newblock In {\em \ITCS{2016}}, 2016.

\bibitem[AHK05]{AHK05}
Sanjeev Arora, Elad Hazan, and Satyen Kale.
\newblock Fast algorithms for approximate semidefinite programming using the
  multiplicative weights update method.
\newblock In {\em \FOCS{2005}}, 2005.

\bibitem[AHK06]{AHK06}
Sanjeev Arora, Elad Hazan, and Satyen Kale.
\newblock A fast random sampling algorithm for sparsifying matrices.
\newblock In {\em \RANDOM{2006}}, 2006.

\bibitem[Aig95]{aigner1995turan}
Martin Aigner.
\newblock Tur{\'a}n's graph theorem.
\newblock {\em The American Mathematical Monthly}, 1995.

\bibitem[AKL13]{AKL13}
Dimitris Achlioptas, Zohar~Shay Karnin, and Edo Liberty.
\newblock Near-optimal entrywise sampling for data matrices.
\newblock In {\em \NIPS{2013}}, 2013.

\bibitem[AKM{\etalchar{+}}17]{Avron:2017vt}
Haim Avron, Michael Kapralov, Cameron Musco, Christopher Musco, Ameya
  Velingker, and Amir Zandieh.
\newblock Random {F}ourier features for kernel ridge regression:
  {A}pproximation bounds and statistical guarantees.
\newblock In {\em \ICML{2017}}, 2017.

\bibitem[Alo86]{alon1986eigenvalues}
Noga Alon.
\newblock Eigenvalues and expanders.
\newblock {\em Combinatorica}, 1986.

\bibitem[Alo21]{alon2021explicit}
Noga Alon.
\newblock Explicit expanders of every degree and size.
\newblock {\em Combinatorica}, 2021.

\bibitem[AM07]{AM07}
Dimitris Achlioptas and Frank McSherry.
\newblock Fast computation of low rank matrix approximations.
\newblock In {\em \STOC{2007}}, 2007.

\bibitem[AW21]{Alman:2021uk}
Josh Alman and Virginia~Vassilevska Williams.
\newblock A refined laser method and faster matrix multiplication.
\newblock In {\em \SODA{2021}}, 2021.

\bibitem[AZL16]{Allen-Zhu:2016vf}
Zeyuan Allen-Zhu and Yuanzhi Li.
\newblock {LazySVD: Even faster SVD} decomposition yet without agonizing pain.
\newblock {\em \NIPS{2016}}, 2016.

\bibitem[BCJ20]{Bakshi:2020uz}
Ainesh Bakshi, Nadiia Chepurko, and Rajesh Jayaram.
\newblock Testing positive semi-definiteness via random submatrices.
\newblock In {\em \FOCS{2020}}, 2020.

\bibitem[BDD{\etalchar{+}}23]{Bhattacharjee:2021wl}
Rajarshi Bhattacharjee, Gregory Dexter, Petros Drineas, Cameron Musco, and
  Archan Ray.
\newblock Sublinear time eigenvalue approximation via random sampling.
\newblock {\em \ICALP{2023}}, 2023.

\bibitem[BDMI14]{boutsidis2014near}
Christos Boutsidis, Petros Drineas, and Malik Magdon-Ismail.
\newblock Near-optimal column-based matrix reconstruction.
\newblock {\em SIAM Journal on Computing}, 2014.

\bibitem[BGPW13]{BGPW13}
Mark Braverman, Ankit Garg, Denis Pankratov, and Omri Weinstein.
\newblock Information lower bounds via self-reducibility.
\newblock In {\em \CSR{2013}}, 2013.

\bibitem[Bha13]{bhatia2013matrix}
Rajendra Bhatia.
\newblock {\em Matrix analysis}.
\newblock Springer Science \& Business Media, 2013.

\bibitem[BJKS04]{BJKS04}
Ziv Bar{-}Yossef, T.~S. Jayram, Ravi Kumar, and D.~Sivakumar.
\newblock An information statistics approach to data stream and communication
  complexity.
\newblock {\em Journal of Computer and System Sciences}, 2004.

\bibitem[BKKS21]{Braverman:2021wj}
Vladimir Braverman, Robert Krauthgamer, Aditya~R Krishnan, and Shay Sapir.
\newblock Near-optimal entrywise sampling of numerically sparse matrices.
\newblock In {\em \COLT{2021}}, 2021.

\bibitem[BKM22]{braverman2022sublinear}
Vladimir Braverman, Aditya Krishnan, and Christopher Musco.
\newblock Sublinear time spectral density estimation.
\newblock In {\em \STOC{2022}}, 2022.

\bibitem[BSST13]{Batson:2013va}
Joshua Batson, Daniel~A Spielman, Nikhil Srivastava, and Shang-Hua Teng.
\newblock Spectral sparsification of graphs: theory and algorithms.
\newblock {\em Communications of the ACM}, 2013.

\bibitem[CG02]{chung2002sparse}
Fan Chung and Ronald Graham.
\newblock Sparse quasi-random graphs.
\newblock {\em Combinatorica}, 2002.

\bibitem[CGL{\etalchar{+}}20]{Chuzhoy:2020wq}
Julia Chuzhoy, Yu~Gao, Jason Li, Danupon Nanongkai, Richard Peng, and
  Thatchaphol Saranurak.
\newblock A deterministic algorithm for balanced cut with applications to
  dynamic connectivity, flows, and beyond.
\newblock In {\em \FOCS{2020}}, 2020.

\bibitem[Cha11]{chatelin2011spectral}
Fran{\c{c}}oise Chatelin.
\newblock {\em Spectral approximation of linear operators}.
\newblock SIAM, 2011.

\bibitem[Coh16]{Cohen:2016td}
Michael~B Cohen.
\newblock Ramanujan graphs in polynomial time.
\newblock In {\em \FOCS{2016}}, 2016.

\bibitem[Cov99]{cover1999elements}
Thomas~M Cover.
\newblock {\em Elements of information theory}.
\newblock John Wiley \& Sons, 1999.

\bibitem[CR12a]{CR12}
Emmanuel~J. Cand{\`{e}}s and Benjamin Recht.
\newblock Exact matrix completion via convex optimization.
\newblock {\em Communications of the {ACM}}, 2012.

\bibitem[CR12b]{CHR12}
Amit Chakrabarti and Oded Regev.
\newblock An optimal lower bound on the communication complexity of
  gap-{H}amming-distance.
\newblock {\em {SIAM} Journal on Computing}, 2012.

\bibitem[CT10]{CT10}
Emmanuel~J. Cand{\`{e}}s and Terence Tao.
\newblock The power of convex relaxation: near-optimal matrix completion.
\newblock {\em {IEEE} Transations on Information Theory}, 2010.

\bibitem[CW17]{clarkson2017low}
Kenneth~L Clarkson and David~P Woodruff.
\newblock Low-rank {PSD} approximation in input-sparsity time.
\newblock In {\em \SODA{2017}}, 2017.

\bibitem[d'A11]{d2011subsampling}
Alexandre d'Aspremont.
\newblock Subsampling algorithms for semidefinite programming.
\newblock {\em Stochastic Systems}, 2011.

\bibitem[DZ11]{DZ11}
Petros Drineas and Anastasios Zouzias.
\newblock A note on element-wise matrix sparsification via a matrix-valued
  {B}ernstein inequality.
\newblock {\em Information Processing Letters}, 2011.

\bibitem[Ger31]{gershgorin1931uber}
Semyon~Aranovich Gershgorin.
\newblock Uber die abgrenzung der eigenwerte einer matrix.
\newblock {\em Izvestiya Rossiyskoy akademii nauk. Seriya matematicheskaya},
  1931.

\bibitem[Kun15]{k15}
Abhisek Kundu.
\newblock {\em Element-wise matrix sparsification and reconstruction}.
\newblock PhD thesis, Rensselaer Polytechnic Institute, {USA}, 2015.

\bibitem[KW92]{Kuczynski:1992va}
Jacek Kuczy{\'n}ski and Henryk Wo{\'z}niakowski.
\newblock Estimating the largest eigenvalue by the power and {L}anczos
  algorithms with a random start.
\newblock {\em SIAM Journal on Matrix Analysis and Applications}, 1992.

\bibitem[LPS88]{lubotzky1988ramanujan}
Alexander Lubotzky, Ralph Phillips, and Peter Sarnak.
\newblock Ramanujan graphs.
\newblock {\em Combinatorica}, 1988.

\bibitem[LSY16]{lin2016approximating}
Lin Lin, Yousef Saad, and Chao Yang.
\newblock Approximating spectral densities of large matrices.
\newblock {\em SIAM Review}, 2016.

\bibitem[Mar73]{margulis1973explicit}
Grigorii~Aleksandrovich Margulis.
\newblock Explicit constructions of concentrators.
\newblock {\em Problemy Peredachi Informatsii}, 1973.

\bibitem[MM15]{musco2015randomized}
Cameron Musco and Christopher Musco.
\newblock Randomized block {K}rylov methods for stronger and faster approximate
  singular value decomposition.
\newblock {\em \NIPS{2015}}, 2015.

\bibitem[MM17]{musco2017recursive}
Cameron Musco and Christopher Musco.
\newblock Recursive sampling for the {N}ystr\"{o}m method.
\newblock {\em \NIPS{2017}}, 2017.

\bibitem[MMMW21]{meyer2021hutch++}
Raphael~A Meyer, Cameron Musco, Christopher Musco, and David~P Woodruff.
\newblock Hutch++: Optimal stochastic trace estimation.
\newblock In {\em Symposium on Simplicity in Algorithms (SOSA)}. SIAM, 2021.

\bibitem[Mor94]{Morgenstern:1994vf}
Moshe Morgenstern.
\newblock Existence and explicit constructions of $q+1$ regular {R}amanujan
  graphs for every prime power $q$.
\newblock {\em Journal of Combinatorial Theory, Series B}, 1994.

\bibitem[MW17]{MW17}
Cameron Musco and David~P. Woodruff.
\newblock Sublinear time low-rank approximation of positive semidefinite
  matrices.
\newblock In {\em \FOCS{2017}}, 2017.

\bibitem[NSW22]{nsw22}
Deanna Needell, William Swartworth, and David~P. Woodruff.
\newblock Testing positive semidefiniteness using linear measurements.
\newblock {\em \FOCS{2022}}, 2022.

\bibitem[Pip87]{pippenger1987sorting}
Nicholas Pippenger.
\newblock Sorting and selecting in rounds.
\newblock {\em SIAM Journal on Computing}, 1987.

\bibitem[Rou16]{roughgarden2016communication}
Tim Roughgarden.
\newblock Communication complexity (for algorithm designers).
\newblock {\em Foundations and Trends in Theoretical Computer Science}, 2016.

\bibitem[Saa11]{Saad:2011wv}
Yousef Saad.
\newblock {\em Numerical methods for large eigenvalue problems: {R}evised
  edition}.
\newblock SIAM, 2011.

\bibitem[SKO21]{schafer2021sparse}
Florian Sch\"{a}fer, Matthias Katzfuss, and Houman Owhadi.
\newblock Sparse {C}holesky factorization by {K}ullback--{L}eibler
  minimization.
\newblock {\em SIAM Journal on Scientific Computing}, 2021.

\bibitem[SS08]{Spielman:2008wm}
Daniel~A Spielman and Nikhil Srivastava.
\newblock Graph sparsification by effective resistances.
\newblock In {\em \STOC{2008}}, 2008.

\bibitem[ST04]{ST04}
Daniel~A. Spielman and Shang{-}Hua Teng.
\newblock Nearly-linear time algorithms for graph partitioning, graph
  sparsification, and solving linear systems.
\newblock In {\em \STOC{2004}}, 2004.

\bibitem[SWXZ22]{srinivas2022memory}
Vaidehi Srinivas, David~P Woodruff, Ziyu Xu, and Samson Zhou.
\newblock Memory bounds for the experts problem.
\newblock {\em \STOC{2022}}, 2022.

\bibitem[Tur41]{turaan1941extremal}
P{\'a}l Tur{\'a}n.
\newblock On an extremal problem in graph theory.
\newblock {\em Mat. Fiz. Lapok}, 1941.

\bibitem[TYUC17]{tropp2017randomized}
JA~Tropp, A~Yurtsever, M~Udell, and V~Cevher.
\newblock Randomized single-view algorithms for low-rank matrix approximation,
  2017.

\bibitem[Ver18]{vershynin2018high}
Roman Vershynin.
\newblock {\em High-dimensional probability: {A}n introduction with
  applications in data science}.
\newblock Cambridge University Press, 2018.

\bibitem[Wey12]{weyl1912asymptotic}
Hermann Weyl.
\newblock The asymptotic distribution law of the eigenvalues of linear partial
  differential equations (with an application to the theory of cavity
  radiation).
\newblock {\em Mathematical Annals}, 1912.

\bibitem[WLZ16]{wang2016spsd}
Shusen Wang, Luo Luo, and Zhihua Zhang.
\newblock {SPSD} matrix approximation vis column selection: Theories,
  algorithms, and extensions.
\newblock {\em The Journal of Machine Learning Research}, 2016.

\bibitem[Woo14]{woodruff2014sketching}
David~P. Woodruff.
\newblock Sketching as a tool for numerical linear algebra.
\newblock {\em Foundations and Trends in Theoretical Computer Science}, 2014.

\bibitem[WS00]{williams2000using}
Christopher Williams and Matthias Seeger.
\newblock Using the {N}ystr{\"o}m method to speed up kernel machines.
\newblock {\em \NIPS{2000}}, 2000.

\bibitem[WS23]{Woodruff:2023ul}
David Woodruff and William Swartworth.
\newblock Optimal eigenvalue approximation via sketching.
\newblock In {\em \STOC{2023}}, 2023.

\bibitem[WWAF06]{weisse2006kernel}
Alexander Weisse, Gerhard Wellein, Andreas Alvermann, and Holger Fehske.
\newblock The kernel polynomial method.
\newblock {\em Reviews of modern physics}, 2006.

\bibitem[WZ93]{wigderson1993expanders}
Avi Wigderson and David Zuckerman.
\newblock Expanders that beat the eigenvalue bound: {E}xplicit construction and
  applications.
\newblock In {\em \STOC{1993}}, 1993.

\bibitem[XG10]{xia2010robust}
Jianlin Xia and Ming Gu.
\newblock Robust approximate {C}holesky factorization of rank-structured
  symmetric positive definite matrices.
\newblock {\em SIAM Journal on Matrix Analysis and Applications}, 2010.

\end{thebibliography}

\end{document}